\pdfoutput=1
\documentclass[11pt]{article}

\usepackage{graphicx}
\usepackage{amssymb,xfrac}
\usepackage{amsthm,amsmath,dsfont}
\usepackage{thmtools}
\allowdisplaybreaks

\usepackage{lineno}

\usepackage[colorlinks=true, citecolor=blue]{hyperref}
\usepackage[algoruled,titlenotnumbered]{algorithm2e}
\usepackage{algorithmic}
\usepackage{wrapfig}
\usepackage[sort&compress,numbers]{natbib}
\usepackage{comment}
\usepackage{multirow}
\usepackage{amsmath,bm}
\usepackage{graphicx}
\usepackage[lofdepth,lotdepth]{subfig}

\usepackage{float}
\usepackage{subfig}
\usepackage{wrapfig}
\usepackage{framed,xcolor}
\usepackage{newfloat}
\usepackage[xcolor,framemethod=TikZ]{mdframed}
\usepackage{mathpazo}
\usepackage{amsthm}
\usepackage{thmtools}
\usepackage{stmaryrd}
\usepackage{mathtools}

\usepackage[colorinlistoftodos]{todonotes}

\newcommand{\RR}{\mathbb{R}}

\newcommand{\EE}{\mathbb{E}}
\newcommand{\EER}{\mathbb{E}^{\boldsymbol{\rho}(t)}}
\newcommand{\PP}{\mathbb{P}}
\newcommand{\GX}{\mathcal{G}}
\newcommand{\AX}{\mathcal{A}}
\newcommand{\AXS}{\mathcal{A}^{\mathrm{src}}}
\newcommand{\NX}{\mathcal{N}}
\newcommand{\dd}{\mathrm{d}}
\newcommand{\IN}{\mathrm{in}}
\newcommand{\OUT}{\mathrm{out}}

\newcommand{\MM}{\mathcal{M}}

\newtheoremstyle{theoremdd}
{\topsep}
{\topsep}
{\itshape}
{0pt}
{\fontfamily{cmss}\selectfont\bfseries}
{.}
{ }
{\thmname{#1}\thmnumber{ #2}\thmnote{ (#3)}}

\theoremstyle{theoremdd}
\newtheorem{theorem}{Theorem}[section]
\newtheorem{lemma}{Lemma}

\newtheorem{corollary}{Corollary}
\newtheorem{definition}{Definition}
\newtheorem{remark}{Remark}

\renewenvironment{proof}{\noindent {\bfseries \fontfamily{cmss}\selectfont Proof.}}{\qed}

\newcounter{daggerfootnote}

\usepackage{array}
\makeatletter
\newcommand{\thickhline}{%
    \noalign {\ifnum 0=`}\fi \hrule height 1.5pt
    \futurelet \reserved@a \@xhline
}
\newcolumntype{"}{@{\hskip\tabcolsep\vrule width 1pt\hskip\tabcolsep}}
\makeatother

\mdfdefinestyle{myStyle}{roundcorner=5pt,backgroundcolor=brown!20,linecolor=red!40!black,linewidth=1.5pt}

\usepackage{titlesec}
\titleformat*{\section}{\fontfamily{cmss}\selectfont\large\bfseries}
\titleformat*{\subsection}{\fontfamily{cmss}\selectfont\normalsize\bfseries}
\titleformat*{\subsubsection}{\fontfamily{cmss}\selectfont\normalsize}

\usepackage[left=1.2 in, right=1.2 in, top=1.00 in, bottom=1.00 in]{geometry}
\graphicspath{{./Figures/}}

\begin{document}
	
		
		
		\title{\fontfamily{cmss}\selectfont Position weighted backpressure intersection control for urban networks
		\\ \vspace{0.1 in} \normalsize Li Li$^\mathrm{a}$, Saif Eddin Jabari$^\mathrm{a,b,}$\footnote{Corresponding author, Email: \url{sej7@nyu.edu}, Website: \url{https://engineering.nyu.edu/faculty/saif-eddin-jabari}}}
		

		

\author{\small $^\mathrm{a}$ Tandon School of Engineering, New York University, Brooklyn, NY 11201, U.S.A.
	\\ \small $^\mathrm{b}$ Division of Engineering, New York University Abu Dhabi, Saadiyat Island, 
	\\ \small P.O. Box 129188, Abu Dhabi, U.A.E.
}

\date{}

\maketitle

{ \fontfamily{cmss}\selectfont\large\bfseries		
\begin{abstract}
{ \normalfont\normalsize
	Decentralized intersection control techniques have received recent attention in the literature as means to overcome scalability issues associated with network-wide intersection control.  Chief among these techniques are backpressure (BP) control algorithms, which were originally developed of for large wireless networks. In addition to being light-weight computationally, they come with guarantees of performance at the network level, specifically in terms of network-wide stability.   The dynamics in backpressure control are represented using networks of point queues and this also applies to all of the applications to traffic control.  As such, BP in traffic fail to capture the spatial distribution of vehicles along the intersection links and, consequently, spill-back dynamics.
	
	This paper derives a position weighted backpressure (PWBP) control policy for network traffic applying continuum modeling principles of traffic dynamics and thus capture the spatial distribution of vehicles along network roads and spill-back dynamics.  PWBP inherits the computational advantages of traditional BP.  To prove stability of PWBP, (i) a Lyapunov functional that captures the spatial distribution of vehicles is developed; (ii) the capacity region of the network is formally defined in the context of macroscopic network traffic; and (iii) it is proved, when exogenous arrival rates are within the capacity region, that PWBP control is network stabilizing.  We conduct comparisons against a real-world adaptive control implementation for an isolated intersection.  Comparisons are also performed against other BP approaches in addition to optimized fixed timing control at the network level.  These experiments demonstrate the superiority of PWBP over the other control policies in terms of capacity region, network-wide delay, congestion propagation speed, recoverability from heavy congestion (outside of the capacity region), and response to incidents. 
	
	\vspace{0.1in}
	
	\noindent \textbf{\fontfamily{cmss}\selectfont Keywords}: Decentralized control, backpressure, max weight, stochastic traffic flow, urban networks, intersection control
	
}
\end{abstract}
}
		
		
	
	
	
\section{Introduction}
\label{S:intro}
Various approaches have been proposed to optimize signal timing for isolated intersections, including mixed-integer linear models, rolling horizon approaches, and store-and-forward models based on model predictive control; see \citep{dujardin2011multiobjective,gartner1983opac,tettamanti2010distributed,mirchandani2001real,you2013coordinated,ma2013coordinated} for examples.  On the one hand, isolated intersection approaches fail to account for spillback from adjacent road segments, which can eventually lead to gridlock throughout a road network \citep{cervero1986unlocking}.
On the other hand, centralized techniques that include coordination between intersection controllers \citep{heung2005coordinated, gettman2007data} are not scalable and difficult to implement in real-world/real-time settings \citep{papageorgiou2003review}.  For example, ACS-Lite \citep{gettman2007data} can handle no more than 12 intersections in real-time.  

Recent articles in traffic control have focused on connected-automated vehicles (CAVs). This is well justified considering the vast opportunities and challenges that CAVs have to offer.  The ability to control both trajectories and signals is one such opportunity that CAVs have to offer; we refer to \citep{li2017recasting,yu2018integrated,feng2018spatiotemporal} and references therein for recent examples.  For more information we refer to \citep{guo2019urban}, a recent review article that covers many aspects of intersection control.

In parallel, decentralized control techniques have been proposed to overcome the scalability issues associated with network control optimization.  Decentralized control techniques can be characterized as adaptive control techniques, which have been shown to have some favorable properties in grid-type networks, e.g., \cite{gayah2014impacts} demonstrated that adaptive control can alleviate gridlock and promote an even distribution of traffic in moderately congested conditions.  Decentralized control techniques expect controllers to be able to measure/estimate local traffic conditions in real-time.  This information includes expected traffic demand at the intersection in the next cycle for heuristic approaches, e.g., \citep{smith1980local, lammer2008self, lammer2010self, smith2011dynamics} or the queue sizes along the intersection arcs in backpressure (BP) based approaches \citep{wongpiromsarn2012distributed, varaiya2013max, xiao2014pressure, le2015decentralized}.  According to \citep{de2011traffic}, control strategies that use traffic conditions along both upstream and downstream arcs are more efficient and reliable than those that utilize upstream traffic conditions only.  BP based approaches are prime examples of techniques that utilize both upstream and downstream information.

{\color{black} The most important feature of BP-based approaches is that, while they are decentralized (applied at the isolated intersection level), they come with provable guarantees of stability at the network level.  We refer to \citep{varaiya2013max} for a discussion of the distinguishing features of BP-based approaches for signal control.} 
They were first independently proposed in \citep{wongpiromsarn2012distributed} and \citep{varaiya2013max} and subsequently refined to incorporate signal timing principles in \citep{le2015decentralized} (specifically related to the sequencing of signal phases). The stability analyses in these original applications of BP to intersection control are based on seminal work in wireless networks \citep{tassiulas1992stability} (see \citep{neely2005dynamic, georgiadis2006resource, neely2010stochastic} for more details).

 {\color{black} Backpressure based techniques have two key characteristics that make them attractive for intersection control.  The first is that they do not require knowledge of demands as discussed in \citep{wongpiromsarn2012distributed,varaiya2013max,le2015decentralized}.  The second is that they naturally decompose by intersection: A distinguishing feature of traffic network control problems is that each movement in the network can only be associated with one intersection.  Therefore, the decision variables assigned to intersection movements only appear in equations pertaining to the node representing that intersection. Consequently, the optimization problems decompose naturally by node.  It is this property that renders BP based approaches scalable to large networks in the case of intersection control.} However, as BP was originally developed for wireless communications networks, the assumptions are not tailored to traffic control problems and in some cases the assumptions are not suitable for traffic networks. Specifically, (i) point queues and (more critically) queues with infinite buffer capacities, (ii) separate queues for different commodities (corresponding to vehicles with different turning desires in traffic) and no interference between commodities.  A consequence of the first assumption is that the models do not account for the \textit{spatial distribution} of the queues, which has great impact on traffic control. For instance, \autoref{F:queues} illustrates three different spatial distributions of vehicles with the \emph{same queue size}.  
\begin{figure}[h!]
	\centering
	
	\resizebox{0.46\textwidth}{!}{%
		\includegraphics{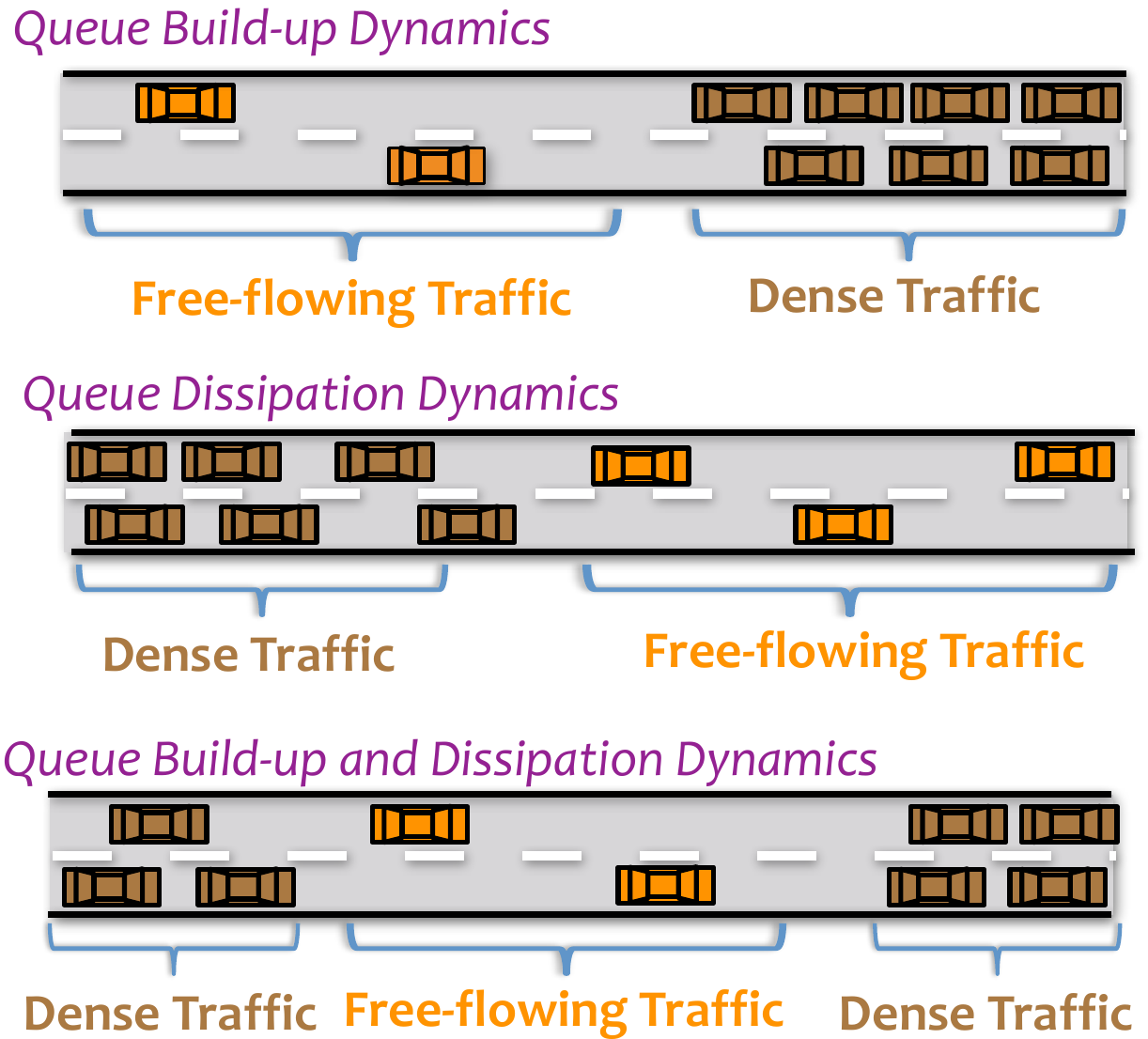} \hspace{0.2in}}
			
	\caption{Three different spatial distributions of queues with same queue size.}
	\label{F:queues}
\end{figure}
Clearly, signal control decisions at the downstream end should be very different for these three cases. A key point here is how vehicle flux out of road segments are affected by the vehicle distribution along the length of the road. While communications networks assume that such ``transmission rates'' will not be influenced by the distribution of packets along the channels, in vehicular traffic the situation is quite different.  Point queueing techniques suffer this same drawback. For example, the flow rate over the course of a short time interval (e.g. 10 sec) at the downstream end of the road segment depicted in \autoref{F:queues} should be very different in the three cases. 
\begin{figure}[h!]
	\centering	
	\resizebox{0.65\textwidth}{!}{%
		\includegraphics{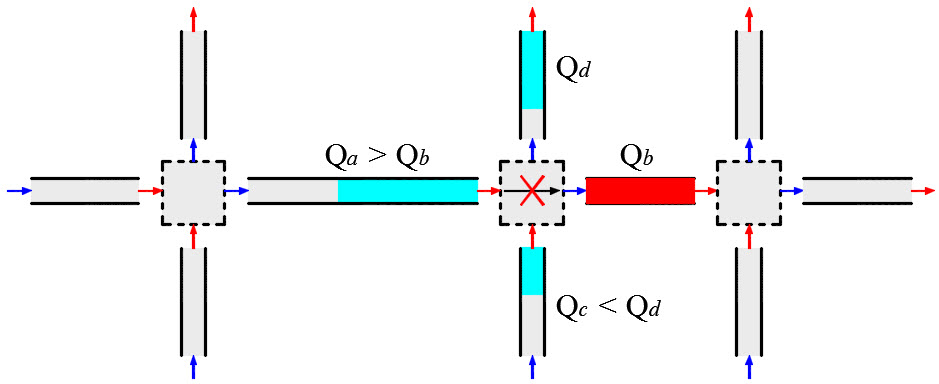} \hspace{0.2in}}
	
	{ \small	(a) }
	\vspace{0.1in}
	
	\resizebox{0.65\textwidth}{!}{%
		\includegraphics{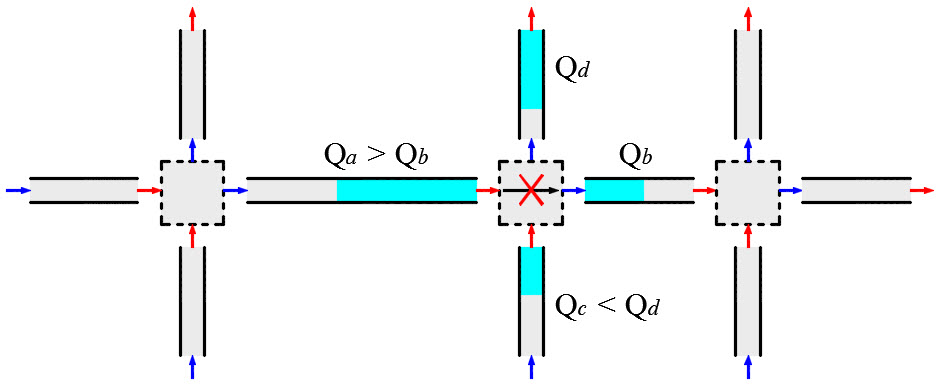} \hspace{0.2in}}
	
	{ \small	(b) }
	\vspace{0.1in}
	
	\resizebox{0.65\textwidth}{!}{%
		\includegraphics{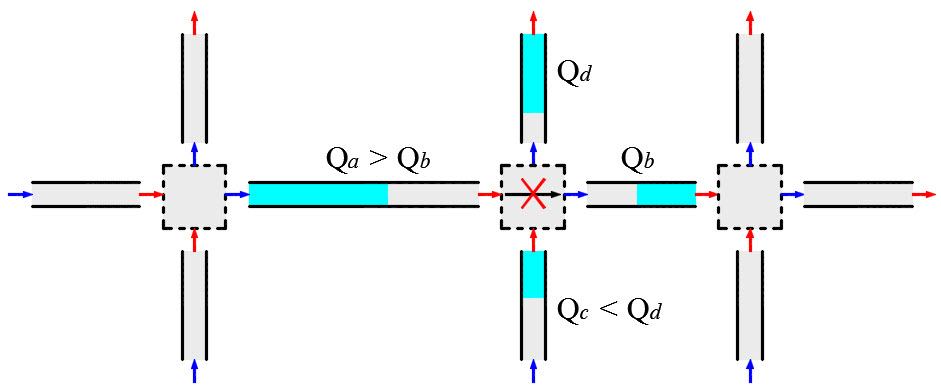}}
	
	{ \small	(c) }
	
	\caption{Three non-work conserving cases (adopted from \citep{gregoire2015capacity} and reproduced)}
	\label{F:NWC}
\end{figure}
A serious consequence of assumption (i) is loss of work conservation, in which flow is prohibited across the intersection despite the availability of (spatial) capacity in the outbound roads. \autoref{F:NWC} shows three cases in which BP control favors the eastbound approach ($Q_a$ to $Q_b$), despite the fact that flow rates will be zero along this approach if given priority.  Recognizing the finite (spatial) capacity issue, \citep{gregoire2015capacity} proposed an improvement, referred to as \textit{capacity aware back pressure} (CABP) control.  However, due to failure to account for the queue's spatial distribution, their approach can only avoid the case illustrated in \autoref{F:NWC}a, but not the two depicted in \autoref{F:NWC}b and \autoref{F:NWC}c (in the former, the downstream queue is concentrated at the ingress of the road segment). Assumption (ii) could be easily violated in traffic networks, e.g., shared lanes. Even when there are no shared lanes, road widening near the exits of intersection inbound roads, a very common geometrical features in urban networks, can create bottlenecks at the lane-branching point.  Different turning movements (commodities) interact at the bottleneck, and one queue may block another if it gets too long as illustrated in \autoref{F:bottleneck}. Work conservation may also be violated here, as traditional BP (and CABP) control would favor the through movement, despite the fact that no through vehicles can actually cross the intersection. This, in fact, serves as one physical mechanism that can lead to the scenario depicted in \autoref{F:NWC}c. Loss of work conservation is a result of zero outflow if the through movement is given priority and the prime culprit is the fact that the spatial distribution of vehicles is not taken into consideration.
\begin{figure}[h!]
	\centering
	\resizebox{.9\textwidth}{!}{%
		\includegraphics{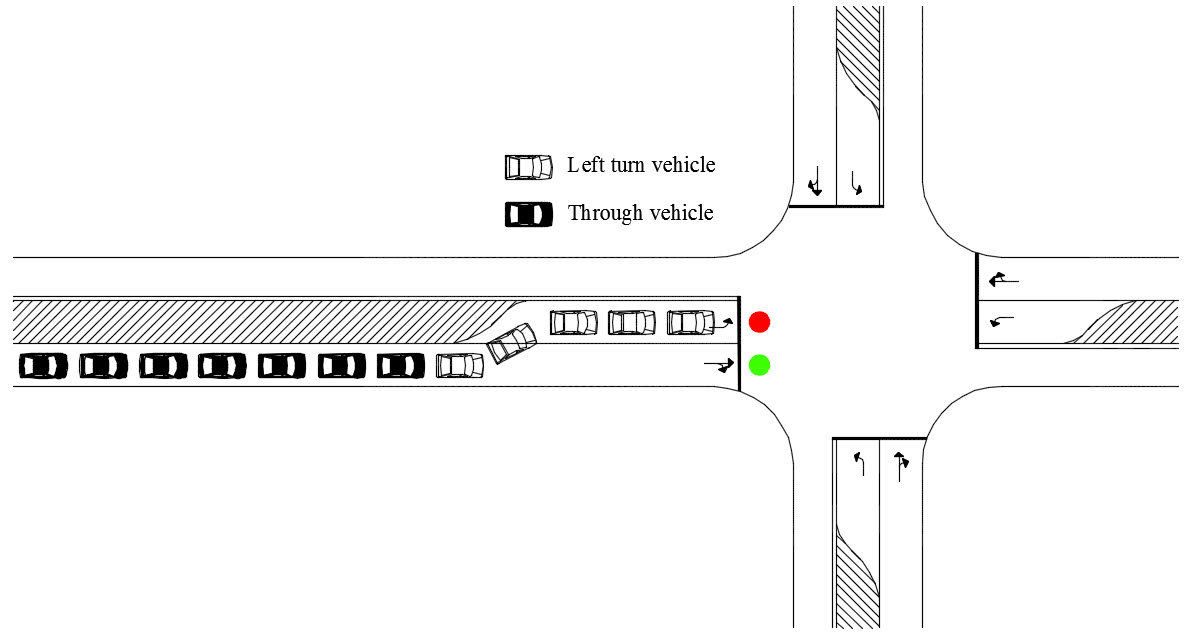}}
	
	\caption{Bottleneck at the lane-branching point.} 
	\label{F:bottleneck}
\end{figure}

This paper proposes a decentralized intersection control technique that is based on  macroscopic traffic theory to overcome the issues described above.  We refer to this approach as \textit{position-weighted backpressure} (PWBP).  PWBP considers the spatial distribution of vehicles along the road, applying higher weights to queues that extend to the ingress of the road, thereby accounting for the possibility of spillback. 
Rates of flow across the intersection depend on both the control (signal status) and vehicle densities profiles (spatial distribution) along the inbound and outbound roads, capturing diminished flow rates at signal phase startups (startup lost times).  That is, we employ a node model for intersection dynamics. We perform comparisons in isolated intersection settings against a real-world implementation of adaptive control, namely SCOOT (Split, Cycle and Offset Optimisation Technique),  as well as in a network setting against fixed intersection control, standard BP, and CABP. We demonstrate superiority of PWBP in terms of capacity region, delay, congestion propagation speed, recoverability from heavy congestion and response to an incident.  The type of control proposed is applied to modern day traffic lights, but it can also be thought of as a prioritization scheme for connected vehicles at network intersections that can guarantee network stability.  In both cases, when accurate measurement of the distribution of vehicles along the roads is not possible, one may employ a light-weight traffic state estimation technique.  We refer to \citep{jabari2013stochastic,seo2017traffic,zheng2018traffic} for recent examples.

The remainder of this paper is organized as follows: \autoref{S:dynamics} describes the traffic dynamics model, macroscopic intersection control, and the proposed PWBP control policy. \autoref{S:stability} rigorously demonstrates the network-wide stability properties of the PWBP approach using Lyapunov drift techniques. A comparison with adaptive control at the isolated intersection level and simulation experiments at the network level are provided in \autoref{S:simulation} and \autoref{S:Conc} concludes the paper.

\section{Problem formulation}
\label{S:dynamics}

\subsection{Notation}
\label{Ss:boundaries}
Consider an urban traffic network represented by the directed graph $\GX = (\NX,\AX)$, where $\NX$ is a set of network nodes, representing intersections and $\AX \subset \NX \times \NX$ is a set of network arcs, representing road segments.  Each element of $\AX$ is in one-to-one correspondence with an ordered pair of elements in $\NX$.  For each node, $n\in \NX$, $\Pi_n$ and $\Sigma_n$ denote, respectively, the set of (predecessor) arcs terminating in $n$ and the set of (successor) arcs emanating from $n$.  We also use $\Pi(a) \subseteq \AX$ to denote the set of predecessor arcs to arc $a \in \AX$.  That is, if $n$ is the ingress node of arc $a$, then $\Pi(a) = \Pi_n$.  Similarly, $\Sigma(a)$ is the set of successor arc to arc $a$.

Fictitious source arcs are appended to the network to represent exogenous network arrivals.  A new junction with indegree zero and outdegree one is created for each exogenous inflow and the new source arc connects this new node to the network boundary node; see \autoref{f_src1}.
\begin{figure}[h!]
	\centering
	\resizebox{0.25\textwidth}{!}{%
		\includegraphics{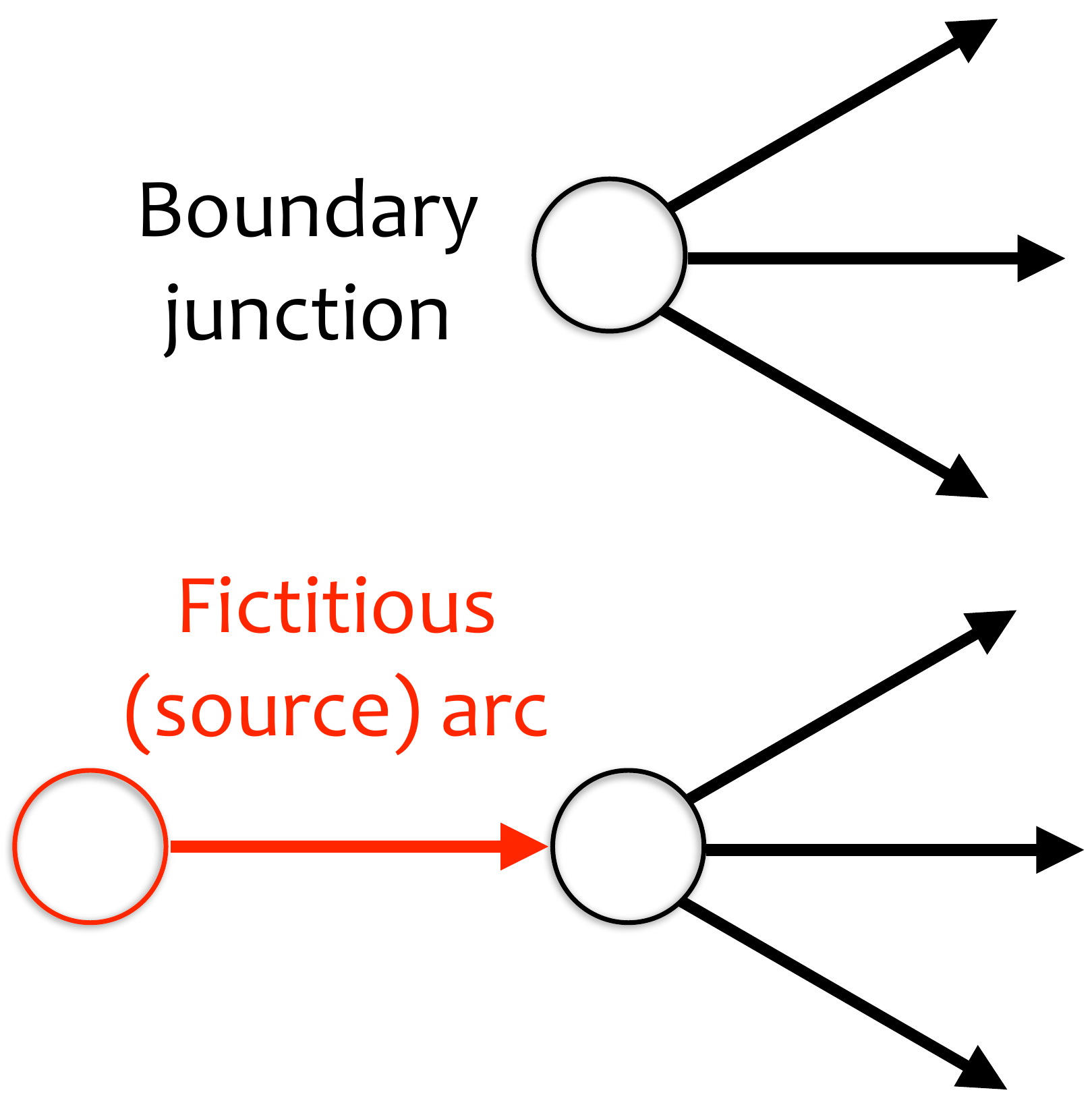}}
	
	\caption{Fictitious boundary source arcs} \label{f_src1}
\end{figure}
When exogenous inflows occur at the interior of the network (i.e., at a junction with non-zero in-degree) representing, for example, a parking ramp/lot, the associated arc can be broken into two arcs with a new node placed at the position of the merge; see \autoref{f_src2}.
\begin{figure}[h!]
	\centering
	\resizebox{0.45\textwidth}{!}{%
		\includegraphics{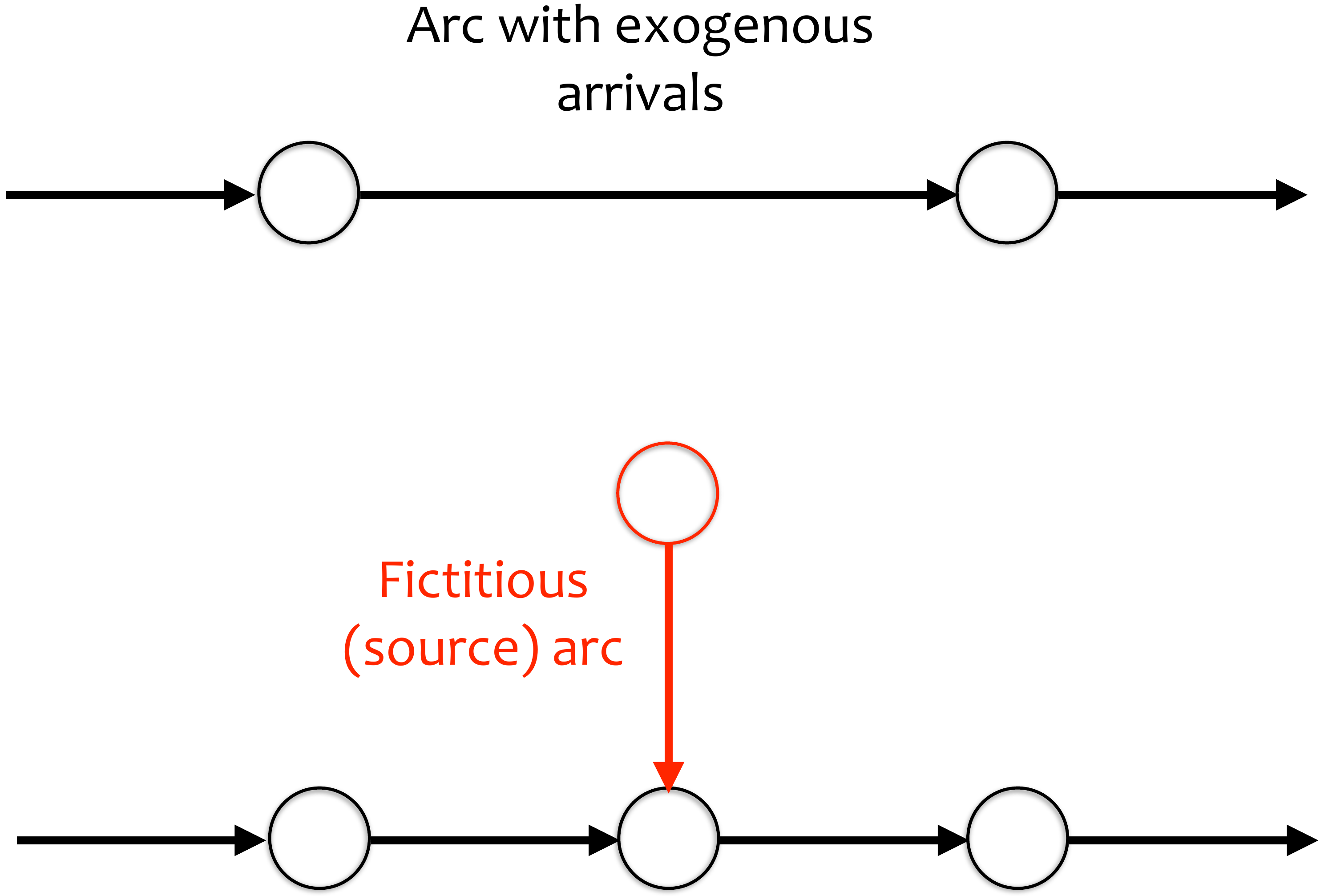}}
	
	\caption{Fictitious interior source arcs} \label{f_src2}
\end{figure}
Source arcs will be assumed to have infinite jam densities (i.e., they serve as fictitious reservoirs), but the flow rates in and out of these arcs are assumed to be finite (i.e., finite capacities).  They shall also be assumed to have zero physical length.  Therefore, the traffic states associated with fictitious source arcs are point queues concentrated at the source node.  We shall denote the set of (fictitious) source arcs by $\AXS \subset \AX$.  Arcs in $\AXS$ serve two purposes: the first (mentioned above) is to model exogenous network inflows.  The second is to capture instabilities in the network: Roads have finite spatial capacities and traffic densities are always finite.  Source arcs with infinite storage capacities are capable of capturing network instabilities.  For example, a signal control policy that results in instabilities is one where congestion propagates to the source arcs and builds up there and exogenous arrivals can no longer be accommodated.

\subsection{Stochastic arc dynamics}
We denote the length of each arc $a \in \AX$ by $l_a$.  With slight notation abuse, the upstream-most position (the entrance position) for each arc $a$ in the network is $x = 0$, while the downstream-most position (the arc exit position) is $x=l_a$ (that these coordinates pertain to arc $a$ only should be understood implicitly).  We consider a multi-commodity framework, where $\rho_a^b(x,t)$ denotes the traffic density at position $x$ along arc $a$ that is destined to outbound arc $b\in \Sigma(a)$ at time instant $t$.  Similarly, $q_a^b(x,t)$ denotes the flow rate at $x$ along $a$ that is destined to $b$ at time $t$.  We define the state of the system at time $t$ as the vector of commodity-specific network traffic densities.  This is denoted as\footnote{We use the `$\cdot$' notation as a function argument to represent the entire curve in the dimension in which it is used.  In other words, $\rho_a^b(\cdot,t)$ denotes the traffic density \textit{curve} along arc $a$ destined to adjacent arc $b$ at time instant $t$.} $\boldsymbol{\rho}(t) \equiv \{\rho_a^b(\cdot,t)\}_{(a,b)\in \MM}$.

On the interiors of network arcs, we have the following conservation equation: for each $a \in \AX$ and $b \in \Sigma(a)$
\begin{equation}
	\frac{\partial \rho_a^b(x,t)}{\partial t} = - \frac{\partial q_a^b(x,t)}{\partial x} ~~ x \in (0,l_a), t \ge 0. \label{E:consLaw}
\end{equation}
In a first order context, one sets $q_a^b(x,t) \equiv \phi_a^b(x,t) \mathcal{Q}_a\big( \rho_a(x,t) \big)$, where $\phi_a^b(x,t)$ is the fraction of vehicles at position $x$ along arc $a$ that is destined to arc $b$ at time $t$ and $\mathcal{Q}_a$ is a (stochastic) stationary flow-density relation pertaining to arc $a$.  In a higher order context, $q_a^b(x,t) = \rho_a^b(x,t) v_a(x,t)$, where $v$ is the macroscopic speed and 
\begin{align}
\frac{\dd v_a(x,t)}{\dd t} = \Big( \frac{\partial}{\partial t} + v_a(x,t) \frac{\partial}{\partial x} \Big) v_a(x,t) 
= A_a^{\mathrm{loc}} \Big( \rho_a(x,t), \mathcal{V}_a \big(\rho_a(x,t)\big), \frac{\partial \rho_a(x,t)}{\partial x}, \frac{\partial \mathcal{V}_a \big(\rho_a(x,t)\big)}{\partial x} \Big),
\end{align}
where $A_a^{\mathrm{loc}}$ are `local' macroscopic acceleration models \citep[Chapter 9]{treiber2013traffic} and $\mathcal{V}_a$ is a stationary stochastic speed-density relation. The stochasticity in $\mathcal{Q}_a$ and $\mathcal{V}_a$ is parametric, that is, they can be described as generalizations of equilibrium fundamental relations that capture heterogeneity in the driving population as described in \citep{jabari2014probabilistic}.  For example, a generalization of Newell's simplified relation \citep{newell2002simplified}:
\begin{equation}
	\mathcal{Q}_a(\rho) = \min \big\{\overline{v}_a \rho, w_a(\rho - \overline{\rho}_a) \big\},
\end{equation}
is one where the parameters $\overline{v}_a$, $w_a$, and $\overline{\rho}_a$, denoting free-flow speed, backward wave speed, and jammed traffic density, respectively, are random variables.  We refer to \citep{jabari2018stochastic,zheng2018traffic} for the properties of the stochastic dynamics that arise as a result of a parametric treatment.

\begin{remark}
	We make no assumptions about the relationship between flux and density. The proposed approach is equally valid in first and second order contexts.  The only assumptions we make are (i) flow conservation, (ii) probabilistic upper bounds on flux and density, and (iii) that arc parameters do not change along the length of the arc.  The last assumption is easy to honor in a general network by splitting arcs with varying parameters into more than one arc.
\end{remark}

\subsection{Boundary dynamics and intersection control}
\label{Ss:junctions}
At the arc boundaries, i.e., for $x \in \{0,l_a\}$, we employ a node model.  Node models represent the coupling between adjacent arcs and are responsible for capturing \textit{queue spillback dynamics}.

\medskip

\textbf{\fontfamily{cmss}\selectfont Notation}.  For each node $n \in \NX$, let $\MM_n$ denote the set of allowed movements between inbound and outbound road segments.  The set $\MM_n$ consists of ordered pairs $(a,b)$ such that $a \in \Pi_n$ and $b \in \Sigma_n$, i.e., $\MM_n \subseteq \Pi_n \times \Sigma_n$.  The set of all network movements is denoted by $\MM \equiv \MM_1 \sqcup \cdots \sqcup \MM_{|\NX|}$.  A signal phase consists of junction movements that do not conflict with one another.  We denote by $\mathcal{P}_n \subseteq 2^{\MM_n}$ the set of \textit{allowable} phases and by $\mathcal{P} \subseteq \otimes_{n \in \NX} \mathcal{P}_n$ the set of allowable network phasing schemes.  Essentially, an allowable phase is one that does not allow crossing conflicts and only allows merging conflicts between a \textit{protected} movement and a \textit{permitted} movement.  Example allowable phases are depicted in \autoref{F:phases}. 
\begin{figure}[h!]
	\centering
	\resizebox{0.53\textwidth}{!}{%
		\includegraphics{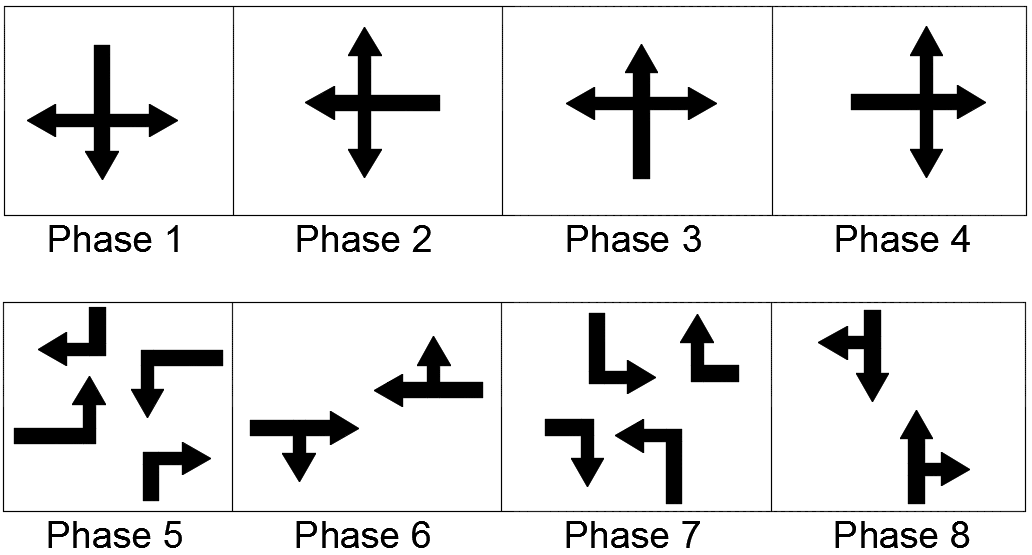}}
	
	\caption{Example phases for a four-leg isolated intersection.} 
	\label{F:phases}
\end{figure}
Boundary flows arise as solutions to node models, typically stated as optimization problems, and depend on traffic densities at the boundaries of all network arcs that interface at the boundary.  They also depend on the control policy at the boundary.  The dependence on intersection control is described below. 

\medskip

\textbf{\fontfamily{cmss}\selectfont Exogenous arrivals}.  For (fictitious) source arcs, we assume random arrivals; for commodity $a \in \AXS$ and $b \in \Sigma(a)$, let $A_a^b(t)$ be a random (cumulative) arrival process with (instantaneous) rate $\lambda_a^b(t) = \EE \frac{\dd A_a^b(t)}{\dd t}$.  Thus, exogenous inflow rates into source arcs are given by\footnote{The processes $A_a^b$ may have jumps.  To be more accurate in such situations, one defines $$\Lambda_a^b(t_1,t_2) \equiv \EE \int_{t_1}^{t_2} \dd A_a^b(t) = \EE \big( A_a^b(t_2) - A_a^b(t_1) \big) = \int_{t_1}^{t_2} \lambda_a^b(t) \dd t.$$  The boundary flux is then given by $$ \underset{\Delta t \downarrow 0}{\lim} \int_{t}^{t + \Delta t} \dd A_a^b(t^{\prime}) = \underset{\Delta t \downarrow 0}{\lim} A_a^b(t + \Delta t) - A_a^b(t);$$ that is, the inflow is the jump size at time $t$.
}
\begin{equation}
	\frac{\dd A_a^b(t)}{\dd t} ~ \mbox{ for } a \in \AXS, b \in \Sigma(a).
\end{equation}

\medskip

\textbf{\fontfamily{cmss}\selectfont Intersection control}.  Let $p_{a,\IN}(t)$ and $p_{a,\OUT}(t)$ denote the upstream and downstream control state at boundaries of arc $a$. The control state at time $t$ is defined as the set of movements that are active at time $t$ as implied by the active phases at the network junctions (e.g., one of the eight phases shown in \autoref{F:phases}).  The boundary flows are given by $q_a^b(0,t) = q_{a,\IN}^b\big(p_{a,\IN}(t)\big)$ and $q_a^b(l_a,t) = q_{a,\OUT}^b\big(p_{a,\OUT}(t)\big)$, where $q_{a,\IN}^b$ and $q_{a,\OUT}^b$ are boundary flux functions, which depend on the (boundary) control variables and, implicitly, on the node dynamics (for instance, $q_{a,\IN}^b$ and $q_{a,\OUT}^b$ cannot exceed local supplies and demands at the arc boundaries). A more accurate way to denote these functions would include the traffic densities adjacent to the boundary, for example $q_{a,\IN}^b\big(p_{a,\IN}(t), \rho_a(0,t), \{\rho_c^b(l_c,t)\}_{c \in \Pi(a)}\big)$ in the case of inbound flow function. We drop the dependence on traffic densities to minimize clutter in our notation.

We denote by $q_{a,b}\big(p_{b,\IN}(t)\big)$ or equivalently $q_{a,b}\big(p_{a,\OUT}(t)\big)$ the rate of flow that departs arc $a \in \Pi(b)$ into arc $b$ at time $t$.  These are related to the commodity flows at the arc boundaries as follows:
\begin{equation}
	q_{a,\IN}^b\big(p_{a,\IN}(t)\big) = \pi_{a,b}(t)\sum_{c \in \Pi(a): (c,a) \in \MM} q_{c,a}\big(p_{a,\IN}(t)\big) \label{E:boudary1}
\end{equation}
and
\begin{equation}
	q_{a,\OUT}^b\big(p_{a,\OUT}(t)\big) = q_{a,b}\big(p_{a,\OUT}(t)\big), \label{E:boudary2}
\end{equation}
where $\pi_{a,b}(t)$ is the percentage of flow into $a$ at time $t$ that is destined to adjacent arc $b \in \Sigma(a)$.  Since $b$ immediately succeeds $a$, movement $(a,b)$ carries all flow out of $a$ that is destined to arc $b$ as stated in \eqref{E:boudary2}.

In the context of signalized urban networks, it was demonstrated in \citep{jabari2016node} that the node coupling, represented by movement flows, is given uniquely by
\begin{align}
	q_{a,b}\big(p_{a,\OUT}(t)\big) = \mathbb{1}_{\{ (a,b) \in p_{a,\OUT}(t) \}}\min \big\{ \delta_{a,b}\big( \rho_a^b(l_a,t) \big), \sigma_b\big( \rho_b(0,t) \big) \big\}, \label{E:nodeModel}
\end{align}
where $\mathbb{1}_{\{ (a,b) \in p_{a,\OUT}(t) \}}$ is an indicator function that returns 1 if the movement $(a,b)$ belongs to the phase $p_{a,\OUT}(t)$ and returns 0 otherwise, $\delta_{a,b}$ is a commodity-specific (local) demand function that depends on the traffic density at the egress of arc $a$, $\sigma_b$ is a (local) supply function that depends on the total traffic density at the ingress of arc $b$: 
\begin{align}
	\rho_b(0,t) = \sum_{c: (b,c) \in \MM} \rho_b^c(0,t).
\end{align}  
\textit{Note that we adopt modified demand functions in order to account for startup lost times}; see \citep{jabari2016node} and references therein for more details. The local demand and supply functions are derived from the stationary flow-density relations $\mathcal{Q}_a$.  Thus, the source of randomness in $q_{a,\IN}^b\big(p_{a,\IN}(t)\big)$ and $q_{a,\OUT}^b\big(p_{a,\OUT}(t)\big)$ is also parametric (i.e., the stochasticity is driven by the random parameters). 
Finally, at the arc boundaries the conservation law \eqref{E:consLaw} is given, for $a \in \AX / \AXS$, by
\begin{align}
\frac{\partial \rho_a^b(x,t)}{\partial t} =- \frac{\partial q_a^b(x,t)}{\partial x} = \left\{
\begin{array}{ll} 
q_{a,\IN}^b\big(p_{a,\IN}(t)\big) - q_a^b(0,t) &  x = 0 \\
& \\
q_a^b(l_a,t) - q_{a,\OUT}^b\big(p_{a,\OUT}(t)\big) & x = l_a
\end{array} \right. \label{E:Boundary}
\end{align}
and for $a \in \AXS$ by
\begin{align}
\frac{\dd \rho_a^b(t)}{\dd t} = - \frac{\dd q_a^b(t)}{\dd x} = \frac{\dd A_a^b(t)}{\dd t} - q_{a,\OUT}^b\big(p_{a,\OUT}(t)\big). \label{E:srcBoundary}
\end{align}
Since $q_a^b(0,t) = q_{a,\IN}^b\big(p_{a,\IN}(t)\big)$ and $q_a^b(l_a,t) = q_{a,\OUT}^b\big(p_{a,\OUT}(t)\big)$, \eqref{E:Boundary} implies that discontinuities in flow rate can occur at the boundaries.  These can be due to the presence of a shock-front, but in the case of arc boundaries, they can also be due to changes in control status at time $t$ and or mismatches between local demands and supplies at the boundaries of the interfacing arcs. 

{\color{black}  We dropped the dependence on $x$ for traffic densities at source arcs.  For source arcs, traffic densities are concentrated at a point and $\rho_a^b(t)$ can be interpreted as the size of the queue at arc $a$ that is destined to downstream arc $b$ at time $t$.  This is equivalent to saying that dynamics at source arcs are governed by point queues (a.k.a. vertical queues).  Equation \eqref{E:srcBoundary} is a simple mass balance equation for multi-commodity queues in continuous time: change in queues size ($\dd \rho_a^b(t) / \dd t$) is equal to rate of inflow ($\dd A_a^b(t) / \dd t$) less the rate of outflow ($q_{a,\OUT}^b\big(p_{a,\OUT}(t)\big)$)\footnote{To see this in discrete time, multiply both sides of \eqref{E:srcBoundary} by $\dd \tau$ and integrate both sides over a short time interval from $t$ to $t + \Delta t$.  We get $$ \rho_a^b(t + \Delta t) = \rho_a^b(t) + \big(A_a^b(t + \Delta t) - A_a^b(t) \big) + \int_t^{t + \Delta t} q_{a,\OUT}^b\big(p_{a,\OUT}(\tau)\big) \dd \tau.$$ The queue size at $a$ destined to $b$ at time $t+\Delta t$ is equal to the queue size at $a$ destined to $b$ at time $t$ plus the cumulative number of vehicles that arrive to $a$ that are destined to $b$ during the time interval from $t$ to $t +\Delta t$, $\big(A_a^b(t + \Delta t) - A_a^b(t) \big)$, less the cumulative number of vehicles that depart $a$ destined to $b$  during the time interval from $t$ to $t+\Delta t$, $\int_t^{t + \Delta t} q_{a,\OUT}^b\big(p_{a,\OUT}(\tau)\big) \dd \tau$.}.}

\subsection{Network capacity region}
Under any network-wide phasing scheme, $p \in \mathcal{P}$, the traffic network can ``support'' arrival processes with certain rates.  Beyond these arrival rates, queues along the source arcs will grow indefinitely. For each $p \in \mathcal{P}$, we say that the network can support an arrival rate vector $\boldsymbol{\lambda}(p)  \in \RR_+^{\AXS}$ if
\begin{align}
	\underset{T \rightarrow \infty}{\lim} \sum_{a \in \AX} \frac{1}{T} \int_0^T \Big( \mathbb{1}_{\{ a \in \AXS \}} \lambda_a(p) + \mathbb{1}_{\{ a \in \AX / \AXS \}}q_{a,\IN}(p) - q_{a,\OUT}(p) \Big) \dd t = 0, \label{E:capReg1}
\end{align}
where with slight abuse of notation, $q_{a,\IN}(p)$ and $q_{a,\OUT}(p)$ are the inflow and outflow rates obtained when the \textit{network} phasing scheme $p$ is active.  
This is interpreted as follows: the phasing scheme $p$ is such that the total arc outflow can accommodate the total arc inflow \textit{in the long run}.  That is, when the initial conditions cease to have influence on performance.   Condition \eqref{E:capReg1} can equivalently be written as 
\begin{align}
	\underset{T \rightarrow \infty}{\lim} \sum_{a \in \AXS} \frac{1}{T} \int_0^T \Big( \lambda_a(p) - q_{a,\OUT}(p) \Big) \dd t = 0, \label{E:capReg2}
\end{align}
which states that, under network phasing scheme $p$, all inbound flows can be accommodated in the long run (i.e., when initial conditions no longer influence the flow rates).  Note that there is an implicit dependence on non-source arcs via the phasing scheme $p$.

In accord with \eqref{E:capReg1} and \eqref{E:capReg2}, each $p \in \mathcal{P}$ defines a set of admissible arrival rates; denote these (convex) polytopes by $\Omega(p)$.  Taking the union of these sets, $\cup_{p \in \mathcal{P}} \Omega(p)$, we get the vectors of all possible arrival rates that the network can support under \textit{all} $p \in \mathcal{P}$.  This is formally defined next.
\begin{definition}[Maximal throughput region]
	The maximal throughput region (a.k.a., capacity region) of the network, denoted by $\boldsymbol{\Lambda}$, is the convex hull of all sets of admissible flows.  That is,
	\begin{equation}
	\boldsymbol{\Lambda} \equiv \mathsf{Conv}\big( \underset{p \in \mathcal{P}}{\cup} \Omega(p) \big).
	\end{equation}
\end{definition}
Arrival rates that lie in $\boldsymbol{\Lambda}$ but \textbf{not} in $\cup_{p \in \mathcal{P}} \Omega(p)$ are interpreted as arrival rates that can be supported by switching between phasing schemes that lie in the latter (i.e., \textit{time-sharing}). A control policy that can support all possible arrival rates in $\boldsymbol{\Lambda}$ is referred to as a \textit{throughput-maximal control policy}.   We denote a control policy by a vector of network control states: at time $t$ the network control state is denoted by $p(t) \equiv [\cdots ~ p_{a,\IN}(t) ~~ p_{a,\OUT}(t) ~ \cdots ]^{\top}$, a policy is an entire curve $p(\cdot)$.

We give two examples to illustrate the notion of capacity region.  The first is the simple isolated intersection of two one-way streets depicted in \autoref{F:Ex1}a.  If Phase 1 is active, the maximum arrival rate that can be accommodated is $\lambda_1^{\max}$, the saturation flow rate of the Arc1-Arc3 (eastbound) movement; if Phase 2 is active, the maximum arrival rate that can be accommodated is $\lambda_2^{\max}$, the saturation flow rate of the Arc2-Arc4 (eastbound) movement. This implies that
\begin{equation}
	\underset{T \rightarrow \infty}{\lim} \frac{1}{T} \int_0^T q_{1,\OUT}(p_1) \in [0, \lambda_1^{\max}]
\end{equation}
whereas
\begin{equation}
	\underset{T \rightarrow \infty}{\lim} \frac{1}{T} \int_0^T q_{1,\OUT}(p_2) = 0.
\end{equation}
Similarly,
\begin{equation}
	\underset{T \rightarrow \infty}{\lim} \frac{1}{T} \int_0^T q_{2,\OUT}(p_1) =0
\end{equation}
while
\begin{equation}
	\underset{T \rightarrow \infty}{\lim} \frac{1}{T} \int_0^T q_{2,\OUT}(p_2) \in [0, \lambda_2^{\max}].
\end{equation}
Thus, when $p=p_1$, $0 \le \lambda_1(p_1) \le \lambda_1^{\max}$ while $\lambda_2(p_1) = 0$.  Similarly, when $p=p_2$, $\lambda_1(p_2) = 0$ while $0 \le \lambda_2(p_2) \le \lambda_2^{\max}$. Hence, 
\begin{align}
	\Omega(p_1) = \{(\lambda_1,\lambda_2): 0 \le \lambda_1 \le \lambda_1^{\max}, \lambda_2 = 0 \}
\end{align}
and
\begin{align} 
	\Omega(p_2) = \{(\lambda_1,\lambda_2): \lambda_1 = 0, 0 \le \lambda_2 \le \lambda_2^{\max} \}.
\end{align}  
The capacity region, depicted in \autoref{F:Ex1}c as a shaded region, is the set of maximal arrival rates $(\lambda_1,\lambda_2)$ that can be accommodated by switching between phases $p_1$ and $p_2$.
\begin{figure}[h!]
	\centering
	\resizebox{0.95\textwidth}{!}{%
		\includegraphics{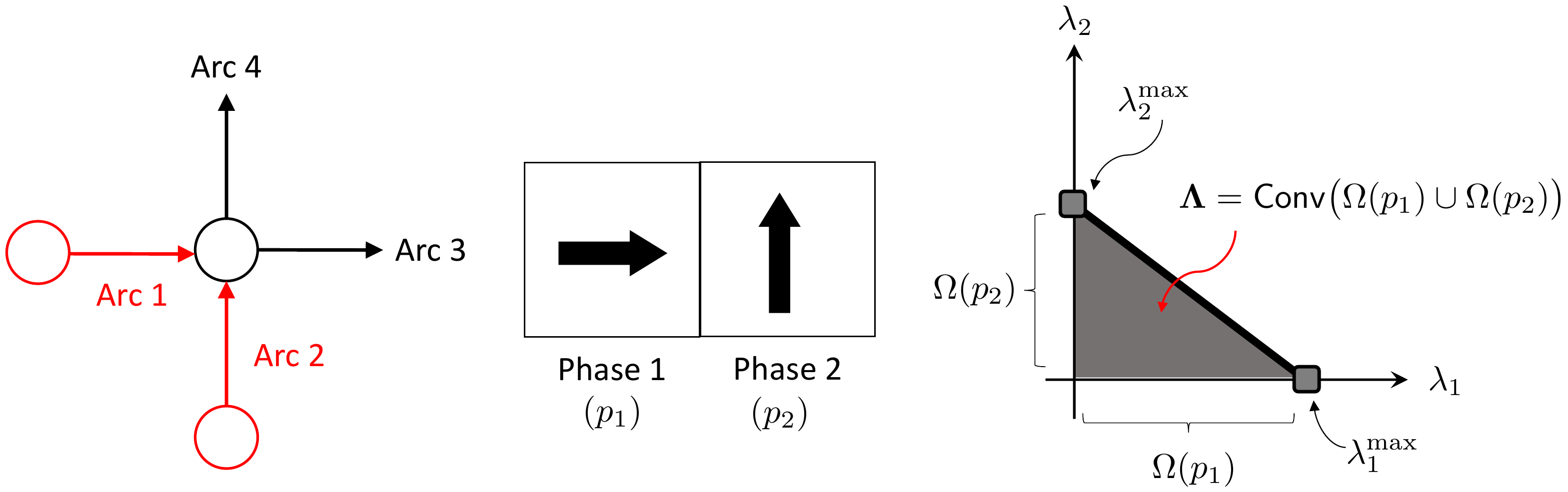}}
	
	\hspace{-.5 in} (a) \hspace{1.5 in} (b) \hspace{1.6 in} (c)
	\caption{Example isolated intersection and the associated capacity region, (a) intersection layout (arcs 1 and 2 are source arcs), (b) the two possible phases, (c) the capacity region.} 
	\label{F:Ex1}
\end{figure}

The second example is borrowed from \citep[Example 3]{varaiya2013max} and illustrates how an instability forms in the proposed model, namely given that we consider finite spatial arc capacities.  In this example, depicted in \autoref{F:Ex2}, there is one source arc. 
\begin{figure}[h!]
	\centering
	\resizebox{0.5\textwidth}{!}{%
		\includegraphics{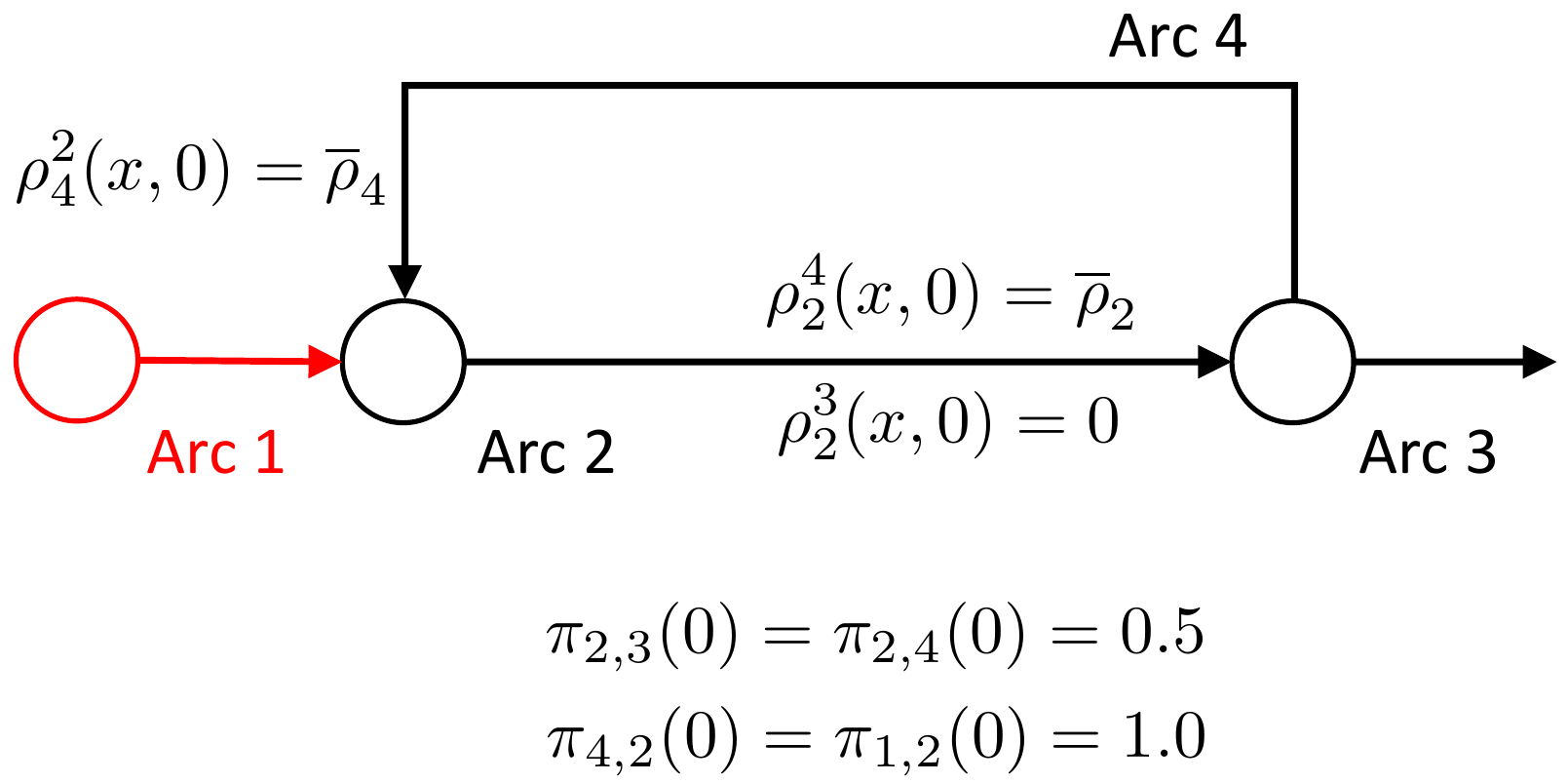}}
	\caption{A gridlock scenario (adopted from \citep{varaiya2013max} and reproduced).} 
	\label{F:Ex2}
\end{figure}
Hence, the capacity region is one dimensional.  The initial conditions depicted in the figure are such that the network is in a state of gridlock at time $t=0$.  Moreover, the turning desires shown in the figure prevent all vehicles from moving into their desired destination arcs.  In this case, the capacity region consists of the singleton set $\lambda = 0$.  That is, the maximal arrival rate the network can accommodate is zero.  \textit{Any other arrival rate is outside of the capacity region and cannot be accommodated by any control policy, not BP, not CABP, not PWBP, nor any signal timing optimization technique.}  In such cases, the only way to relieve gridlock is to re-route vehicles; the subject of control+routing is beyond the scope of the present paper, we leave it to future research.

\subsection{Position-weighted back-pressure (PWBP)}
For any intersection $n \in \NX$, we assume that controllers are capable of assessing the (average) movement fluxes associated with any possible phase $p \in \mathcal{P}_n$.  That is, for any $(a,b) \in \MM_n$, $\EE \big( q_{a,b}(p) \big| \boldsymbol{\rho}(t) \big) \equiv \EER q_{a,b}(p)$ is known or can be estimated by the controller (locally). Omitting dependence of $\delta_{a,b}$ and $\sigma_b$ in \eqref{E:nodeModel} to simplify notation, define $P_{a,b}^{\boldsymbol{\rho}(t)} \equiv \PP\big(\delta_{a,b} - \sigma_b \le 0 |\boldsymbol{\rho}(t) \big)$.  Then
\begin{align}
	\EER q_{a,b}(p) = \mathbb{1}_{\{(a,b) \in p\}} \EER \min \big\{ \delta_{a,b}, \sigma_b \big\} 
	= \mathbb{1}_{\{(a,b) \in p\}} \Big( P_{a,b}^{\boldsymbol{\rho}(t)} \EER \delta_{a,b} + \big(1 - P_{a,b}^{\boldsymbol{\rho}(t)} \big) \EER \sigma_b \Big).
\end{align}
Note that $P_{a,b}^{\boldsymbol{\rho}(t)}$, $\EER \delta_{a,b}$, and $\EER \sigma_b$ are deterministic functions of $\boldsymbol{\rho}(t)$ that depend on the distributions of the parameters of $\delta_{a,b}$ and $\sigma_b$.  These distributions can be established empirically using historical data \citep{jabari2014probabilistic}.  The splits $\pi_{a,b}(t)$ are also treated as random quantities that are to be estimated or measured. In a fully automated system, these random variables may degenerate, that is, it is easy to imagine that they can be measured with high accuracy and become deterministic quantities.  In present day settings they need to be estimated.  The setting envisaged in this paper is one with mixed automated/connected and traditional vehicles.  Connected vehicles announce their turning desires upon entering arc $a$ and may serve as probes to allow the controller to estimate traffic conditions along the arc and the split variables.  Empirical techniques may also be employed for this purpose; we refer to \citep{zheng2017estimating} for a recent approach and to \citep{rodriguez2019location} for a recent article on reconstructing turning movements.

The traffic state at time $t$, $\boldsymbol{\rho}(t)$, requires a traffic state estimation procedure that is capable of producing real-time estimates under present day instrumentation in the real world.  We refer to \citep{seo2017traffic,zheng2018traffic,van2018macroscopic} and references therein for recent research on traffic state estimation tools. 

For each $n \in \NX$ and each $(a,b) \in \MM_n$, we define the \textit{weight variable}
\begin{align}
w_{a,b}(t) \equiv
 \left\{
\begin{array}{ll} 
\bigg| c_{a,b} \rho_a^b(t) 
-  \int_0^{l_b} \Big| \frac{l_b - x}{l_b} \Big| \sum_{c \in \Sigma(b): \atop (b,c) \in \MM} c_{b,c} \pi_{b,c}(t) \rho_b^c(x,t) \dd x \bigg|, & a \in \AXS \\
& \\
\bigg| c_{a,b} \int_0^{l_a} \Big| \frac{x}{l_a} \Big| \rho_a^b(x,t) \dd x 
-  \int_0^{l_b} \Big| \frac{l_b - x}{l_b} \Big| \sum_{c \in \Sigma(b): \atop (b,c) \in \MM} c_{b,c} \pi_{b,c}(t) \rho_b^c(x,t) \dd x \bigg|, & a \not\in \AXS
\end{array}, \right. \label{E:BPwt}
\end{align}
which depends on the (commodity) density curves along arcs $a$ and $b$.  To interpret this, first note that
\begin{equation}
	\int_0^{l_a} \rho_a^b(x,t) \dd x
\end{equation}
is just the total traffic volume (queue size) along arc $a$ that is destined to arc $b$.  Then the first integral inside the square brackets in \eqref{E:BPwt} can be interpreted as a weighted queue size, where traffic densities at the downstream end of arc $a$ (at $x = l_a$) have the (maximal) weight of one, while traffic densities at the upstream end of $a$ (at $x = 0$) have a weight of zero.  In between, the weights increase linearly with $x$.  Similarly, the second integral inside the square brackets in \eqref{E:BPwt} can also be interpreted as a weighted queue size, but with the weights decreasing linearly with $x$. Hence, the weight associated with movement $(a,b)$ decreases as the traffic densities at upstream end (ingress) of arc $b$ increase and vice versa, it increases when the traffic densities are concentrated at the downstream end of arc $a$ and vice versa.  The movement constants $c_{a,b}$ in \eqref{E:BPwt} allow for assigning higher weights to certain movements.  

{\color{black}
Let $p_{\mathrm{PWBP}}(t) \in \mathcal{P}$ denote the network-wide phasing scheme chosen by PWBP at time $t$.  It is the phasing scheme that solves the following problem:
\begin{align}
	p_{\mathrm{PWBP}}(t) \in \underset{p \in \mathcal{P}}{\arg \max} \sum_{(a,b) \in \MM} w_{a,b}(t) \EER q_{a,b}(p), \label{E:PWBP_0}
\end{align}
where set inclusion ($\in$) as opposed to equality is used since the right-hand side may not be unique.  The optimization problem can alternatively be written as
\begin{align}
	p_{\mathrm{PWBP}}(t) \in \underset{p_1 \in \mathcal{P}_1, \hdots, p_{|\NX|} \in \mathcal{P}_{|\NX|}}{\arg \max} ~ \sum_{n \in \NX} \sum_{(a,b) \in \MM_n} w_{a,b}(t) \EER q_{a,b}(p_n). \label{E:PWBP_1}
\end{align}
Notice that the decision variables, the node phase schemes $p_1, \hdots, p_{|\NX|}$, each only appear in one term in the outer sum and, more importantly, for any $n,m \in \NX$ such that $n\ne m$, we have that $p_n \cap p_m = \emptyset$.  Hence, the problem \eqref{E:PWBP_1} decomposes naturally by node; that is, we have that
\begin{align}
\underset{p_1 \in \mathcal{P}_1, \hdots, p_{|\NX|} \in \mathcal{P}_{|\NX|}}{\max}  \sum_{n \in \NX} \sum_{(a,b) \in \MM_n} w_{a,b}(t) \EER q_{a,b}(p_n) 
= \sum_{n \in \NX} \underset{p \in \mathcal{P}_n}{\max}  \sum_{(a,b) \in \MM_n} w_{a,b}(t) \EER q_{a,b}(p) \label{E:PWBP_2}
\end{align}
and consequently, the problem \eqref{E:PWBP_0} can be parallelized\footnote{This is a distinguishing feature of intersection control.  In wireless communication, BP is often applied to packet \textit{routing}; there the problem cannot be parallelized as network paths can have common arcs. The decomposition in such cases is heuristic.  In fact such problems are known to be NP-Hard. Applications of BP to routing in transportation networks inherit this drawback.}.  Then, under PWBP control, the phase that is active at node $n$ at time $t$, denoted $p_{n,\mathrm{PWBP}}(t)$, is given by
\begin{equation}
	p_{n,\mathrm{PWBP}}(t) \in \underset{p \in \mathcal{P}_n}{\arg \max} \sum_{(a,b) \in \MM_n} w_{a,b}(t) \EER q_{a,b}(p). \label{E:PWBP}
\end{equation}

Since the number of possible phases at any intersection tends to be small (typically four-eight), \eqref{E:PWBP} can be easily solved by direct enumeration.  This allows for real-time distributed implementation of the proposed approach.  That \eqref{E:PWBP_0} and, consequently \eqref{E:PWBP}, is network stabilizing is the subject of \autoref{S:stability}. 
}

When there exists more than one solution to \eqref{E:PWBP}, a randomization procedure that applies equal weight to all the maximizers is employed.  This helps ensure work conservation as discussed below. Implementation of PWBP control for node $n$ at time instant $t$ is summarized in Algorithm \ref{A:PWBP}. 
\renewcommand{\algorithmicrequire}{\textbf{Input:}}
\renewcommand{\algorithmicensure}{\textbf{Iterate:}} 
\begin{algorithm}[h!]
	\caption{Position weighted backpressure phasing for node $n$ at time $t$: $\mathsf{PWBP}(n,t)$}
	\label{A:PWBP}
	\small
	\begin{algorithmic}[1]
		\REQUIRE Road geometry: $\{l_a\}_{a \in \Pi_n}$, $\{l_b\}_{b \in \Sigma_n}$; constants: $\{c_{a,b}\}_{(a,b) \in \MM_n}$; (estimated) traffic state at \\ time $t$: $\boldsymbol{\rho}(t)$; (estimated) splits at time $t$: $\{\pi_{a,b}(t)\}_{(a,b) \in \MM}$; the distributions of the parameters  \\ of $\{\delta_{a,b}\}_{(a,b) \in \MM_n}$ and $\{\sigma_b\}_{b \in \Sigma_n}$
		\ENSURE 
		\FOR { $(a,b) \in \MM_n$}
		\STATE $w_{a,b}(t) \mapsfrom \left\{
		\begin{array}{ll} 
		\bigg| c_{a,b} \rho_a^b(t) 
		-  \int_0^{l_b} \Big| \frac{l_b - x}{l_b} \Big| \sum_{c \in \Sigma(b): \atop (b,c) \in \MM} c_{b,c} \pi_{b,c}(t) \rho_b^c(x,t) \dd x \bigg|, & a \in \AXS \\
		& \\
		\bigg| c_{a,b} \int_0^{l_a} \Big| \frac{x}{l_a} \Big| \rho_a^b(x,t) \dd x 
		-  \int_0^{l_b} \Big| \frac{l_b - x}{l_b} \Big| \sum_{c \in \Sigma(b): \atop (b,c) \in \MM} c_{b,c} \pi_{b,c}(t) \rho_b^c(x,t) \dd x \bigg|, & a \not\in \AXS
		\end{array}, \right.$
		\FOR {$p \in \mathcal{P}_n$}
		\STATE $\EE^{\boldsymbol{\rho}(t)} q_{a,b}(p) \mapsfrom \mathbb{1}_{\{(a,b) \in p\}} \EER \min \big\{ \delta_{a,b}, \sigma_b \big\}$
		\ENDFOR
		\ENDFOR
		\renewcommand{\algorithmicensure}{\textbf{Output:}}
		\ENSURE
		\IF {$\big|\big\{ \arg \max_{p \in \mathcal{P}_n} \sum_{(a,b) \in \MM_n} w_{a,b}(t) \EER q_{a,b}(p) \big\}\big| = 1$}
		\STATE $p_{n,\mathrm{PWBP}}(t) \mapsfrom \underset{p \in \mathcal{P}_n}{\arg \max} \sum_{(a,b) \in \MM_n} w_{a,b}(t) \EER q_{a,b}(p)$
		\ELSE
		\STATE Select $p_{n,\mathrm{PWBP}}(t)$ at random from the set of optima (each assigned equal probability):  $p_{n,\mathrm{PWBP}}(t) \sim \Big\{ \underset{p \in \mathcal{P}_n}{\arg \max} \sum_{(a,b) \in \MM_n} w_{a,b}(t) \EER q_{a,b}(p) \Big\}$.
		\ENDIF
	\end{algorithmic}
\end{algorithm}

One of the advantages of a continuous time formulation is that Algorithm \ref{A:PWBP} can be implemented at pre-specified cadence.  Moreover, the cadence can vary from one intersection to another in order to accommodate such constraints as minimum greens (to avoid aggressive oscillations in the control dynamics), pedestrian movements, and so on.  To elaborate, let $\tau_n$ denote the minimum phase length for node $n$.  The signal phasing sequence is given by $p_{n,\mathrm{PWBP}}(k\tau_n) = \mathsf{PWBP}(n,k\tau_n)$ where $k$ is a positive integer.  Note that this can be easily generalized to the case where $\tau_n$ is not only node-specific, but also phase-specific allowing for varying minimum greens by approach or even intersection movement.  The cadence PWBP can be tuned in practice based on standard principles of traffic engineering.  The main advantage of the proposed approach is that it is decentralized; that is, the calculations can be parallelized over the network nodes.  

\medskip

\textbf{\fontfamily{cmss}\selectfont PWBP and work conservation}.  Work conservation of PWBP control follows from two features of the proposed approach.  The first feature is the node model used: It was demonstrated in \citep{jabari2016node} that the node model produces \textit{holding-free solutions}.  Hence, for any chosen phase $p \in \mathcal{P}$ (which dictates the node supplies), as long as there exist supply along the outbound arcs, demands at the inbound are guaranteed to be served.  The second feature is that the phase chosen by PWBP depends on both the movement weights $w_{a,b}$ and the expected movement fluxes, $\EER q_{a,b}(p)$.  Since the weights are non-negative, the phase chosen is guaranteed to result in (holding-free) flow across the node as long as at least one of the movements has a non-negative expected flux, $\EER q_{a,b}(p)$.  Holding only occurs when $\EER q_{a,b}(p) = 0$ for all movements $(a,b) \in \MM_n$.  However, this is a gridlock scenario and no work is lost.  In the case where the expected fluxes are positive but the weights are zero, if all other expected fluxes are zero, work may be lost.  This corresponds to an alternative scenario where
\begin{align} 
\underset{p \in \mathcal{P}_n}{\max} \sum_{(a,b) \in \MM_n} w_{a,b}(t) \EER q_{a,b}(p) = 0.
\end{align}  
In this case, the randomization procedure ensures with probability 1 that loss of work does not persist.

\section{Network stability}
\label{S:stability}

\subsection{Lyapunov functional and stability}
The traffic network is said to be \textit{strongly stable} if \citep[Definition 2.7]{neely2010stochastic}:
\begin{equation}
	\underset{T \rightarrow \infty}{\lim \sup} ~ \frac{1}{T} \int_0^T \EE \bigg( \sum_{\substack{(a,b) \in \MM: \\ a \in \AXS}} \rho_a^b(t) + \sum_{\substack{(a,b) \in \MM: \\ a \not\in \AXS}} \int_0^{l_a} \rho_a^b(x,t) \dd x \bigg) \dd t < \infty. \label{E:stable}
\end{equation}
{\color{black} Note that stability conditions are typically stated as absolute values of the dynamical variables.  We have omitted the absolute values since traffic densities and queue sizes are non-negative with probability 1.  We refer to \citep{jabari2012stochastic,jabari2012stochasticDiss,jabari2018stochastic} for an analysis of non-negativity of stochastic traffic models.}   
Since the network traffic densities depend (implicitly) on the control decisions at the network nodes, strong stability implies that the control in place ensure that the network densities do not grow without bound \emph{in the long run}.  This section demonstrates that as long as such a control policy exists\footnote{Otherwise, there does not exist a control policy capable of stabilizing the network.  Hence, this is a feasibility assumption.}, the PWBP algorithm ensures strong stability.

Since the spatial capacities of non-source arcs in the network are naturally bounded, the stability condition in \eqref{E:stable} can be restated in terms of source movements only, that is, the traffic network is said to be strongly stable if
\begin{equation}
	\underset{T \rightarrow \infty}{\lim \sup} ~ \frac{1}{T} \int_0^T \sum_{\substack{(a,b) \in \MM: \\ a \in \AXS}} \EE \rho_a^b(t) \dd t < \infty. \label{E:stable1}
\end{equation}
{\color{black}
It must be emphasized, however, that this does not imply that stability can be established by analyzing the queue sizes at the source arcs independent of the non-source arcs.  In simple terms, stability is achieved when \emph{long-run outflows from the network queues exceed the long-run inflows}.  Otherwise, queues build up indefinitely.  A control policy is stabilizing when it ensures that long-run outflows exceed long-run inflows.  Since the outflows from source arcs depend on non-source arcs, stability cannot be established without examining traffic dynamics along the non-source arcs in the network.}

Consider the network-wide energy functional $V: D \rightarrow \RR$ with domain $D$ being an appropriately defined $|\AX|$-dimensional set of curves.  $V$ is defined as
\begin{align}
	V \big( \boldsymbol{\rho}(t) \big) \equiv  \frac{1}{2} \sum_{\substack{(a,b) \in \MM: \\ a \in \AXS}} c_{a,b} \big(\rho_a^b(t) \big)^2 
	+ \frac{1}{2} \sum_{\substack{(a,b) \in \MM: \\ a \not\in \AXS}} c_{a,b} \int_0^{l_a} \int_0^{l_a} \Big| \frac{l_a - x - x^{\prime}}{l_a} \Big| \rho_a^b(x^{\prime},t)\rho_a^b(x,t) \dd x^{\prime} \dd x, \label{E:Lyapunov}
\end{align}
where $\{c_{a,b} \}_{(a,b) \in \MM}$ are non-negative finite constants. 
It can be easily shown that $V$ is a Lyapunov functional: (i) $V \big(\boldsymbol{\rho}(t) \big) \ge 0$ almost surely since traffic densities are non-negative (with probability 1) and (ii) $V(\boldsymbol{\rho}) = 0$ if and only if $\rho_a^b(t) = \rho_a^b(x,t) = 0$ almost surely for all $(a,b) \in \MM$ and all $x \in [0,l_a]$\footnote{One can construct pathological density curves with non-zero density spikes, where $V = 0$.  However, such densities occur with probability zero.  Technically, these are overcome by using equivalence classes of density curves, but we shall avoid this level of technicality to promote readability.}. \autoref{thm:LyapunovDrift} below provides a sufficient condition for strong stability using the definition of Lyapunov functionals. A notational convention used below is $\dot{V} \equiv \dd V / \dd t$.  

\begin{lemma}
	\label{thm:LyapunovDrift}
	For the Lyapunov functional \eqref{E:Lyapunov}, suppose $\EE V \big( \boldsymbol{\rho}(0) \big) < \infty$.  If there exist constants $0 < K < \infty$ and $0 < \epsilon <\infty$ such that
	\begin{align}
		\EER \dot{V} \big( \boldsymbol{\rho}(t) \big) \le K - \epsilon \Big( \sum_{\substack{(a,b) \in \MM: \\ a \in \AXS}} \EE \rho_a^b(t) + \sum_{\substack{(a,b) \in \MM: \\ a \not\in \AXS}} \EE \int_0^{l_a} \rho_a^b(x,t) \dd x \Big) \label{E:stCond}
	\end{align}
	holds for all $t \ge 0$ and all possible $\boldsymbol{\rho}(t)$, then the traffic network is strongly stable.
\end{lemma}
\begin{proof}
	We first integrate both sides of \eqref{E:stCond} over the interval $[0,T]$ and take expectation of both sides of the inequality to obtain
	\begin{align}
		\EE \int_0^T \EER \dot{V} \big( \boldsymbol{\rho}(t) \big) \dd t  \le KT - \epsilon \int_0^T \EE \bigg(\sum_{\substack{(a,b) \in \MM: \\ a \in \AXS}} \rho_a^b(t) + \sum_{\substack{(a,b) \in \MM: \\ a \not\in \AXS}} \int_0^{l_a} \rho_a^b(x,t) \dd x \bigg) \dd t.	\label{temp1}
	\end{align}
	Reversing the order of the (outer) expectation and the integral, the left-hand side becomes
	\begin{align}
		\int_0^T \EE  \EER \big( \dot{V} \big( \boldsymbol{\rho}(t) \big) \dd t = \int_0^T \EE \dot{V} \big( \boldsymbol{\rho}(t) \big) \dd t = \EE \int_0^T \dot{V} \big( \boldsymbol{\rho}(t) \big) \dd t.
	\end{align}
	This is equal to $\EE V \big( \boldsymbol{\rho}(T) \big) - \EE V \big( \boldsymbol{\rho}(0) \big)$. Combining this with \eqref{temp1}, we get
	\begin{align}
	\EE V \big( \boldsymbol{\rho}(T) \big) - \EE V \big( \boldsymbol{\rho}(0) \big) \le KT - \epsilon \int_0^T \EE \bigg(\sum_{\substack{(a,b) \in \MM: \\ a \in \AXS}} \rho_a^b(t) + \sum_{\substack{(a,b) \in \MM: \\ a \not\in \AXS}} \int_0^{l_a} \rho_a^b(x,t) \dd x \bigg) \dd t.
	\end{align}
	Dividing both sides by $T\epsilon$ and noting that $\EE V \big( \boldsymbol{\rho}(T) \big) \ge 0$ (by definition of $V$), we get the inequality
	\begin{align}
		\frac{1}{T} \int_0^T \EE \bigg(\sum_{\substack{(a,b) \in \MM: \\ a \in \AXS}} \rho_a^b(t) + \sum_{\substack{(a,b) \in \MM: \\ a \not\in \AXS}} \int_0^{l_a} \rho_a^b(x,t) \dd x \bigg) \dd t \le \frac{K}{\epsilon} + \frac{1}{\epsilon T} \EE V \big( \boldsymbol{\rho}(0) \big).
	\end{align}
	Noting that $\EE V \big( \boldsymbol{\rho}(0) \big) < \infty$, the result follows immediately upon taking the limit on both sides as $T \rightarrow \infty$ in accord with \eqref{E:stable}.
\end{proof}

\subsection{Stability of PWBP}
According to \autoref{thm:LyapunovDrift}, finding (finite) constants $K$ and $\epsilon$ that satisfy the condition \eqref{E:stCond} will ensure strong stability of the dynamics at the network level.  The constant $K$ is established using the boundedness of the fluxes $q_a(\cdot,\cdot)$, which is a property of traffic flow (i.e., a physical property that must be ensured by any model).  On the other hand, $\epsilon$ depends on the intersection control polices. \autoref{Lem:flowBound} provides a necessary ingredient that will be used later to establish the constant $K$.

\begin{lemma}
	\label{Lem:flowBound}
	Let $a \in \AX / \AXS$ and suppose there exist constants $0 \le \overline{q}_a < \infty$ and $0 \le \overline{\rho}_a < \infty$ such that $\PP(q_a^b(x,t) \le \overline{q}_a) = 1$ and $\PP(\rho_a^b(x,t) \le \overline{\rho}_a) = 1$ for any $(a,b) \in \MM$, any $x \in [0,l_a]$, and any $t \ge 0$.  Then, there exist constants $0 \le K_a^{(1)} < \infty$ and $0 \le K_a^{(2)} < \infty$ such that, with probability 1,
	\begin{itemize}
		\item[(i)] for any $(x_1,x_2) \subseteq [0,l_a]$
		\begin{equation}
			\EER \Big( q_a^b(l_a,t)  \int_{x_1}^{x_2} \rho_a^b(x,t) \dd x  \Big)  \le K_a^{(1)}, \label{E:lemmaI}
		\end{equation}
		\item[(ii)] for any $(x_1,x_2) \subseteq (0,l_a)$ and any $(x_3,x_4) \subseteq [0,l_a]$
		\begin{equation}
			- \EER \int_{x_3}^{x_4} \int_{x_1}^{x_2} \rho_a^b(x,t) \frac{\partial q_a^b(x^{\prime},t)}{\partial x} \dd x^{\prime} \dd x \le K_a^{(2)}. \label{E:lemmaII}
		\end{equation}
	\end{itemize}	
\end{lemma}

\begin{proof}
	First note that, with probability 1
	\begin{align}
		- \EER \int_{x_1}^{x_2} \frac{\partial q_a^b(x,t)}{\partial x} \dd x = \EER q_a^b(x_1,t) - \EER q_a^b(x_2,t) \le \overline{q}_a.
	\end{align}
	By the boundedness properties of $\rho_a^b(\cdot,\cdot)$ for $a \in \AX / \AXS$, it holds that
	\begin{align}
		\PP \left( \int_{x_3}^{x_4} \rho_a^b(x,t) \dd x  \le l_a \overline{\rho}_a \right)  = 1. \label{E:lemmaPfTemp}
	\end{align}
	Hence, with probability 1, we have that
	\begin{align}
		-\EER  \int_{x_3}^{x_4} \int_{x_1}^{x_2} \rho_a^b(x,t) \frac{\partial q_a(x^{\prime},t)}{\partial x} \dd x^{\prime} \dd x
		= - \Big(\int_{x_3}^{x_4} \rho_a^b(x,t) \dd x \Big) \EER \int_{x_1}^{x_2} \frac{\partial q_a^b(x,t)}{\partial x} \dd x  
		\le l_a \overline{\rho}_a \overline{q}_a.
	\end{align}
	The bound in \eqref{E:lemmaII} follows immediately and \eqref{E:lemmaI} follows from \eqref{E:lemmaPfTemp} along with the boundedness of $q_a^b(x,t)$ and then applying the Cauchy-Schwartz inequality.
\end{proof}

\begin{corollary}
	\label{Cor:flowBound}
	Let $a \in \AX / \AXS$ and assume the probabilistic bounds of \autoref{Lem:flowBound}.  Then, there exists a constant $0 \le K < \infty$ such that
	\begin{align}
		K \ge \sum_{(a,b) \in \MM: a \not\in \AXS} c_{a,b} \bigg[ & \EER \Big(  q_a^b(l_a,t)  \int_0^{l_a} \Big| \frac{x}{l_a} \Big| \rho_a^b(x,t) \dd x  \Big) \nonumber \\
		& - \EER \Big( \int_0^{l_a} \int_{0+}^{l_a-} \Big| \frac{l_a - x - x^{\prime}}{l_a} \Big| \rho_a^b(x,t) \frac{\partial q_a^b(x^{\prime},t) }{\partial x}  \dd x^{\prime} \dd x \Big) \bigg]. \label{E:corrTemp}
	\end{align}
\end{corollary}
\begin{proof}
	Since $| (l_a - x - x^{\prime}) / l_a | \le 1$ for all $(x,x^{\prime}) \in [0,l_a] \times (0,l_a)$, it follows from \autoref{Lem:flowBound} that there exists constants $0 \le K_a^{(1)} < \infty$ and $0 \le K_a^{(2)} < \infty$ for each $a \in \AX / \AXS$ that bind each of the terms in the sums in \eqref{E:corrTemp} from above.  Defining 
	\begin{equation}
		K \equiv \sum_{(a,b) \in \MM: a \not\in \AXS} \big( c_{a,b} K_a^{(1)} + c_{a,b} K_a^{(2)} \big), 
	\end{equation}	
	the result follows immediately.
\end{proof}

\begin{theorem}[Stability of PWBP]
	\label{thm:pwbp}
	Assume that the boundedness conditions of \autoref{Lem:flowBound} hold for all $a \in \AX / \AXS$, assume that arrival rates lie in $\boldsymbol{\Lambda}$, that is, there exists a control policy $p^*(\cdot)$ that can stabilize the network in the sense defined in \autoref{thm:LyapunovDrift}, and let $w_{a,b}(\cdot)$ for each movement $(a,b) \in \MM$ be as defined in \eqref{E:BPwt}.  Then the policy
	\begin{equation}
		p_{n,\mathrm{PWBP}}(t) \equiv \underset{p \in \mathcal{P}_n}{\arg \max} \sum_{(a,b) \in \MM_n} w_{a,b}(t) \EER q_{a,b}(p) ~~ \forall n \in \NX, \label{E:BP}
	\end{equation}
	ensures strong stability of the traffic network.
\end{theorem}
\begin{proof}
	We demonstrate stability by showing that the conditions of \autoref{thm:LyapunovDrift} hold:
	\begin{align}
		\EER \dot{V} \big( \boldsymbol{\rho}(t) \big) = & ~ \EER \sum_{(a,b) \in \MM: \atop a \in \AXS} c_{a,b} \frac{\dd \rho_a^b(t)}{\dd t} \rho_a^b(t) \nonumber \\
		& + \frac{1}{2}\EER \sum_{(a,b) \in \MM: \atop a \not\in \AXS} c_{a,b} \frac{\dd}{\dd t} \int_0^{l_a} \int_0^{l_a} \Big| \frac{l_a - x - x^{\prime}}{l_a} \Big| \rho_a^b(x^{\prime},t)\rho_a^b(x,t) \dd x^{\prime} \dd x.
	\end{align}
	Applying the Leibniz rule to the second term on the on the right-hand side, the derivatives of the integrals inside the sums become
	\begin{align}
		\int_0^{l_a} \int_0^{l_a} \Big| \frac{l_a - x - x^{\prime}}{l_a} \Big| \frac{\partial }{\partial t} \big( \rho_a^b(x^{\prime},t)\rho_a^b(x,t) \big)  \dd x^{\prime} \dd x.
	\end{align}
	From \eqref{E:consLaw}, this becomes
	\begin{align}
		-  \int_0^{l_a} \int_0^{l_a} \Big| \frac{l_a - x - x^{\prime}}{l_a} \Big| \rho_a^b(x,t) \frac{\partial q_a^b(x^{\prime},t) }{\partial x}  \dd x^{\prime} \dd x.
	\end{align}
	For each $a$, these integrals can be decomposed as
	\begin{align}
		\frac{\partial q_a^b(0,t) }{\partial x} \int_0^{l_a}  \Big| \frac{l_a - x}{l_a} \Big| & \rho_a^b(x,t)  \dd x + \frac{\partial q_a^b(l_a,t) }{\partial x} \int_0^{l_a} \Big| \frac{x}{l_a} \Big| \rho_a^b(x,t) \dd x \nonumber \\
		&+ \int_0^{l_a} \int_{0+}^{l_a-} \Big| \frac{l_a - x - x^{\prime}}{l_a} \Big| \rho_a^b(x,t) \frac{\partial q_a^b(x^{\prime},t) }{\partial x}  \dd x^{\prime} \dd x.
	\end{align}
	Then from \eqref{E:Boundary}, we have that
	\begin{align}
		\EER \dot{V} \big( \boldsymbol{\rho}(t) \big)
		 & = \sum_{(a,b) \in \MM: a \in \AXS} c_{a,b} \EER \Big( \frac{\dd A_a^b(t)}{\dd t} \rho_a^b(t) - q_{a,\OUT}^b\big(p_{a,\OUT}(t)\big)  \rho_a^b(t) \Big) \nonumber \\
		& ~~+ \sum_{(a,b) \in \MM: a \not\in \AXS} c_{a,b} \EER \Big( q_{a,\IN}^b\big(p_{a,\IN}(t)\big) \int_0^{l_a} \Big| \frac{l_a - x}{l_a} \Big| \rho_a^b(x,t) \dd x \Big) \nonumber \\ 
		& ~~- \sum_{(a,b) \in \MM: a \not\in \AXS} c_{a,b} \EER \Big( q_a^b(0,t)  \int_0^{l_a} \Big| \frac{l_a - x}{l_a} \Big| \rho_a^b(x,t) \dd x \Big)  \nonumber \\
		& ~~- \sum_{(a,b) \in \MM: a \not\in \AXS} c_{a,b} \EER \Big( q_{a,\OUT}^b \big(p_{a,\OUT}(t)\big) \int_0^{l_a} \Big| \frac{x}{l_a} \Big| \rho_a^b(x,t) \dd x \Big) \nonumber \\ 
		& ~~+ \sum_{(a,b) \in \MM: a \not\in \AXS} c_{a,b} \bigg[ \EER \Big( q_a^b(l_a,t)  \int_0^{l_a} \Big| \frac{x}{l_a} \Big| \rho_a^b(x,t) \dd x  \Big)  \nonumber \\
		& \qquad \qquad \qquad - \EER \Big( \int_0^{l_a} \int_{0+}^{l_a-} \Big| \frac{l_a - x - x^{\prime}}{l_a} \Big| \rho_a^b(x,t) \frac{\partial q_a^b(x^{\prime},t) }{\partial x}  \dd x^{\prime} \dd x \Big) \bigg].
	\end{align}
	Appeal to \autoref{Cor:flowBound} and noting that the third sum on the right-hand side is non-negative (and does not involve control variables), we have that there exists a constant $0 < \widetilde{K} < \infty$ such that
	\begin{align}
	& \EER \dot{V} \big( \boldsymbol{\rho}(t) \big)  \le \widetilde{K} -  \sum_{(a,b) \in \MM: \atop a \in \AXS} c_{a,b} \EER \Big( q_{a,\OUT}^b\big(p_{a,\OUT}(t)\big)  \rho_a^b(t) -  \frac{\dd A_a^b(t)}{\dd t} \rho_a^b(t) \Big) \nonumber \\
	& - \sum_{(a,b) \in \MM: \atop a \not\in \AXS} c_{a,b} \EER \bigg[ \int_0^{l_a} \Big| \frac{x}{l_a} \Big| \rho_a^b(x,t) q_{a,\OUT}^b\big(p_{a,\OUT}(t)\big) \dd x 
 - \int_0^{l_a} \Big| \frac{l_a - x}{l_a} \Big| \rho_a^b(x,t) q_{a,\IN}^b\big(p_{a,\IN}(t)\big) \dd x \bigg]. \label{E:PfTemp}
	\end{align}
	Rearranging terms and utilizing the properties of conditional expectation, we get
	\begin{align}
	\EER \dot{V} \big( \boldsymbol{\rho}(t) \big) &\le \widetilde{K} -  \sum_{(a,b) \in \MM: \atop a \in \AXS} \rho_a^b(t) \EER \bigg[ c_{a,b} q_{a,\OUT}^b\big(p_{a,\OUT}(t)\big)   -  c_{a,b} \frac{\dd A_a^b(t)}{\dd t} \bigg] \nonumber \\
	& - \sum_{(a,b) \in \MM: \atop a \not\in \AXS}  \int_0^{l_a} \rho_a^b(x,t) \bigg[ c_{a,b} \Big| \frac{x}{l_a} \Big| q_{a,\OUT}^b\big(p_{a,\OUT}(t)\big)  - c_{a,b} \Big| \frac{l_a - x}{l_a} \Big| q_{a,\IN}^b\big(p_{a,\IN}(t)\big) \bigg] \dd x.
	\end{align}
	By assumption, we have that there exist constants $0 < K^* < \infty$ and $\epsilon^* > 0$ associated with the policy $p^*(\cdot)$ such that
	\begin{align}
	\EER \dot{V} \big( \boldsymbol{\rho}(t) \big) \le K^* - \epsilon^* \Big( \sum_{\substack{(a,b) \in \MM: \\ a \in \AXS}} \EE \rho_a^b(t) + \sum_{\substack{(a,b) \in \MM: \\ a \not\in \AXS}} \EE \int_0^{l_a} \rho_a^b(x,t) \dd x \Big).
	\end{align}
	By definition, we have for each $t \ge 0$ that
	\begin{align}
	\epsilon^* \Big( & \sum_{\substack{(a,b) \in \MM: \\ a \in \AXS}} \EE \rho_a^b(t) + \sum_{\substack{(a,b) \in \MM: \\ a \not\in \AXS}} \EE \int_0^{l_a} \rho_a^b(x,t) \dd x \Big) \le \nonumber \\
	& \underset{p \in \mathcal{P}}{\max} ~ \sum_{(a,b) \in \MM: \atop a \in \AXS} \EE \rho_a^b(t) \EER \bigg[ c_{a,b} q_{a,\OUT}^b\big(p_{a,\OUT}(t)\big)   -  c_{a,b} \frac{\dd A_a^b(t)}{\dd t} \bigg] \nonumber \\
	& \qquad ~~ + \sum_{(a,b) \in \MM: \atop a \not\in \AXS} \EE \int_0^{l_a} \rho_a^b(x,t) \bigg[ c_{a,b} \Big| \frac{x}{l_a} \Big| q_{a,\OUT}^b\big(p_{a,\OUT}(t)\big)  - c_{a,b} \Big| \frac{l_a - x}{l_a} \Big| q_{a,\IN}^b\big(p_{a,\IN}(t)\big) \bigg] \dd x. \label{E:themPf1}
	\end{align}
	Hence, setting $K \equiv \max(K^*,\widetilde{K})$, we have by appeal to \autoref{thm:LyapunovDrift} that the control policy, $p(\cdot) \in \mathcal{P}$, which maximizes the right-hand side of \eqref{E:themPf1} for each $t \ge 0$ is also network stabilizing.  
	
	It remains to show that is equivalent to \eqref{E:BP}.  We have from \eqref{E:themPf1} that
	\begin{align}
	& \underset{p \in \mathcal{P}}{\arg \max} \sum_{(a,b) \in \MM: \atop a \in \AXS} c_{a,b} \rho_a^b(t) q_{a,\OUT}^b(p) + \sum_{(a,b) \in \MM: \atop a \not\in \AXS} c_{a,b} \int_0^{l_a} \rho_a^b(x,t) \EER \bigg[ \Big| \frac{x}{l_a} \Big| q_{a,\OUT}^b(p) 
	- \Big| \frac{l_a - x}{l_a} \Big| q_{a,\IN}^b(p) \bigg] \dd x \label{E:opt1}
	\end{align}
	for each $t \ge 0$ is network stabilizing.  The term corresponding to exogenous arrivals in \eqref{E:themPf1} was dropped from the optimization problem since it constitutes an additive constant to the problem.  We have also applied the principle of \textit{opportunistically maximizing an expectation} to drop the expectations in \eqref{E:themPf1} from the problem.
	
	The optimization problem \eqref{E:opt1} can be stated in terms of the intersection movements, by appeal to \eqref{E:boudary1} and \eqref{E:boudary2}:
	\begin{align}
		\underset{p \in \mathcal{P}}{\arg \max} & \sum_{(a,b) \in \MM: \atop a \in \AXS} c_{a,b} \rho_a^b(t) q_{a,b}(p) \nonumber \\
		& + \sum_{(a,b) \in \MM: \atop a \not\in \AXS} c_{a,b} \int_0^{l_a} \rho_a^b(x,t) \EER \bigg[ \Big| \frac{x}{l_a} \Big| q_{a,b}(p) 
		- \Big| \frac{l_a - x}{l_a} \Big| \sum_{c \in \Pi(a): \atop (c,a) \in \MM} \pi_{a,b}(t) q_{c,a}(p) \bigg] \dd x,  \label{E:opt2}
	\end{align}
	which upon re-arranging terms and the orders of summation and integration becomes
	\begin{align}
		\sum_{(a,b) \in \MM} \widetilde{w}_{a,b}(t) \EER q_{a,b}(p), \label{temp2}
	\end{align}
	where 
	\begin{align}
		\widetilde{w}_{a,b}(t) \equiv  c_{a,b} \rho_a^b(t)  -  \int_0^{l_b} \Big| \frac{l_b - x}{l_b} \Big| \sum_{\substack{c \in \Sigma(b): \\ (b,c) \in \MM }} c_{b,c} \pi_{b,c}(t) \rho_b^c(x,t) \dd x 
	\end{align}
	for $a \in \AXS$ and
	\begin{align}
		\widetilde{w}_{a,b}(t) \equiv  c_{a,b} \int_0^{l_a} \Big| \frac{x}{l_a} \Big| \rho_a^b(x,t) \dd x  -  \int_0^{l_b} \Big| \frac{l_b - x}{l_b} \Big| \sum_{\substack{c \in \Sigma(b): \\ (b,c) \in \MM }} c_{b,c} \pi_{b,c}(t) \rho_b^c(x,t) \dd x 
	\end{align}
	for $a \not\in \AXS$. Define $w_{a,b}(t) \equiv |\widetilde{w}_{a,b}(t)|$ for all $(a,b) \in \MM$, then \eqref{temp2} is bounded from above by
	\begin{align}
		\sum_{(a,b) \in \MM} w_{a,b}(t) \EER q_{a,b}(p), \label{temp3}
	\end{align}
	Since intersection movements do not interact across nodes instantaneously { \color{black}and $p_n \cap p_n = \emptyset$ for any $n,m \in \NX$ such that $n \ne m$}, the optimization problem naturally decomposes by intersection.  That is, maximizing \eqref{temp3} is equivalent to solving the $|\NX|$ problems
	\begin{align}
		\underset{p \in \mathcal{P}_n}{\arg \max} \sum_{(a,b) \in \MM_n} w_{a,b}(t) \EE^{\boldsymbol{\rho}(t)} q_{a,b}(p) \quad \forall n \in \NX.
	\end{align}
	This completes the proof.
\end{proof}

\section{Experiments}
\label{S:simulation}

\subsection{Real-world isolated intersection experiment}
In this section, we perform a comparison between PWBP and a real-world implementation of adaptive control using SCOOT. We perform the experiment on the isolated intersection depicted in \autoref{F:IP17}, which is the intersection of Hamdan Bin Mohammed Street - Fatima Bint Mubarak Street intersection in Abu Dhabi in the United Arab Emirates (UAE).  
\begin{figure}[h!]
	\centering
	
	\resizebox{0.45\textwidth}{!}{%
		\includegraphics{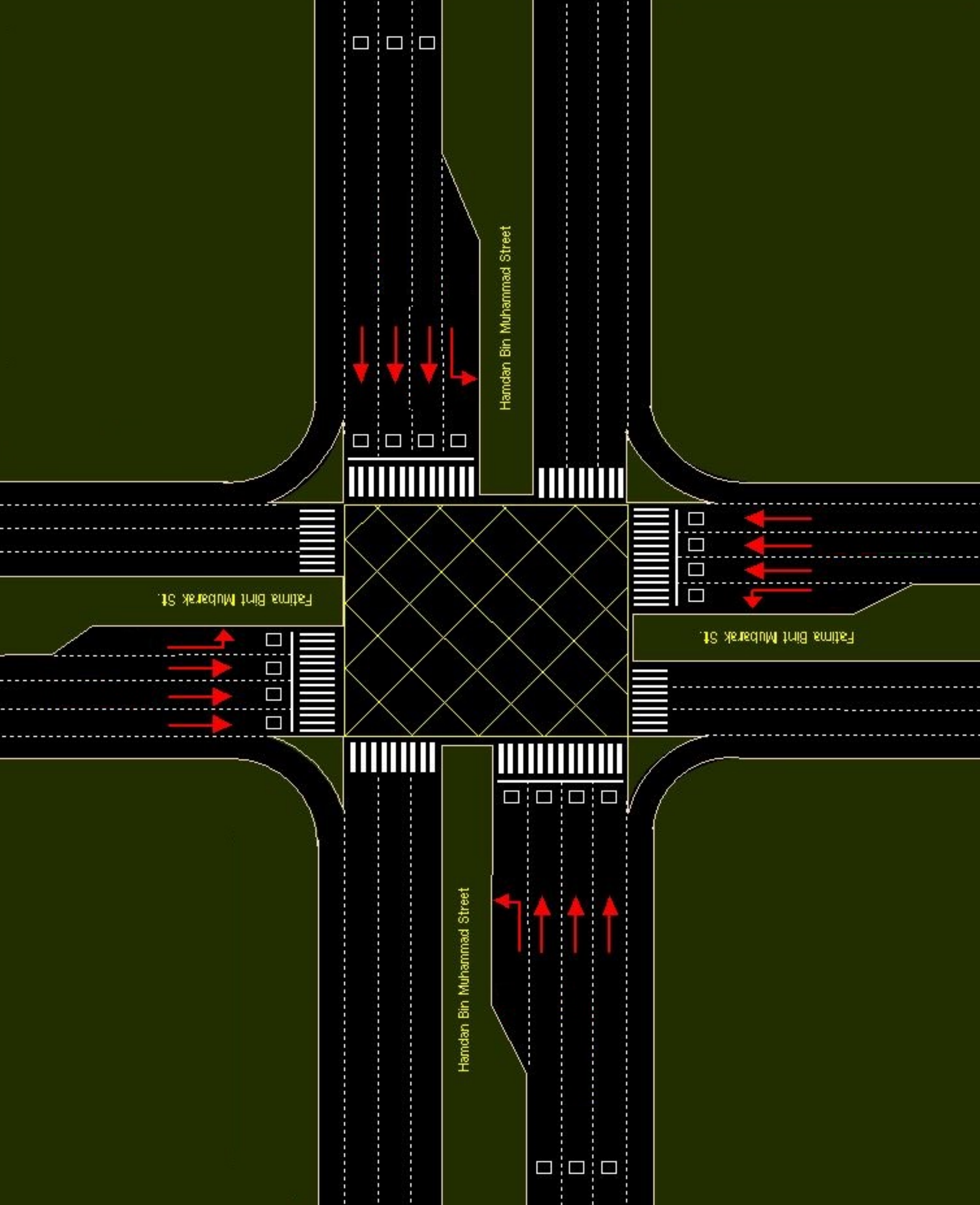} \hspace{0.2in}}
	
	\caption{Layout of the calibrated single intersection in Abu Dhabi}
	\label{F:IP17}
\end{figure}
We utilize \textit{high-resolution traffic data} obtained for all the loop detectors in each lane along with second-by-second signal status data. This allows us to build a microscopic traffic simulation model that serves as an accurate reconstruction of traffic at the intersection (both in terms of demands and signal status).  The data covers a 24-hour period on December 6, 2017, a typical working day.  The performance of PWBP compared to the real world adaptive control implementation at the intersection is depicted in \autoref{F:SCOOT}.
\begin{figure}[h!]
	\centering
	
	\resizebox{0.6\textwidth}{!}{%
		\includegraphics{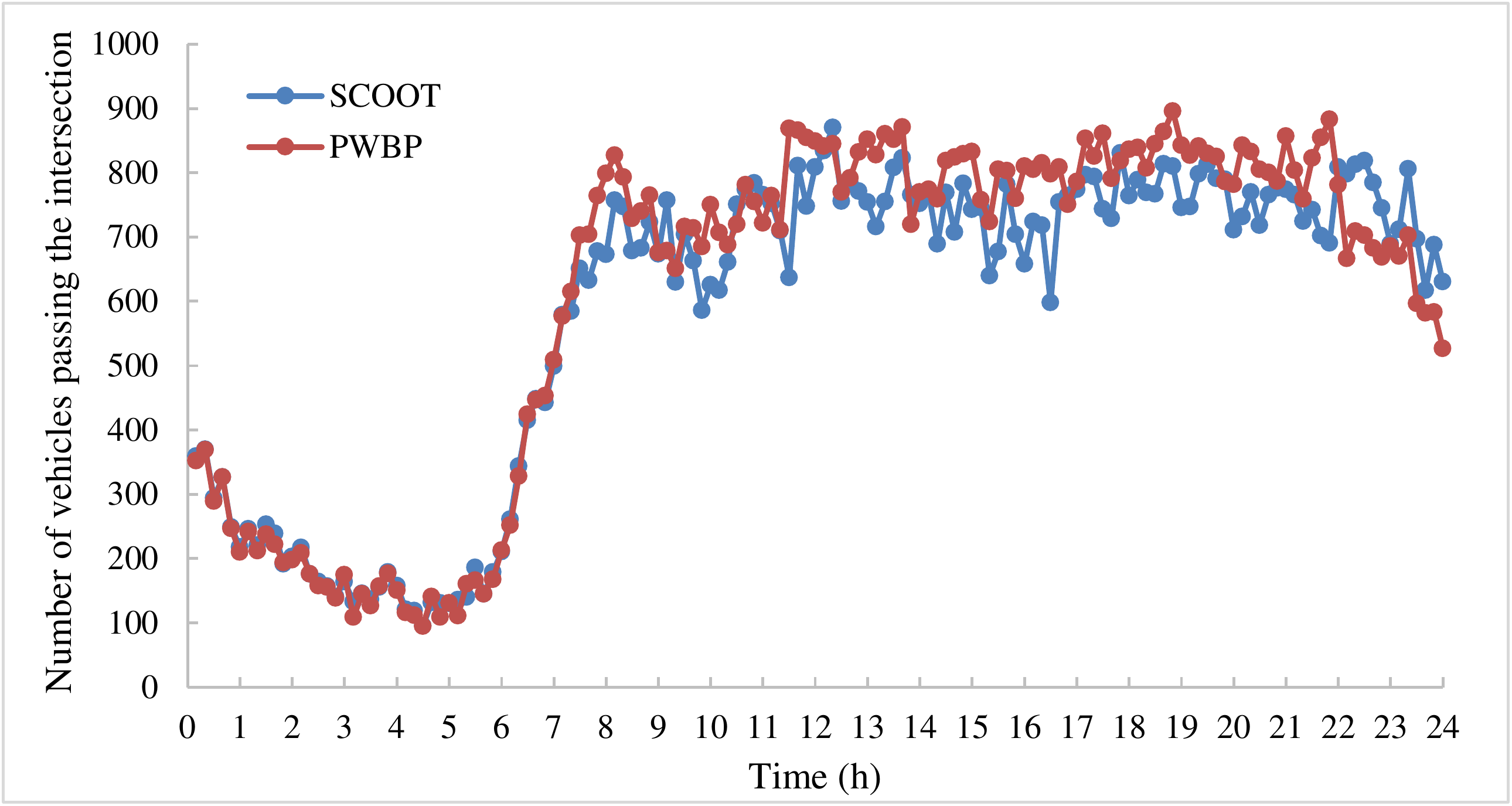} \hspace{0.2in}}
	
	{ \small	(a) }
	
	\resizebox{0.6\textwidth}{!}{%
		\includegraphics{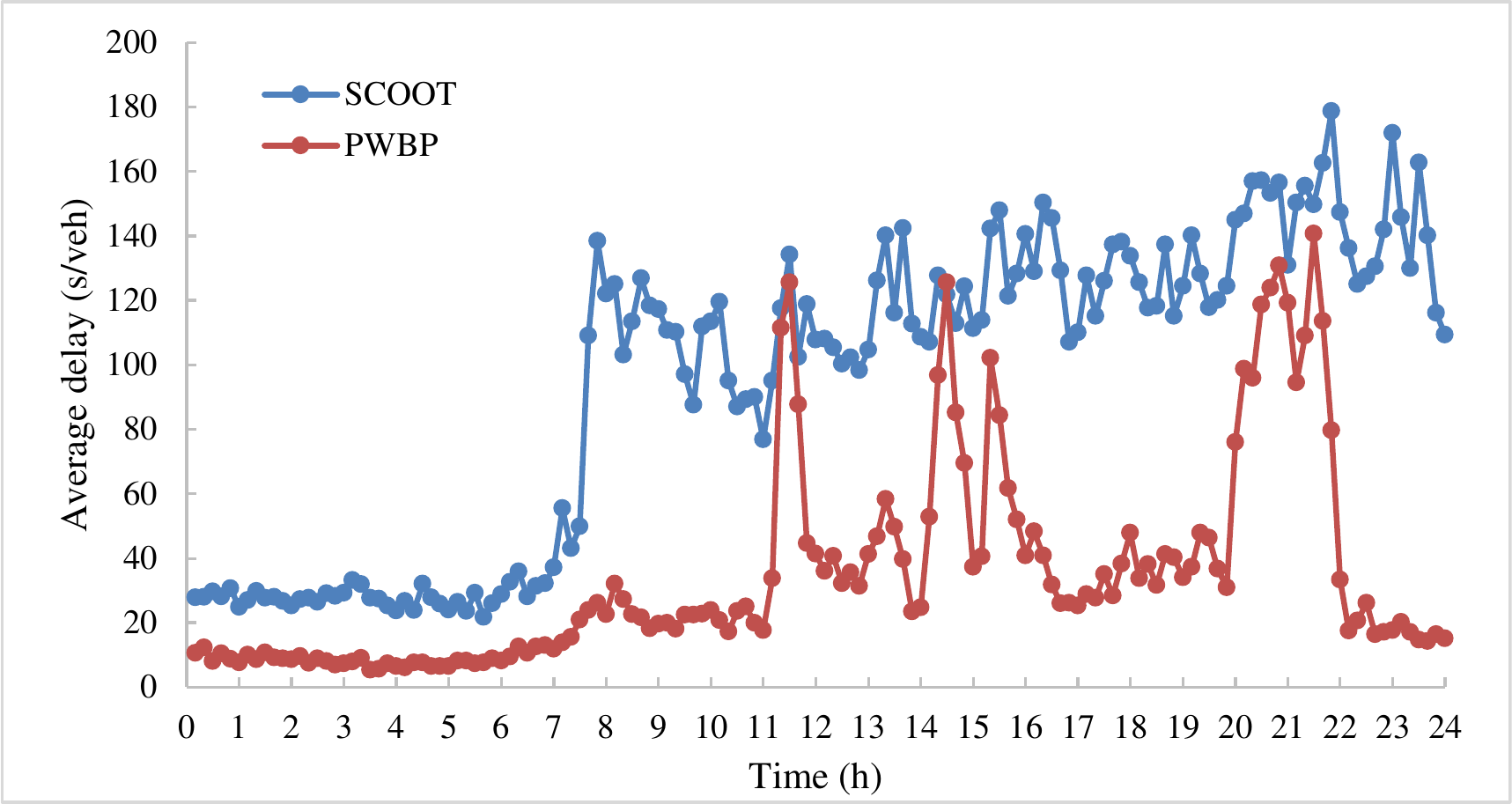} \hspace{0.2in}}
	
	{ \small	(b) }
	
	\caption{Comparison between SCOOT and PWBP in (a)flow, (b)delay.}
	\label{F:SCOOT}
\end{figure}
\autoref{F:SCOOT}a compares the intersection throughput for the two control policies and \autoref{F:SCOOT}b depicts the average delay profiles associated with the two control strategies.  We see modest improvements in throughput, but substantial improvements in delay associated with PWBP.  This hints at significant improvements at the individual intersection level. 
The figures demonstrate that PWBP outperforms SCOOT substantially in terms of delay under both low and high demands. The average delay over the entire day for SCOOT is 95 seconds, while the average delay of PWBP is only 35 seconds.  In addition, when demand is high, PWBP control has higher throughput.

\subsection{Network experiments}
\textbf{\fontfamily{cmss}\selectfont Network description}. We utilize a microscopic traffic simulation network of a part of the city of Abu Dhabi consisting of eleven signalized intersections but also containing unsignalized intermediate junctions. The network layout is shown in \autoref{F:network}.
\begin{figure}[h!]
	\centering
	\resizebox{0.45\textwidth}{!}{%
		\includegraphics{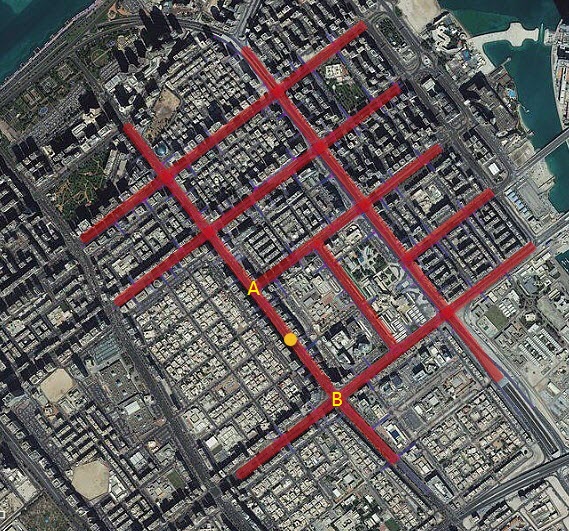}}
	
	\caption{Simulation network in Abu Dhabi.} 
	\label{F:network}
\end{figure}
We compare PWBP control with three other control policies: fixed time, standard BP control, and CABP control.  The fixed timing plans are optimized and include optimal offsets (i.e., signal coordination).  BP, CABP, and PWBP are all implemented using a software interface.  To simplify the experiments, we utilize a uniform demand at the boundaries, which we vary to gauge the capacity region of the network.  Using a uniform (average) demand level allows us to use a single number (namely the demand) as a way to gauge the capacity region.

\bigskip

\textbf{\fontfamily{cmss}\selectfont Average network delay and network capacity region}. \autoref{F:CR} shows the total network delay under different demand scenarios (ranging from 500 to 1800 veh/h on average) for BP, CABP, and PWBP using two types of phasing schemes: one with four phases (``4-phase'' scheme) and a scheme with eight phases (``8-phase scheme'').
\begin{figure}[h!]
	\centering
	
	\resizebox{.9\textwidth}{!}{%
		\includegraphics{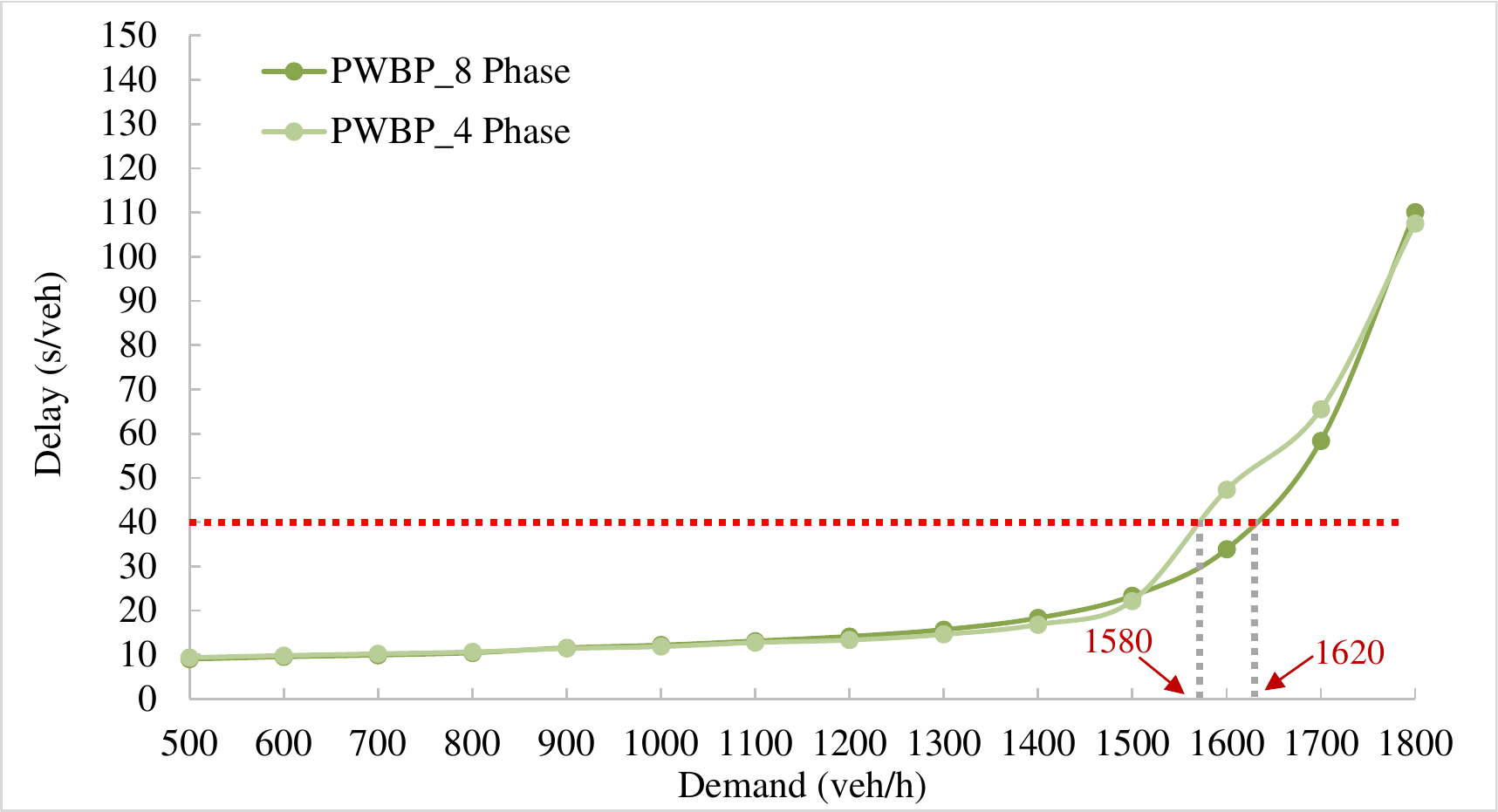} \hspace{0.2in}
		\includegraphics{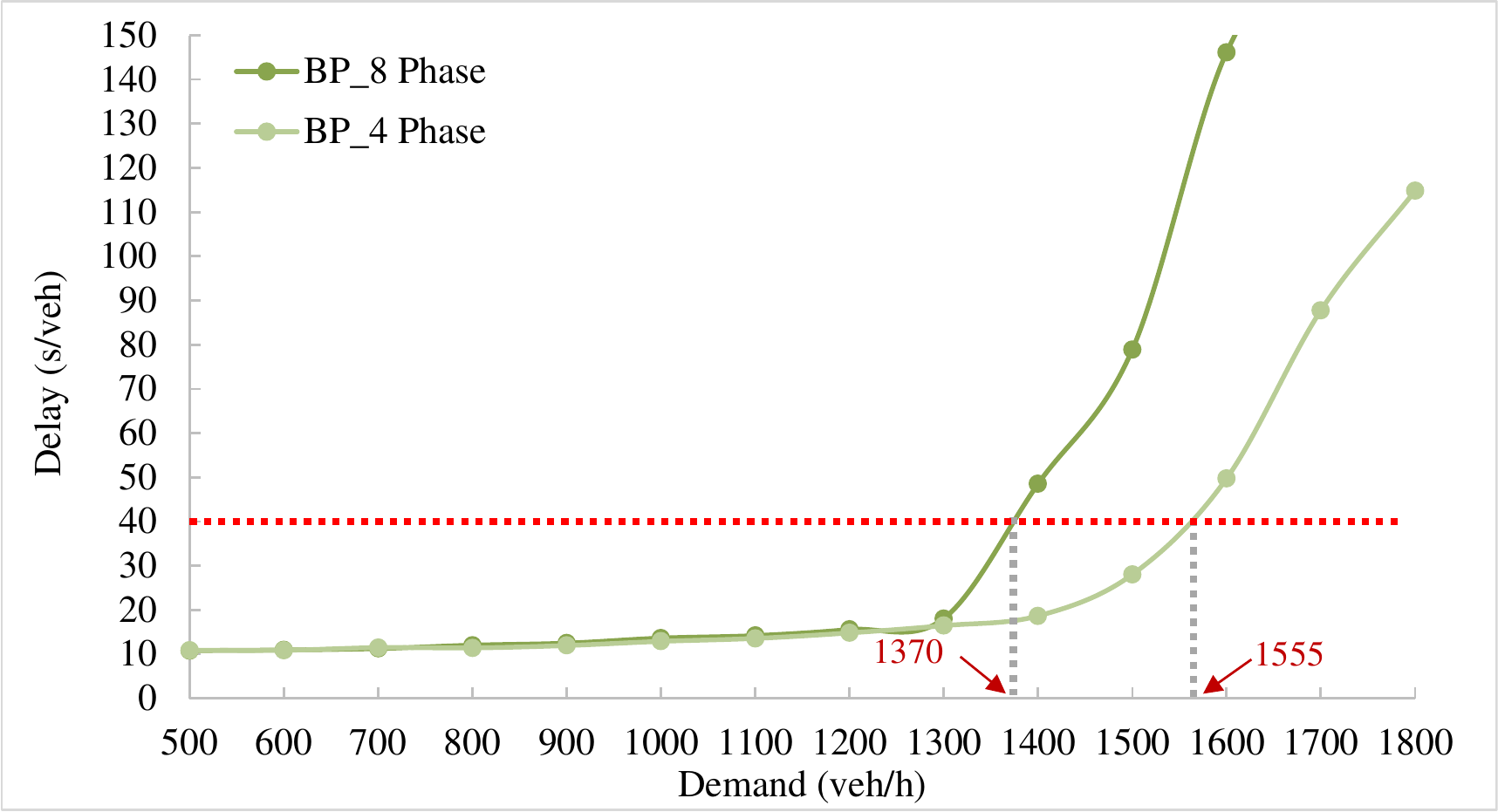}}
	
	(a) \hspace{2.2in} (b)
	
	\resizebox{.9\textwidth}{!}{%
		\includegraphics{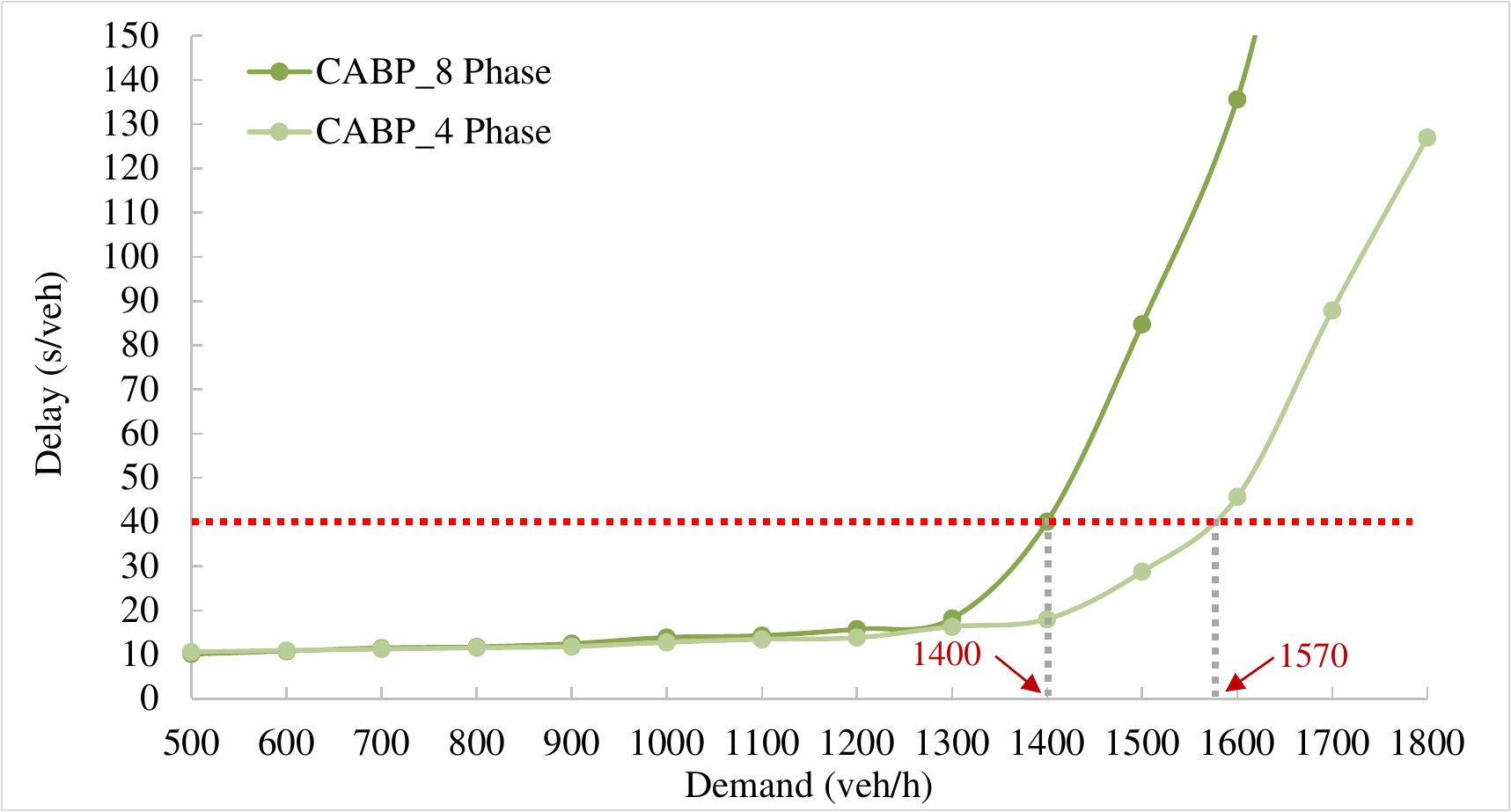} \hspace{0.2in}
		\includegraphics{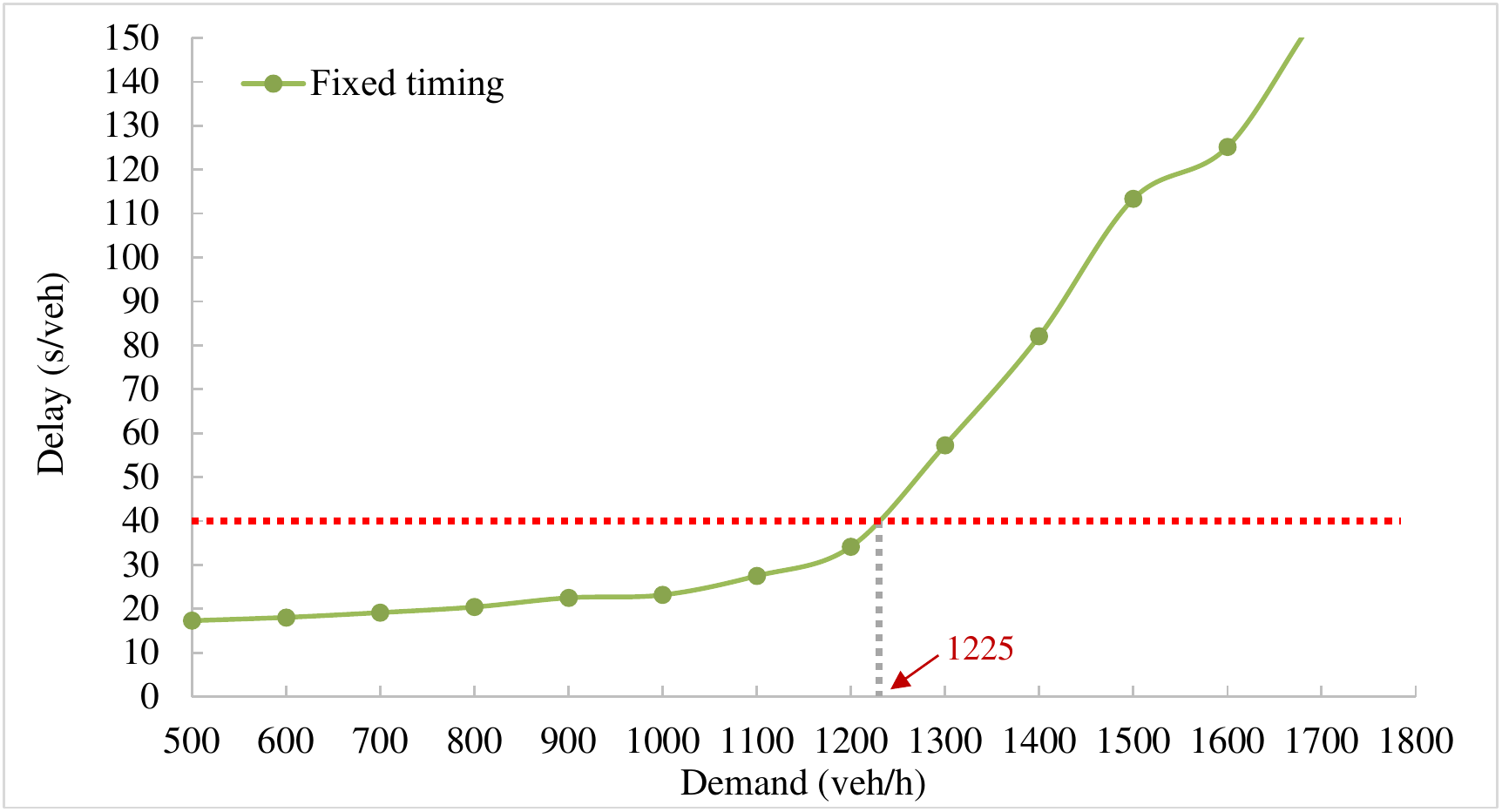}}
	
	(c) \hspace{2.2in} (d) 
	
	\caption{Delay patterns at varying demand levels for different control policies.} 
	\label{F:CR}
\end{figure}
The 4-phase scheme includes phases 1-4 in \autoref{F:phases}, while the 8-phase scheme is all eight phases in \autoref{F:phases}.  We observe that 40 s/veh is a threshold delay, beyond which the delay increases dramatically.  We can hence treat 40 s/veh as indicative of reaching the boundary of the capacity region. From \autoref{F:CR}, with the 8-phase scheme, we see that delays begin to increase rapidly at a higher average demand levels for the PWBP: 1620 veh/h for the 8-phase scheme vs. 1580 veh/h for the 4-phase scheme.  However, this is not the case for BP and CABP control, since they do not distinguish left-turning and through queues, which results in blocking at the points where roads widen (left-turn lane addition). This indicates that BP and CABP have a wider capacity region using a 4-phase scheme compared to the 8-phase scheme. All subsequent experiments use an 8-phase scheme with PWBP and 4-phase schemes with BP and CABP.  The demands at which delays begin to increase quickly for fixed signal timing, BP, CABP, and PWBP are 1225, 1555, 1570, and 1620 veh/h, respectively.  \autoref{F:delay_no_inc} shows a comparison of network delays for the four control policies under varying demands. 
\begin{figure}[h!]
	\centering
	\resizebox{0.6\textwidth}{!}{%
		\includegraphics{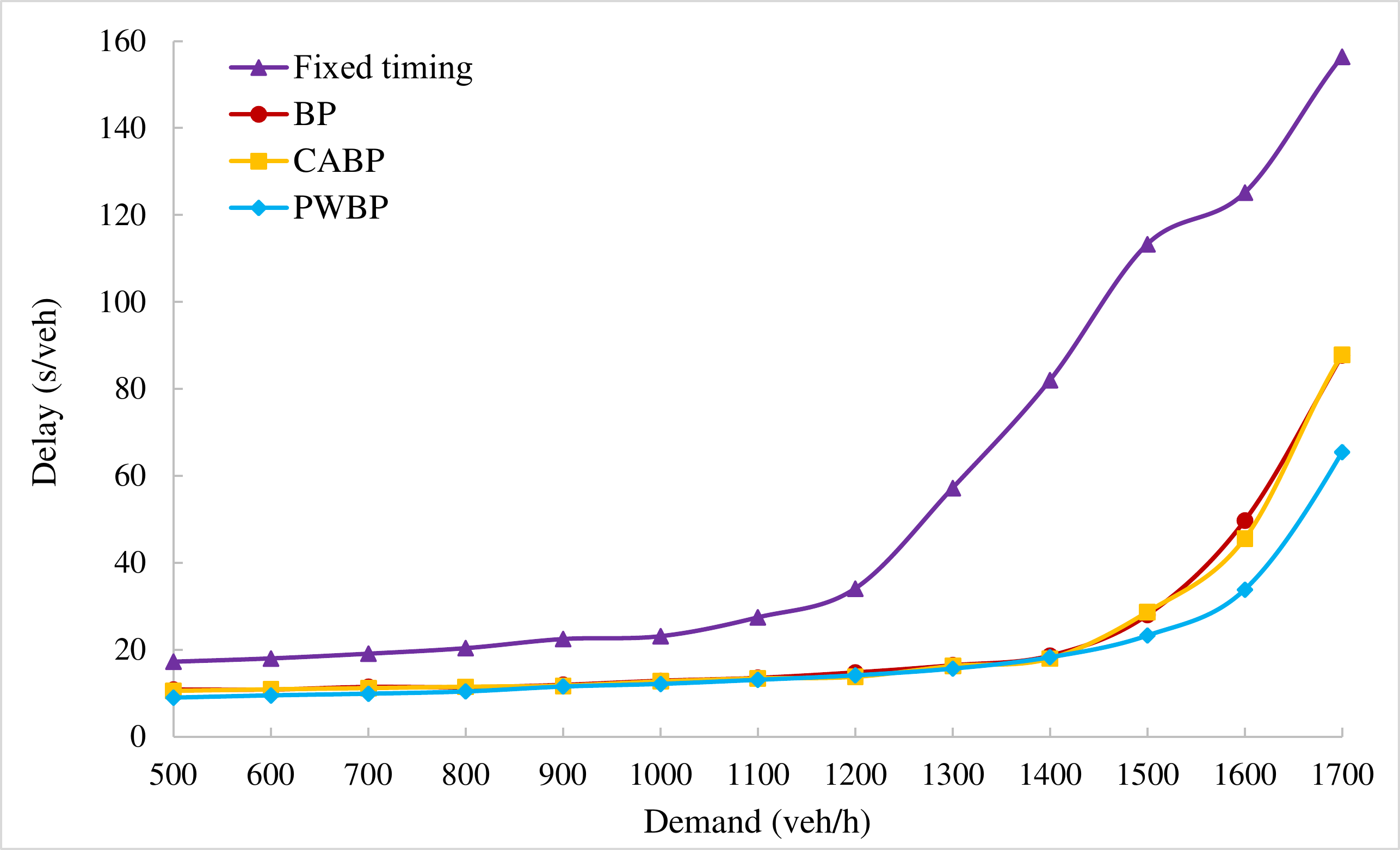}}
	
	\caption{Network delays associated with different control policies.} 
	\label{F:delay_no_inc}
\end{figure}

\medskip

\textbf{\fontfamily{cmss}\selectfont Congestion propagation}.  In the following experiments, we set the demand levels to the deterioration bounds of the control policies and compare how congestion levels propagate over time.  Since the deterioration bounds for BP and CABP are close, we just use CABP's bound (1570 veh/h); we, hence, compare three demand scenarios.   \autoref{F:speed1} -- \ref{F:speed2} show how the speeds of all vehicles within the network are distributed under demand levels 1225, 1570 and 1620 veh/h.   
\begin{figure}[h!]
	\centering
	\resizebox{.9\textwidth}{!}{%
		\includegraphics{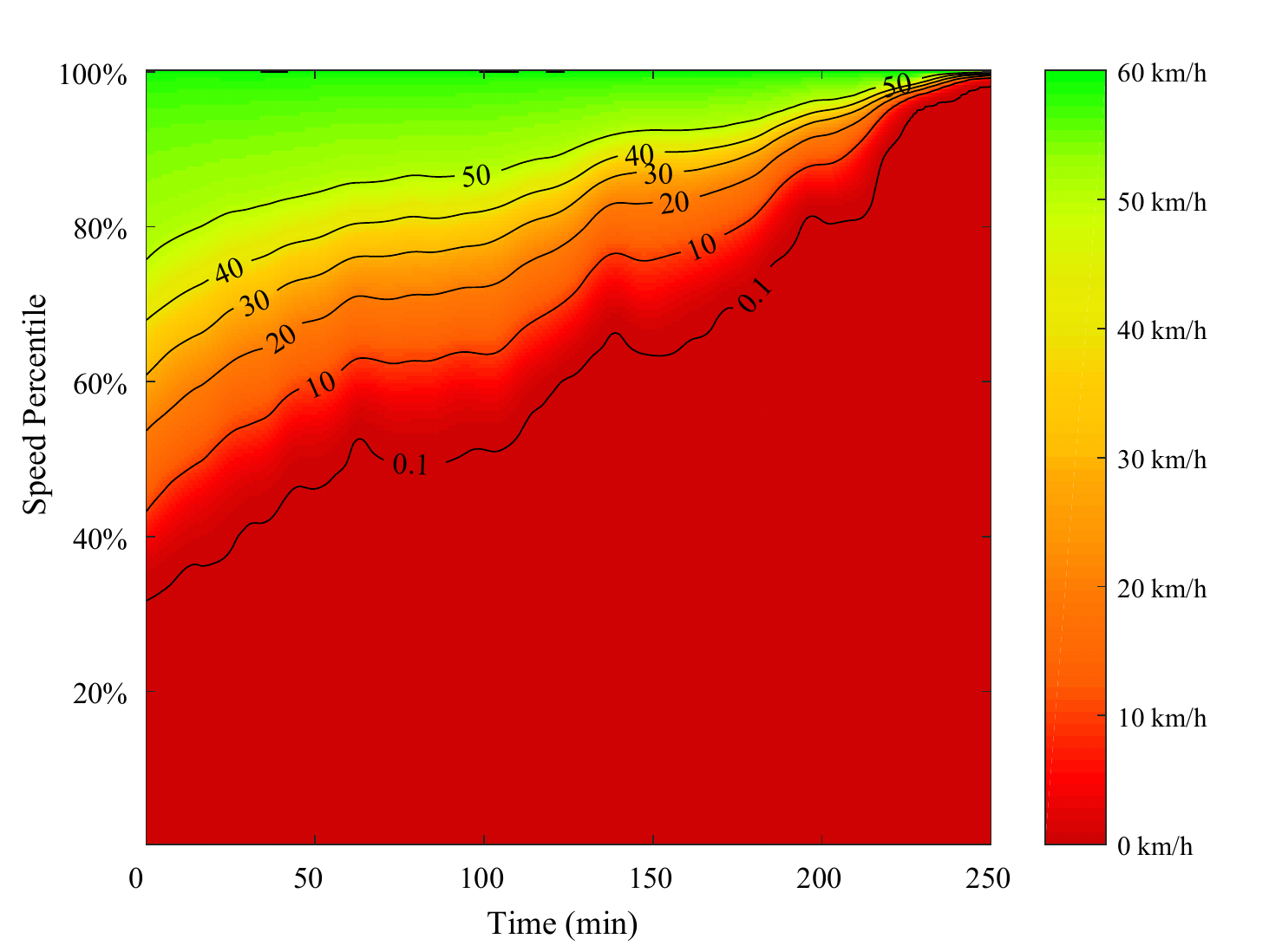} \hspace{0.1in}
		\includegraphics{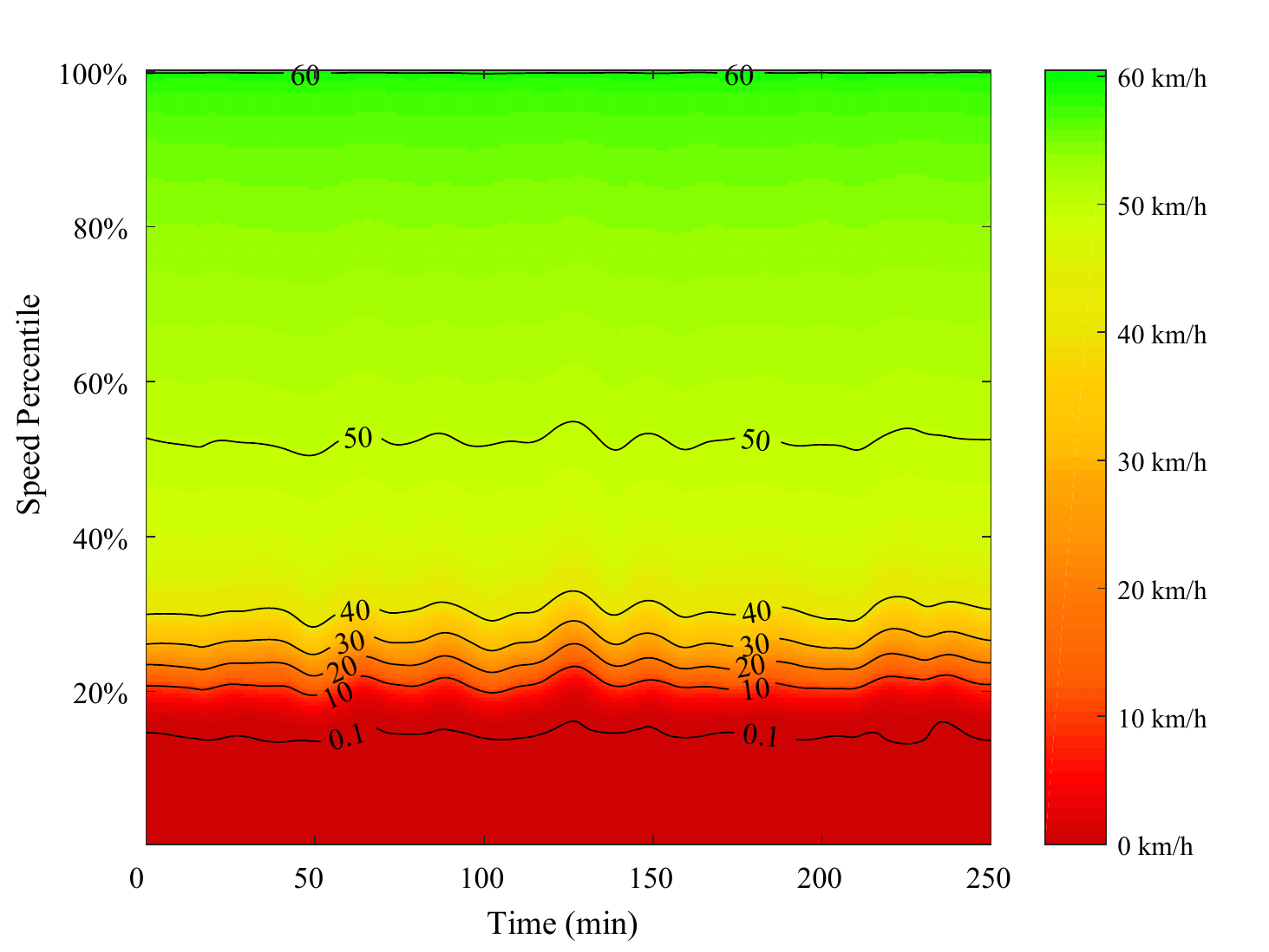}}
	
	{ \small(a) FT@1225vph \hspace{1.7in} (b) BP@1225vph }
	
	\resizebox{.9\textwidth}{!}{%
		\includegraphics{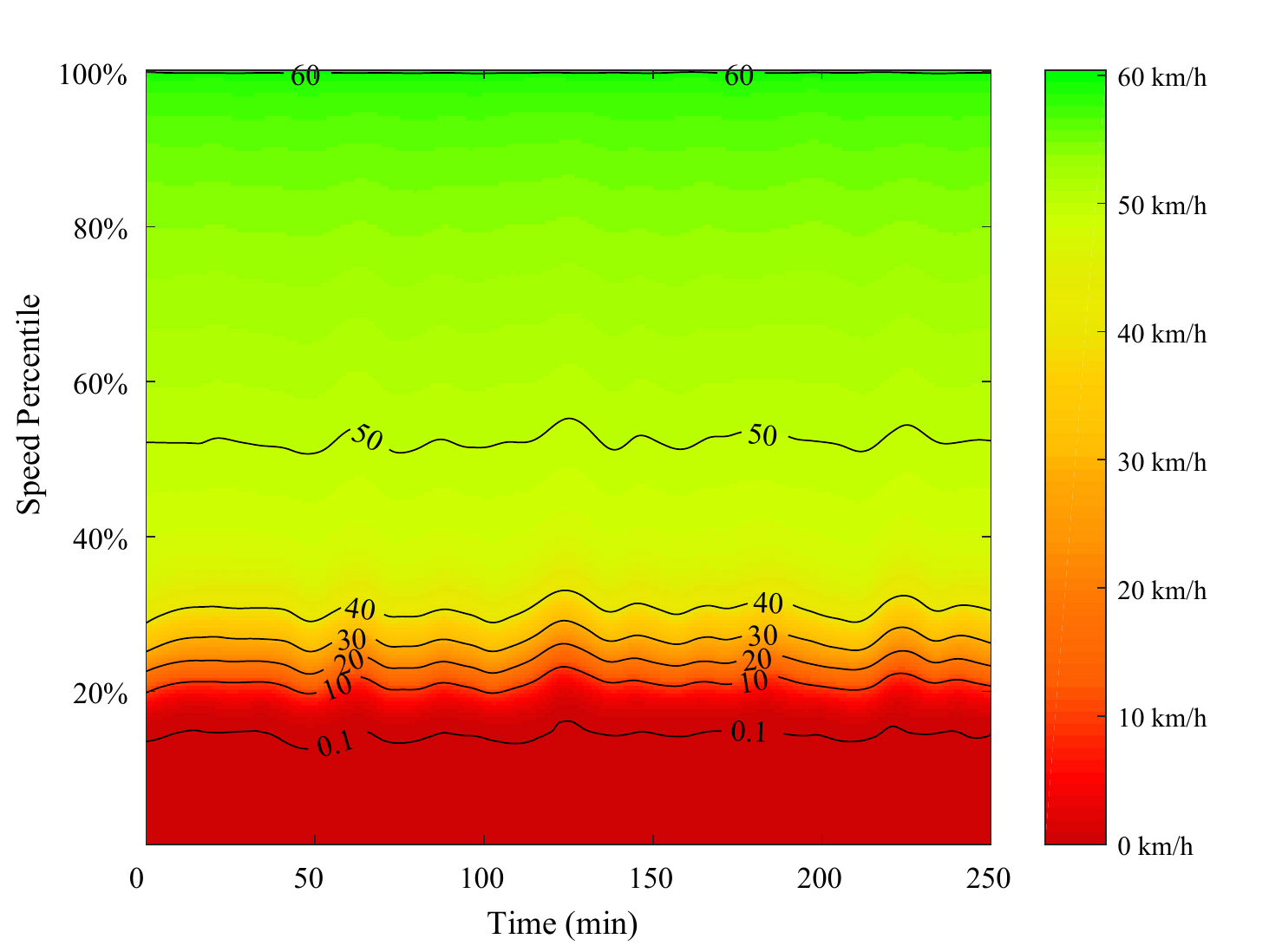} \hspace{0.1in}
		\includegraphics{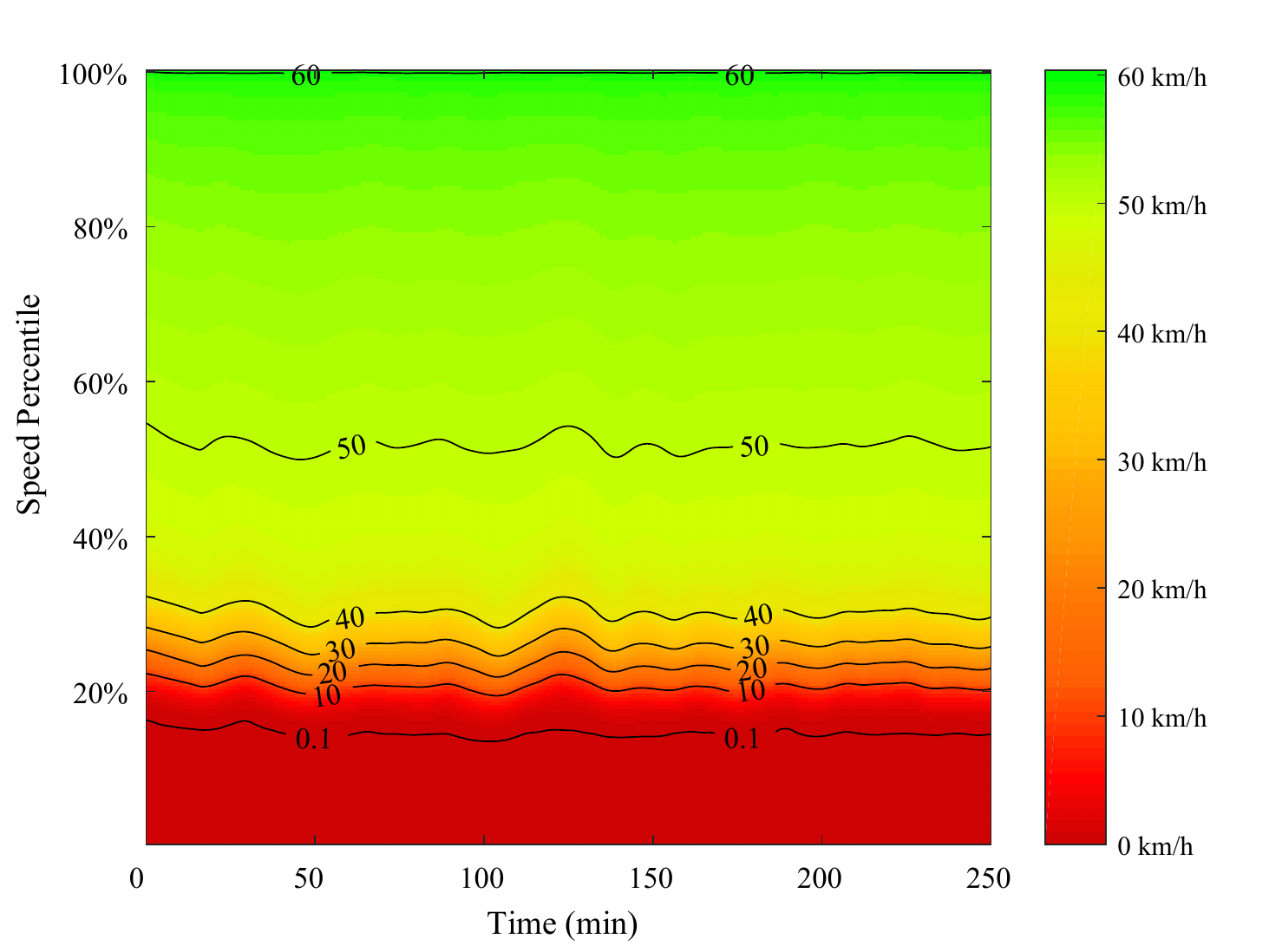}}
	
	{ \small(c) CABP@1225vph \hspace{1.4in} (d) PWBP@1225vph}
	
	\resizebox{.9\textwidth}{!}{%
		\includegraphics{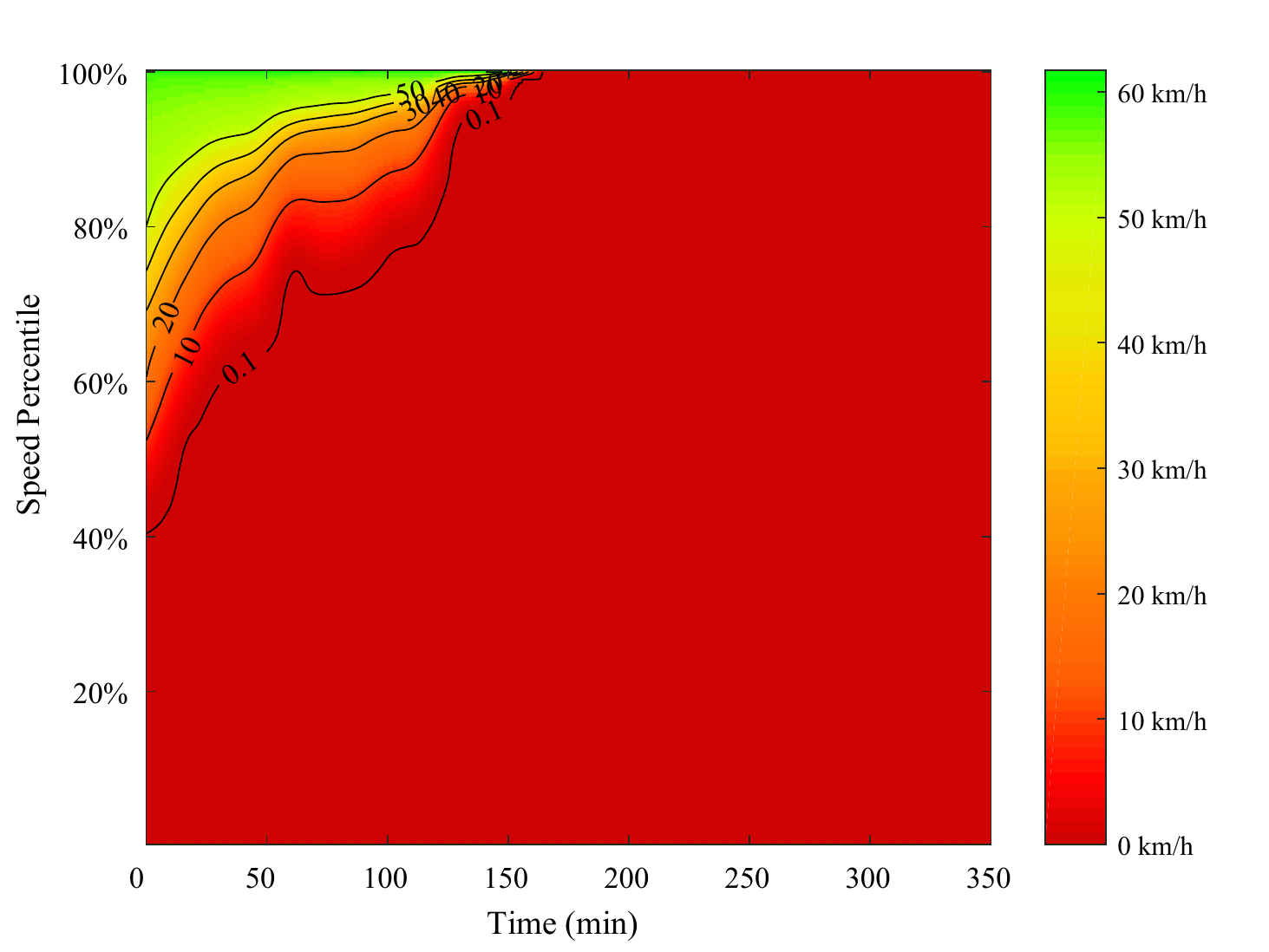} \hspace{0.1in}
		\includegraphics{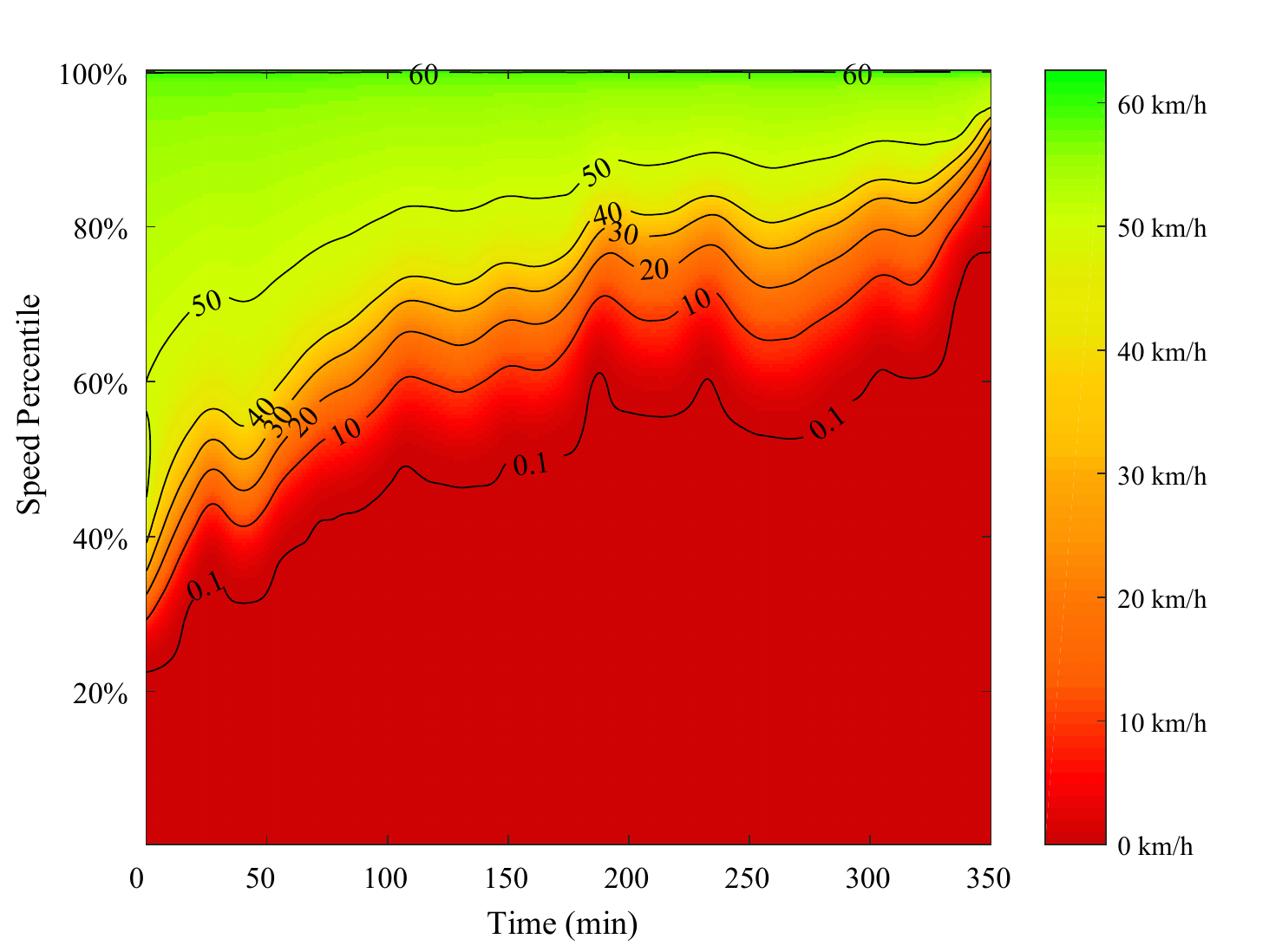}}
	
	{ \small(e) FT@1570vph \hspace{1.7in} (f) BP@1570vph }
	
	\caption{Network speed evolution, (a) fixed timing under a demand level of 1225 veh/h, (b) BP under a demand level of 1225 veh/h, (c) CABP under a demand level of 1225 veh/h, and (d) PWBP under a demand level of 1225 veh/h, (e) fixed timing under a demand level of 1570 veh/h, (f) BP under a demand level of 1570 veh/h.} 
	\label{F:speed1}
\end{figure}
\begin{figure}[h!]
	\centering	
	
	\resizebox{.9\textwidth}{!}{%
		\includegraphics{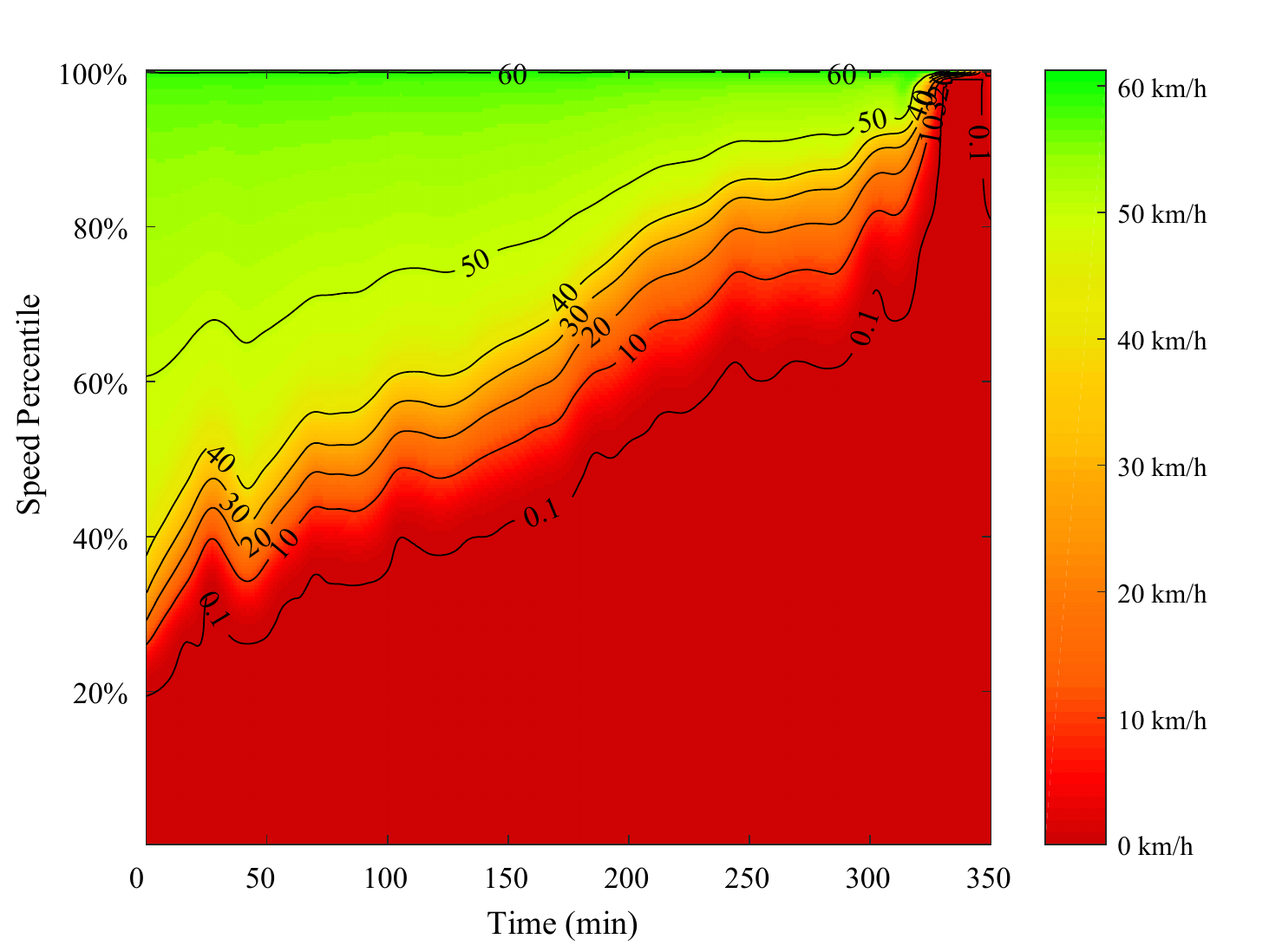} \hspace{0.1in}
		\includegraphics{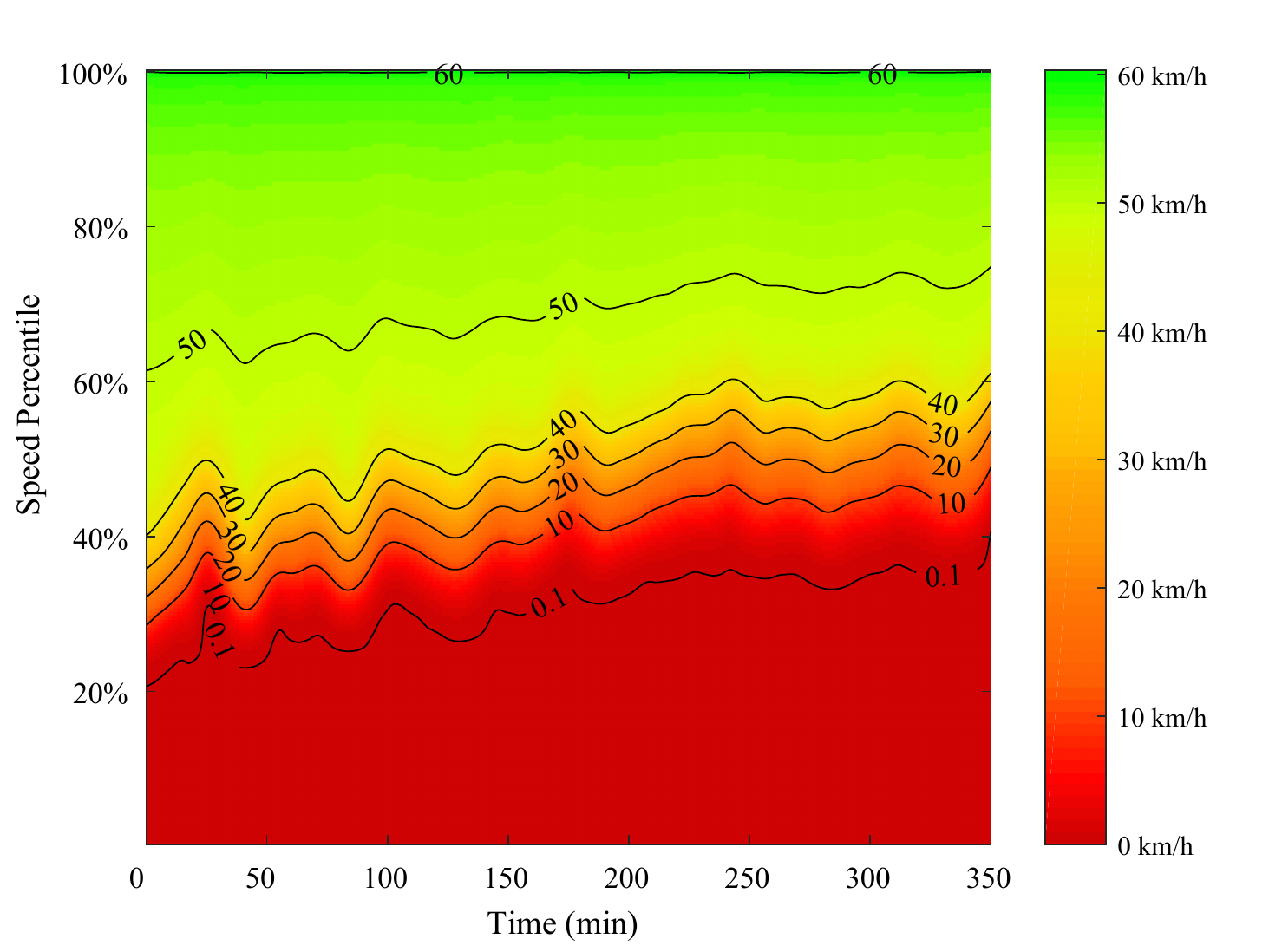}}
	
	{ \small(a) CABP@1570vph \hspace{1.4in} (b) PWBP@1570vph}
	
		\resizebox{.9\textwidth}{!}{%
		\includegraphics{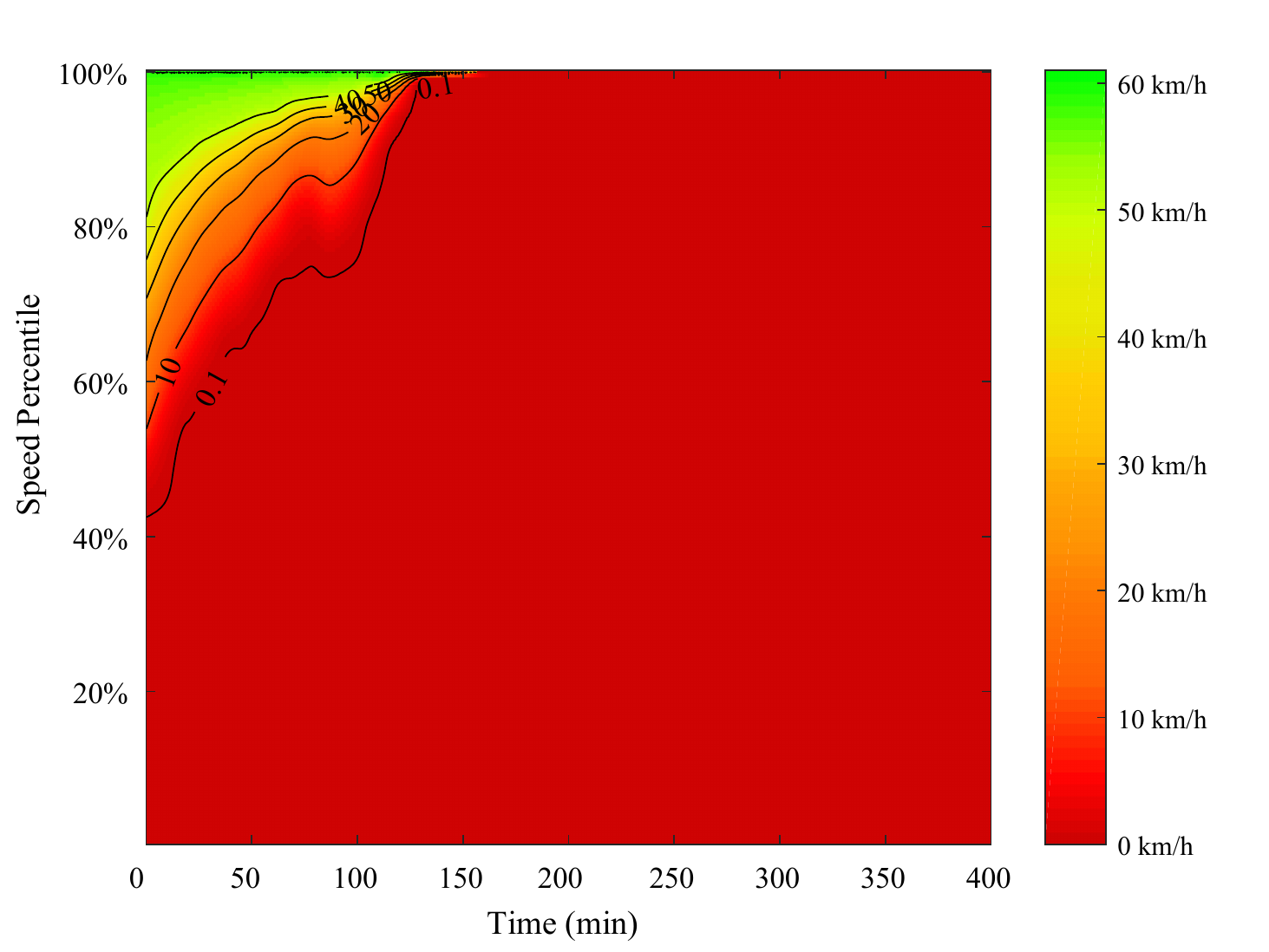} \hspace{0.1in}
		\includegraphics{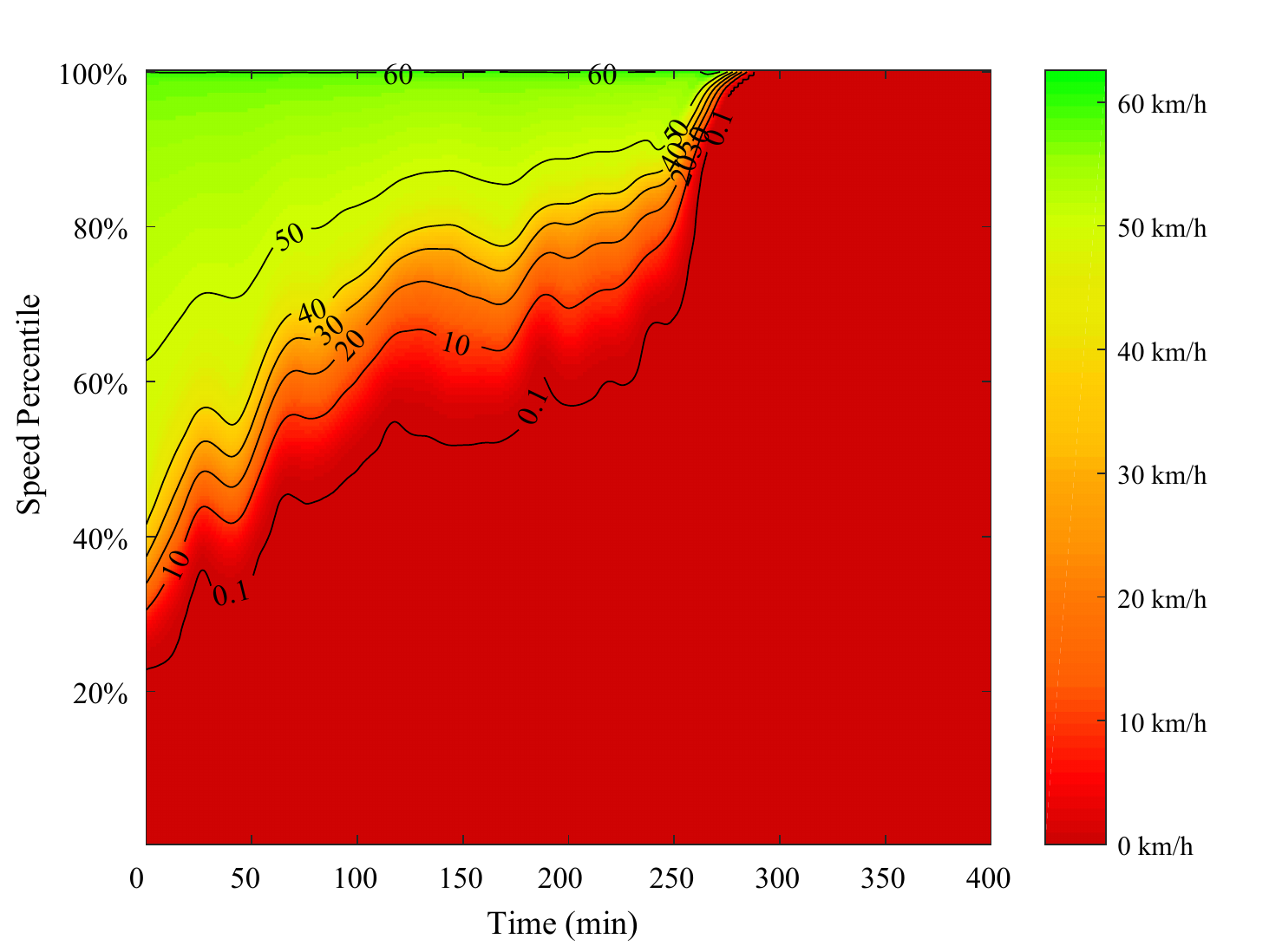}}
	
	{ \small(c) FT@1620vph \hspace{1.7in} (d) BP@1620vph }
	
	\resizebox{.9\textwidth}{!}{%
		\includegraphics{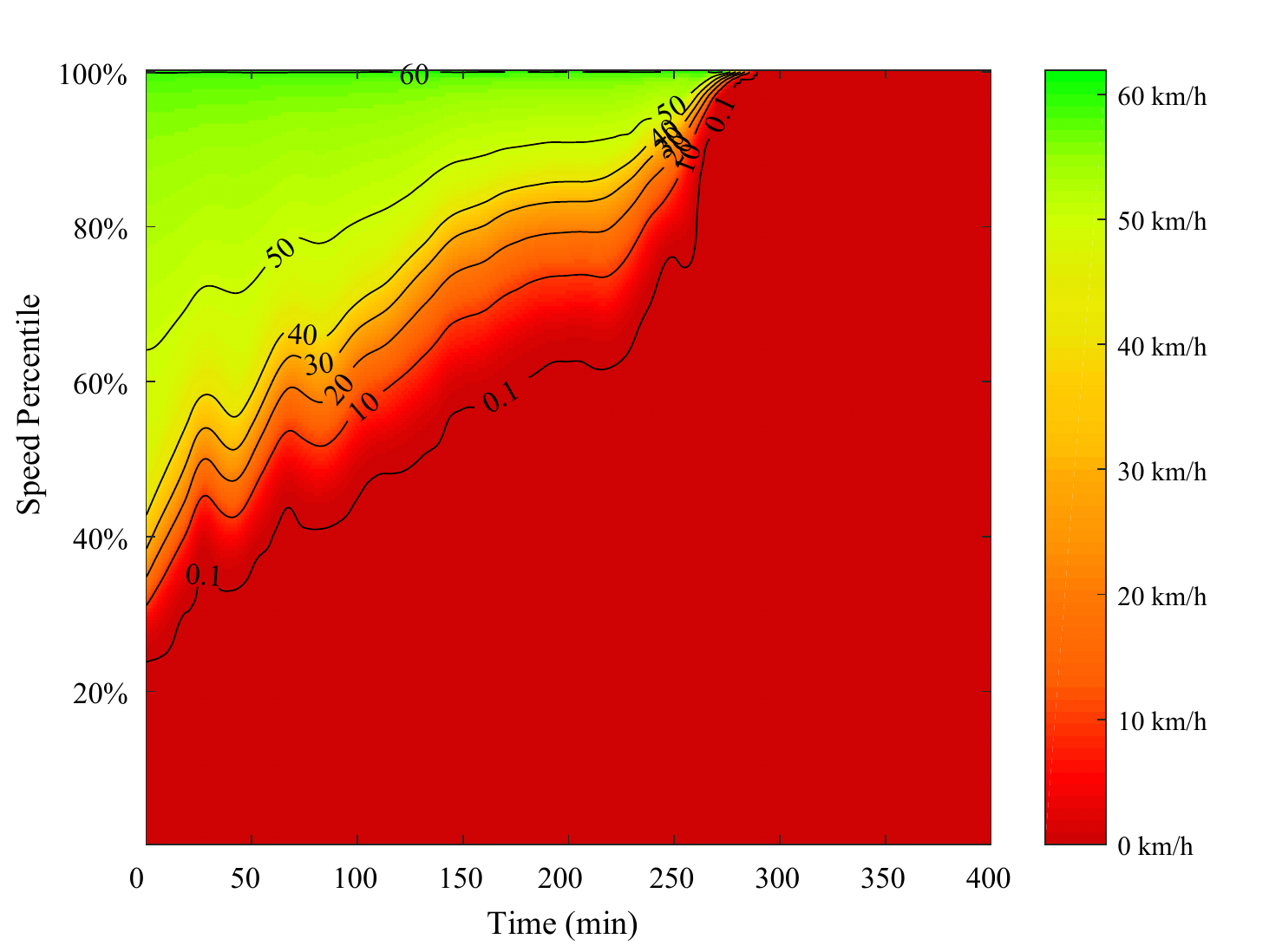} \hspace{0.1in}
		\includegraphics{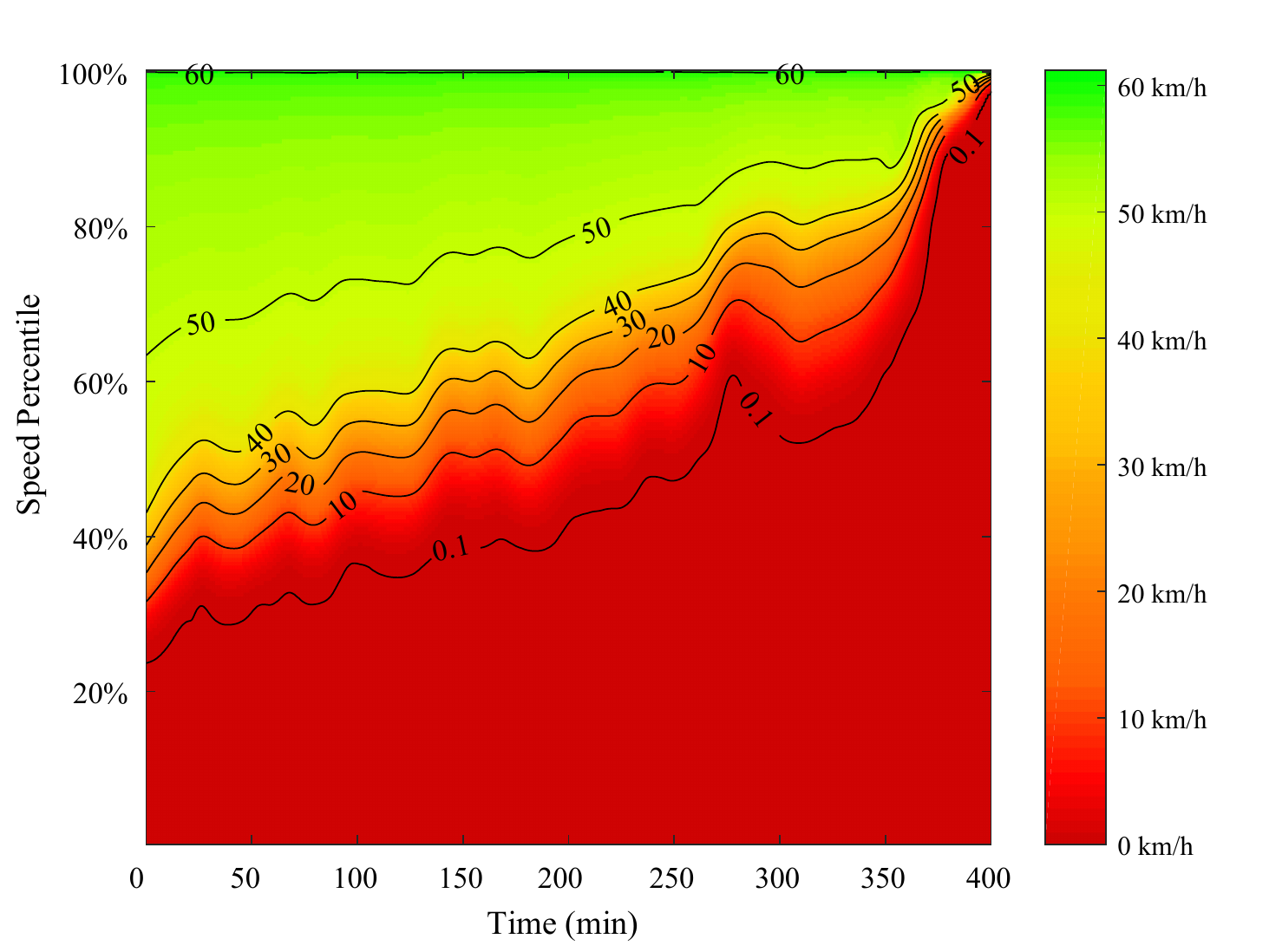}}
	
	{ \small (e) CABP@1620vph \hspace{1.4in} (f) PWBP@1620vph }
	
	\caption{Network speed evolution, (a) CABP under a demand level of 1570 veh/h, (b) PWBP under a demand level of 1570 veh/h, (c) fixed timing under a demand level of 1620 veh/h, and (d) BP under a demand level of 1620 veh/h, (e) CABP under a demand level of 1620 veh/h, (f) PWBP under a demand level of 1620 veh/h.} 
	\label{F:speed2}
\end{figure}
The horizontal axes in these figures are time and the vertical axes are percentage of vehicles traveling at or below the color-coded speeds.  Under the different demand levels, the network eventually becomes grid-locked (at different levels for the different control policies). Specifically, it takes about four hours until total network gridlock under a fixing timing plan when the demand reaches 1225 veh/h, under BP and CABP it takes approximately six hours (at 1570 veh/h) until gridlock, and for PWBP, it takes approximately seven hours.  This indicates that PWBP is more resilient than the other policies. \autoref{F:VehNo} shows how the total number of vehicles (stuck) in the network evolves with time.

\begin{figure}[h!]
	\centering
	\resizebox{0.65\textwidth}{!}{%
		\includegraphics{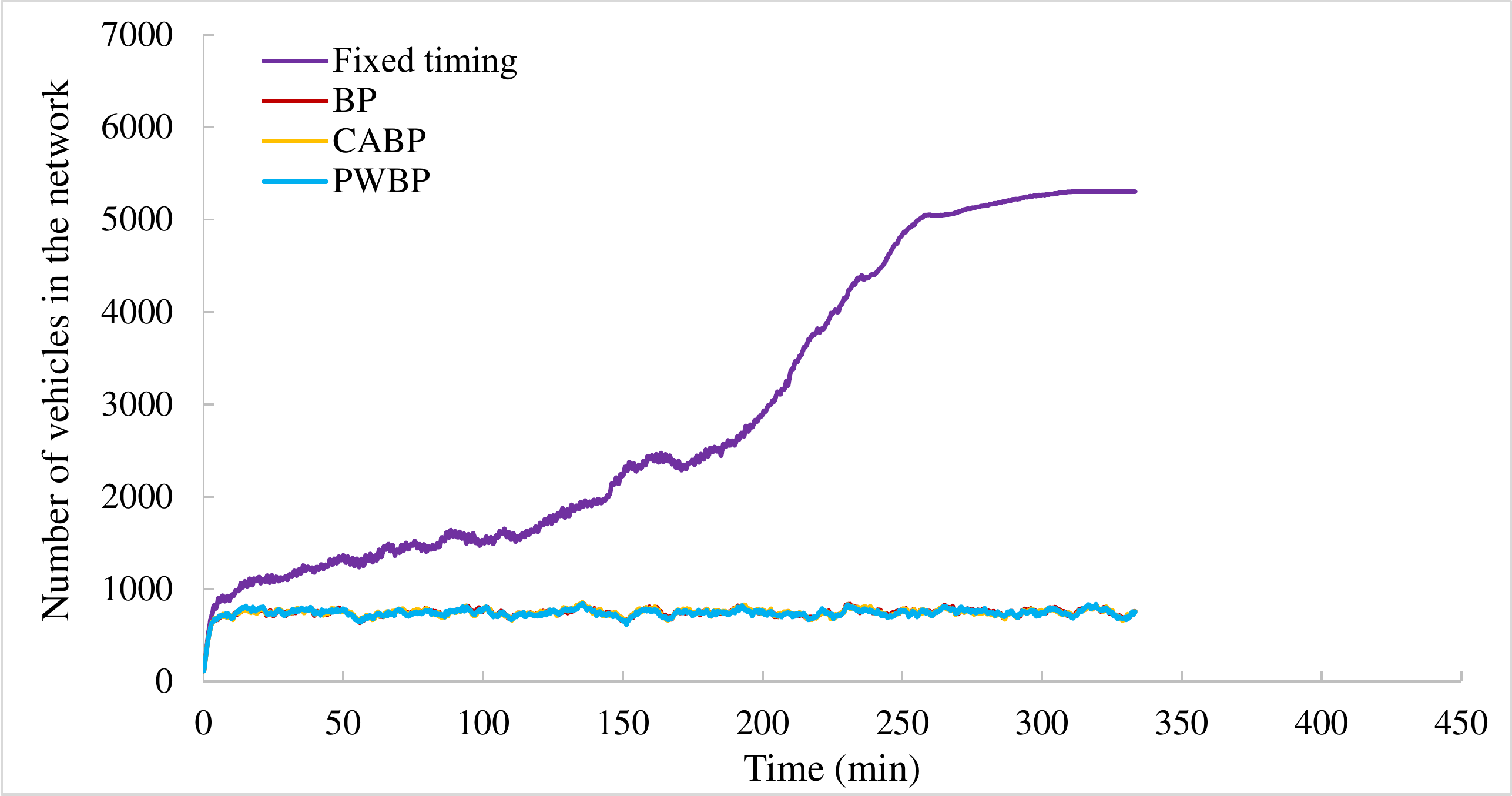}}
		
{ \small	(a) }

	\resizebox{0.65\textwidth}{!}{%
		\includegraphics{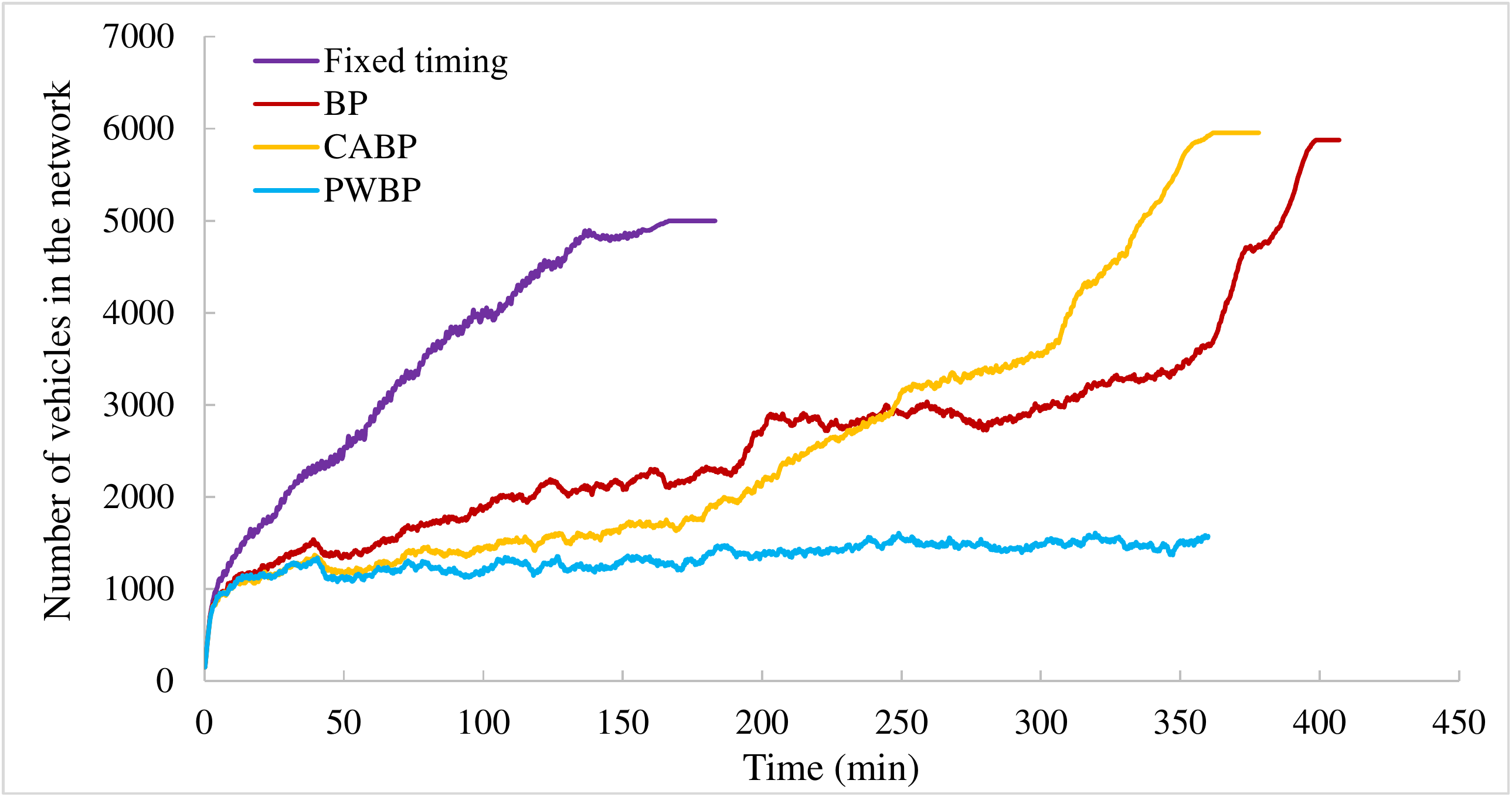}}
		
{ \small	(b) }
	
	\resizebox{0.65\textwidth}{!}{%
		\includegraphics{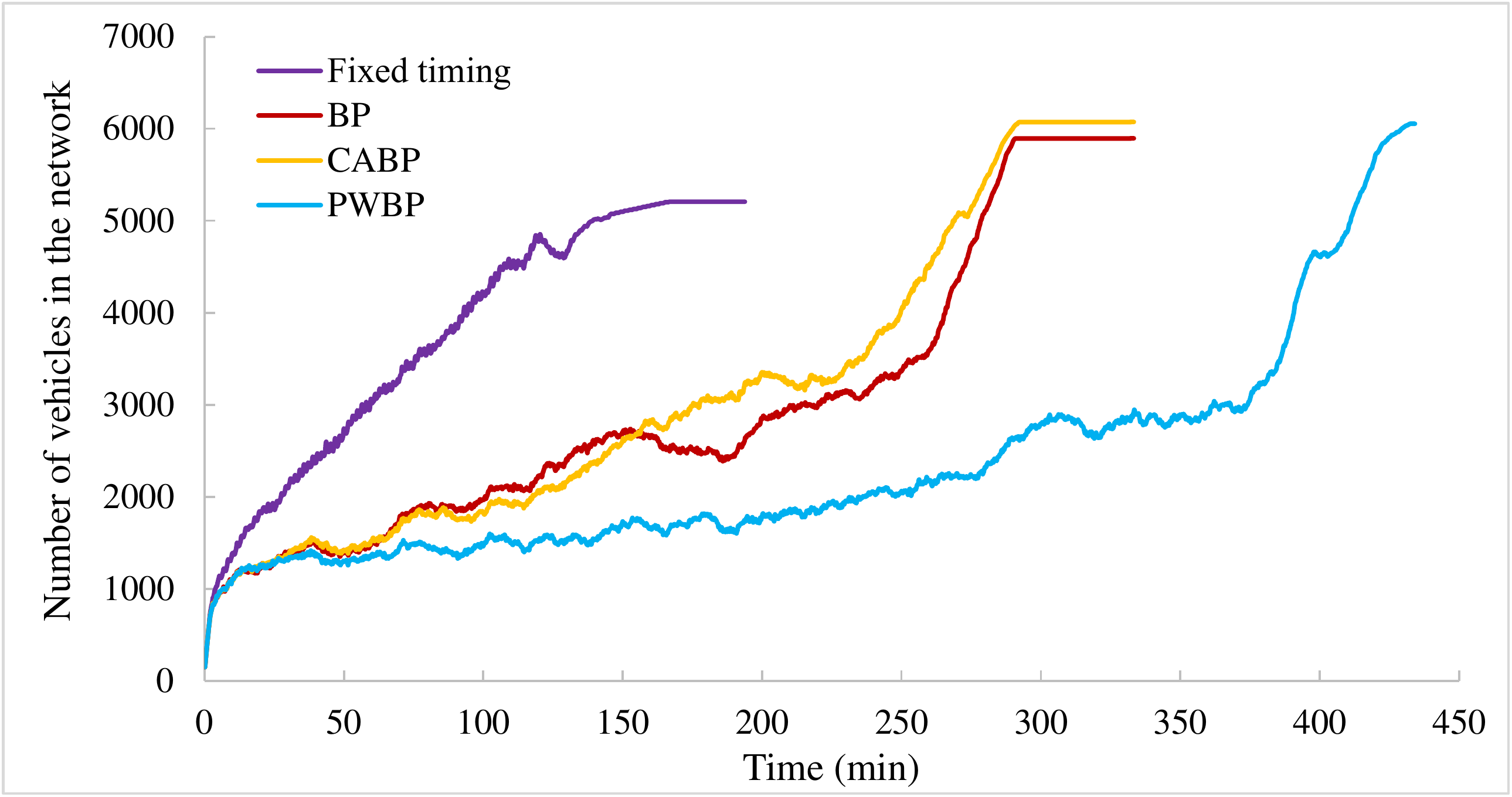}}
		
{ \small	(c) }
		
	\caption{Evolution of total numbers vehicles in the network under different control policies and demand levels of (a) 1225 veh/h, (b) 1570 veh/h, and (c) 1620 veh/h.} 
	\label{F:VehNo}
\end{figure}

\medskip

\textbf{\fontfamily{cmss}\selectfont Recoverability from congestion}. \autoref{F:recover} shows how different control policies recover from congestion. 
\begin{figure}[h!]
	\centering
	\resizebox{0.7\textwidth}{!}{%
		\includegraphics{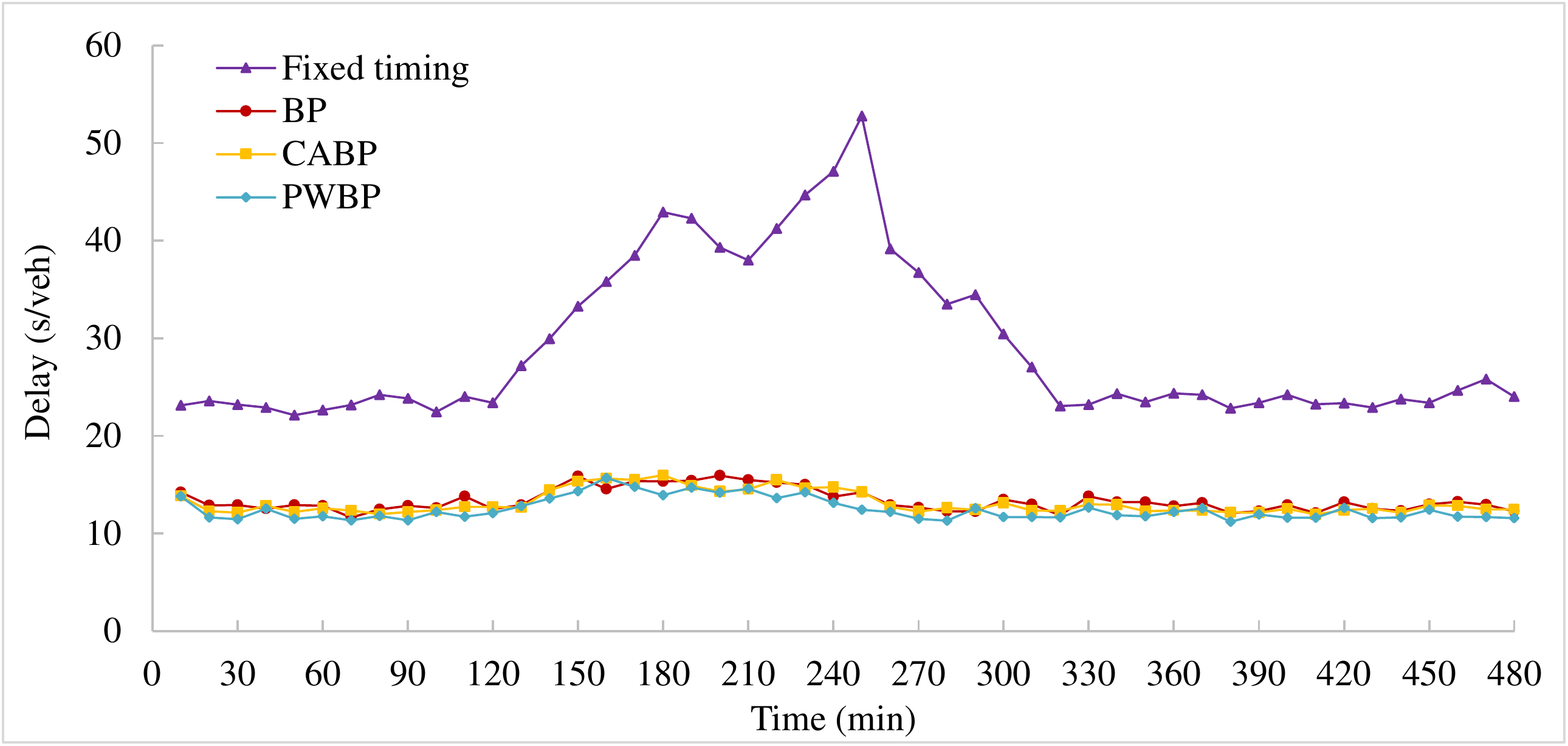}}
	
	{ \small	(a)	demand@1225vph	}
	
	\resizebox{.7\textwidth}{!}{%
		\includegraphics{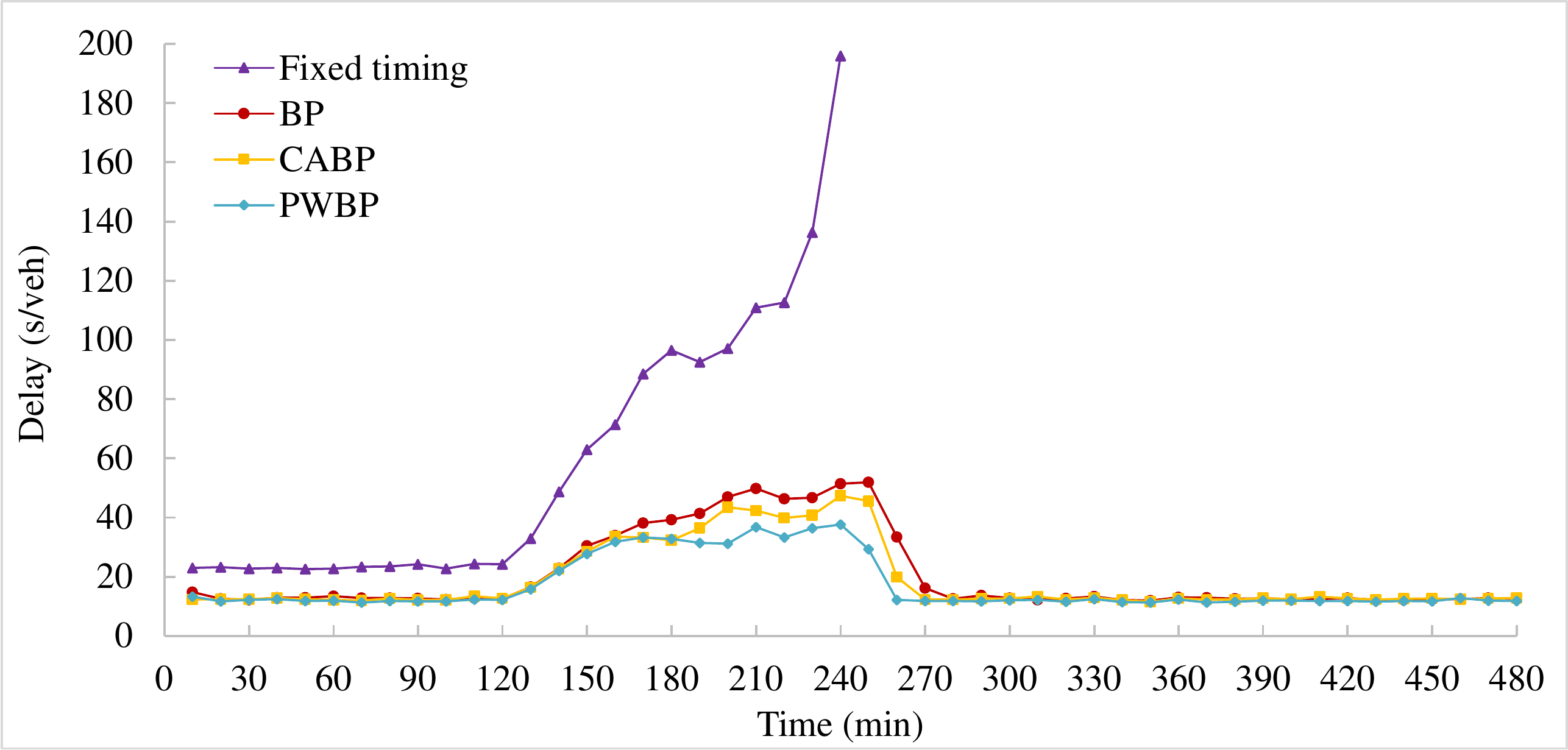}}
	
	{ \small	(b)	demand@1570vph	}
	
	\resizebox{.7\textwidth}{!}{%
		\includegraphics{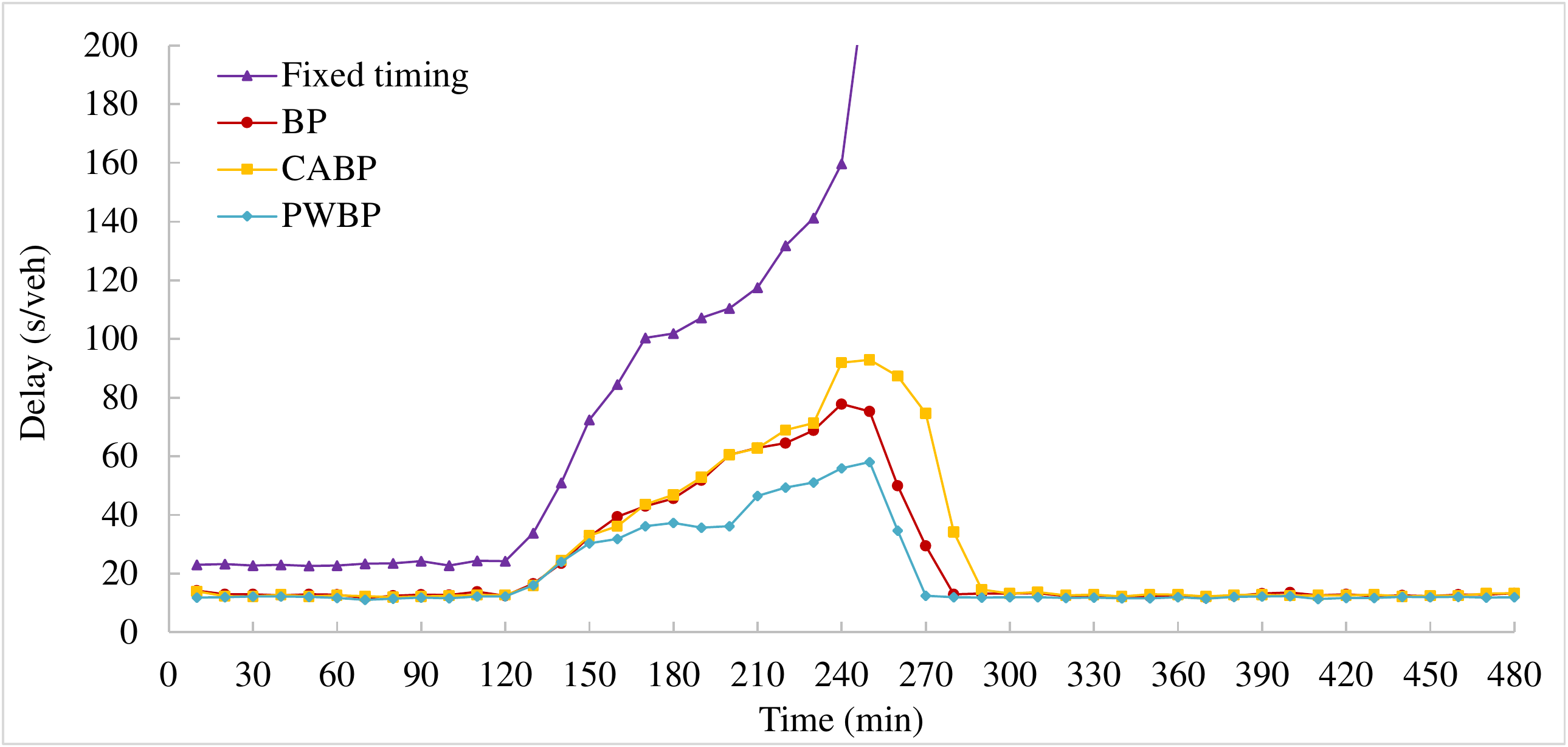}}
	
	{ \small	(c)	demand@1620vph	}
	
	\caption{Average network delay under varying peak period demands.} 
	\label{F:recover}
\end{figure}
The total simulation time is eight hours, the time interval from $t=$ 120 min to $t=$ 240 min is set as a congested period, during which demand levels are set to the deterioration bounds. We set a demand of 1000 veh/h for the remainder of the eight-hour simulation time.  \autoref{F:recover}a, b and c only differ in the demand levels during the congested period. The congested period demand levels are 1225, 1570 and 1620 veh/h in Fig.\ref{F:recover}a, b and c, respectively. 
According to \autoref{F:recover}, for all tested scenarios, PWBP outperforms the other three control policies in terms of both delay and recovery time.  Even when the peak demand reaches 1620 veh/h, PWBP only needs 30 min to recover from the congestion, while fixed timing needs about 90 min to recover with a peak demand of 1225 veh/h.  Note that when the peak demand reaches 1570 and 1620 veh/h, the delay levels under fixed timing becomes too high and hence cannot be shown in \autoref{F:recover}b and c. We also see that using fixed timing, the network does \textit{not} eventually recover from congestion.

\medskip

\textbf{\fontfamily{cmss}\selectfont Response to an incident}. We investigate the performance of PWBP in the presence of an incident located at the yellow spot in \autoref{F:network}.  The incident is located half-way between intersections A and B, along a 3-lanes arc. We test scenarios where one lane and two lanes are blocked for a duration of one and two hours, and under different demand levels. \autoref{F:incident1lane} shows the results for one-lane blocked cases when demand is 1500 veh/h. 
\begin{figure}[h!]
	\centering
	\resizebox{0.65\textwidth}{!}{%
		\includegraphics{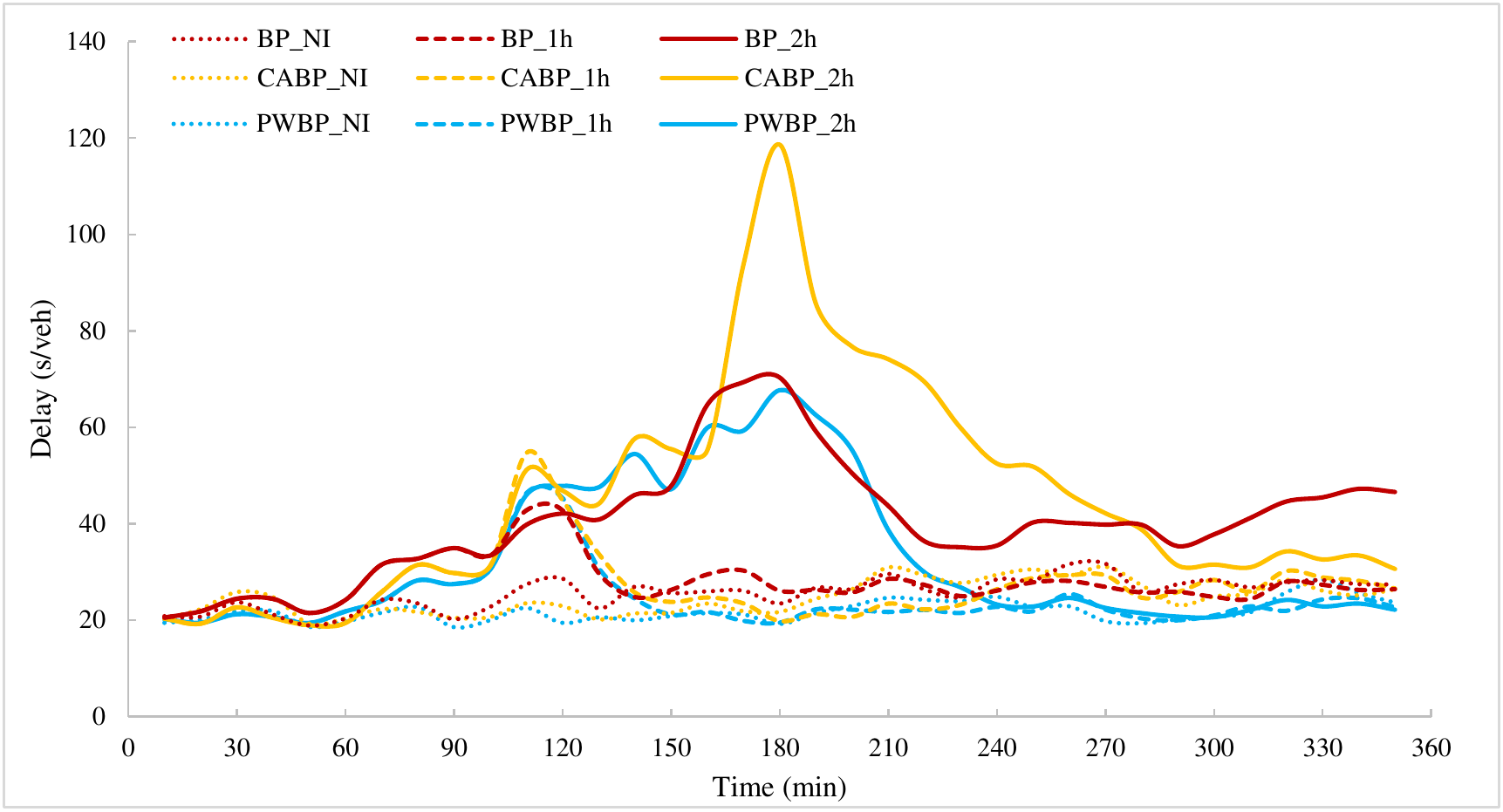}}
	
	\caption{Delays associated with different policies with one lane blocked by the incident under a demand level of 1500 veh/h.} 
	\label{F:incident1lane}
\end{figure}
Fixed timing is not included here since 1500 veh/h is beyond its capacity region and the delays will only increase without bound.  Dotted lines represent the non-incident cases, while dashed and solid lines represent the incident cases with one and two hour durations, respectively. The incident starts at the 60th min in both cases.  When the incident duration is one hour, we see that the network recovers within 30 minutes after the incident is cleared under BP, CABP and PWBP.  However, when the incident duration is two hours, PWBP only needs one hour to completely recover, while congestion in the network persists for significantly longer under BP and CABP: the effects of the incident are still felt in the network three hours after the incident is cleared (compared to the no-incident scenarios).

\autoref{F:incident2lane} shows the two-lanes-blocked cases when demand is 1200 veh/h. 
\begin{figure}[h!]
	\centering
	\resizebox{0.65\textwidth}{!}{%
		\includegraphics{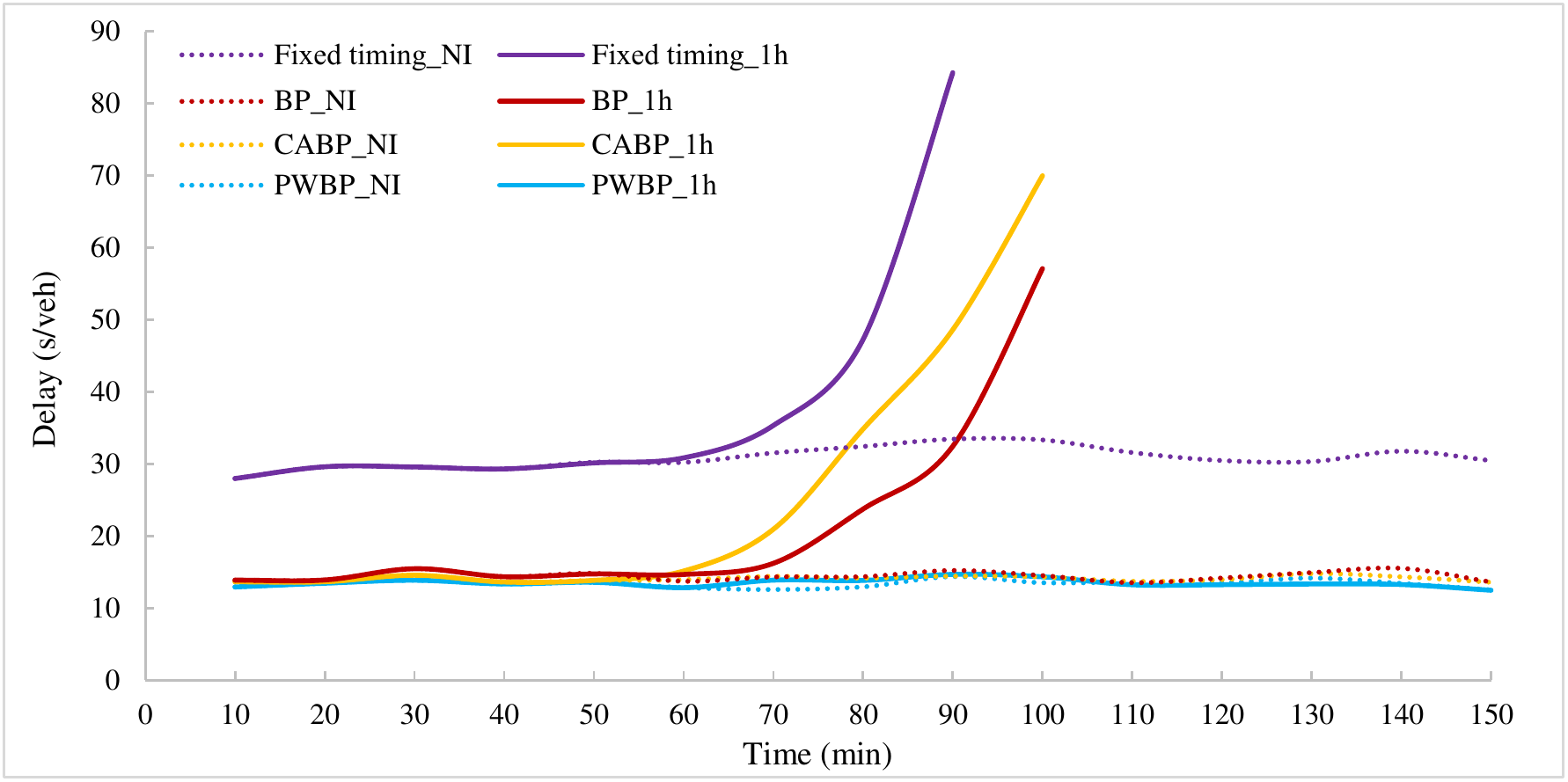}}
	
	\caption{Delays associated with different policies with two lanes blocked by the incident under a demand level of 1200 veh/h.} 
	\label{F:incident2lane}
\end{figure}
The network fails to recover under fixed timing, BP and CABP control when the incident blocks two of the three lanes.  The delays increase sharply and the whole network becomes gridlocked.  In contrast, using PWBP control the incident only has limited impact on network delay and the network recovers quickly from congestion. To further illustrate this, \autoref{F:FT} -- \autoref{F:PWBP} provide snapshots of the spatial distribution of speeds at four time instants after an incident that blocks two lanes (which lasts for 1 hour), one figure for each of the four control schemes.\footnote{Readers can find videos covering the entire simulation at \href{https://youtu.be/KKmwFBbdbT0}{FT}, \href{https://youtu.be/iYyOxJLHITk}{BP}, \href{https://youtu.be/07VAjKcdMD4}{CABP}, and \href{https://youtu.be/QPPn89ICwzw}{PWBP}.}
\begin{figure}[h!]
	\centering	

	\resizebox{1.0\textwidth}{!}{%
		\includegraphics{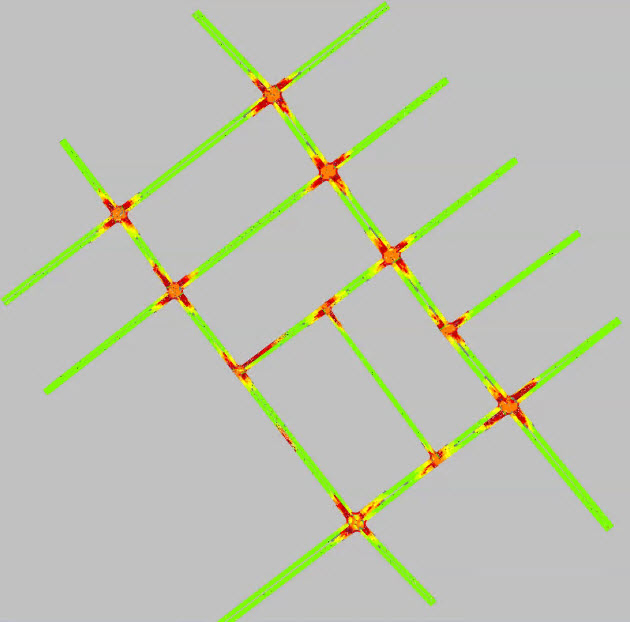} \hspace{0.1in}
		\includegraphics{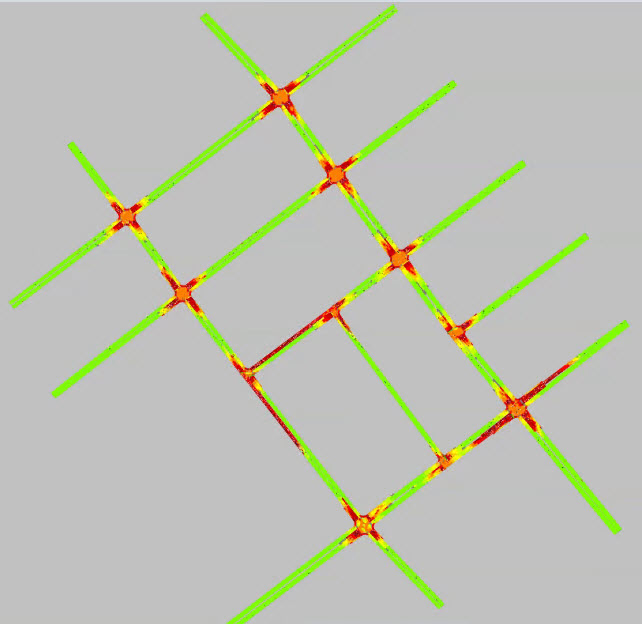}}
	
	{ \small(a) FT@10min \hspace{1.5in} (b) FT@40min}
	
		\resizebox{1.0\textwidth}{!}{%
		\includegraphics{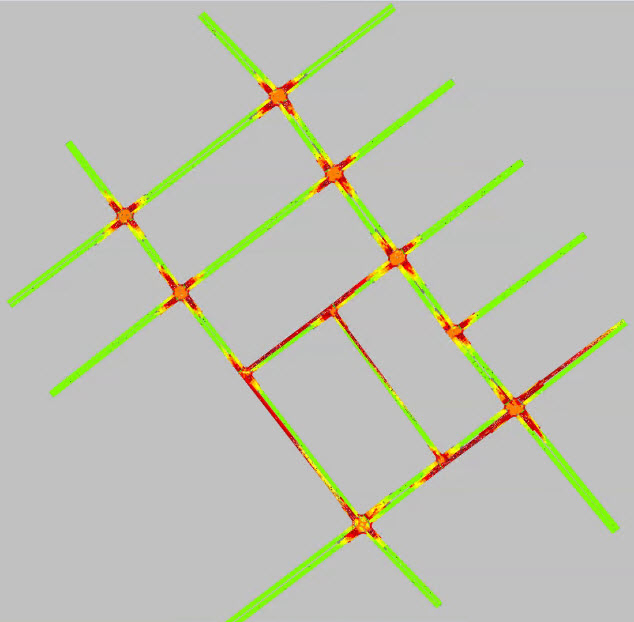} \hspace{0.1in}
		\includegraphics{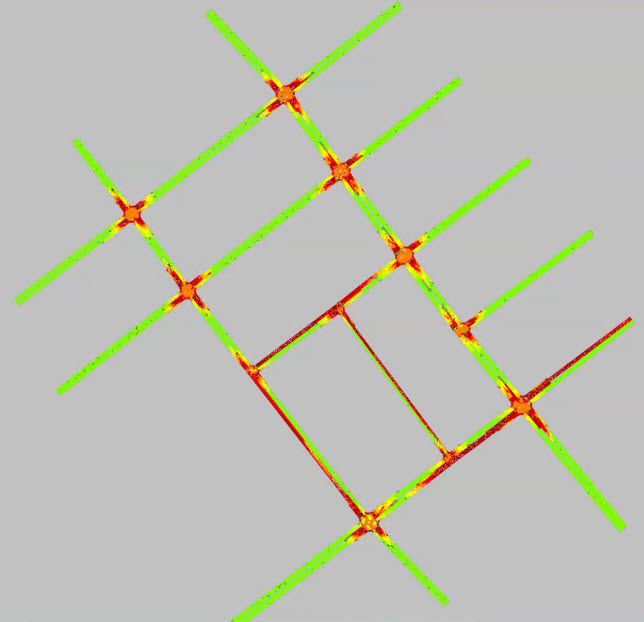}}
	
	{ \small(c) FT@70min \hspace{1.5in} (d) FT@90min}
	
	\caption{Network speed spatial distribution under fixed timing control, (a) 10 minutes after the incident takes place (b) 40 minutes after the incident takes place, (c) 70 minutes after the incident takes place, and (d) 90 minutes after the incident takes place.} 
	\label{F:FT}
\end{figure}
\begin{figure}[h!]
	\centering	

	\resizebox{1.0\textwidth}{!}{%
		\includegraphics{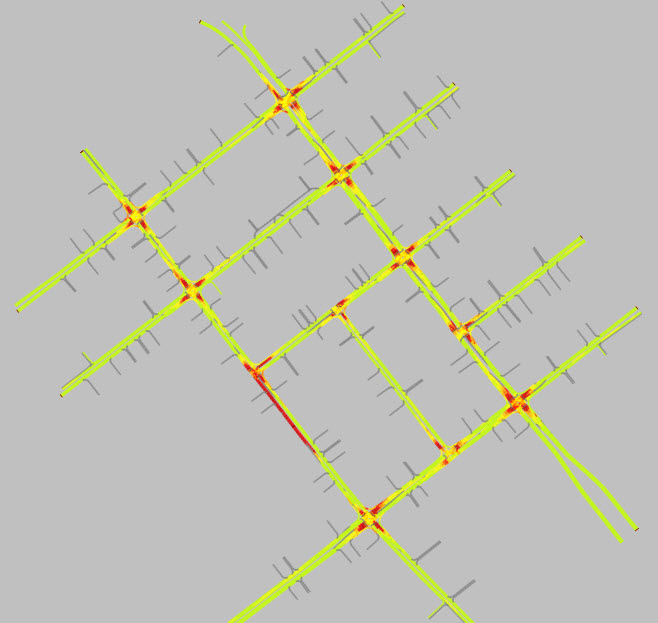} \hspace{0.1in}
		\includegraphics{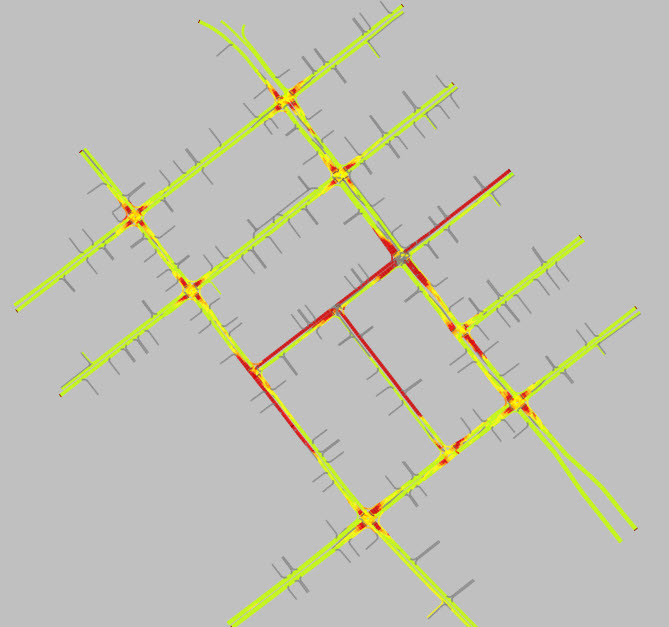}}
	
	{ \small(a) BP@10min \hspace{1.5in} (b) BP@40min}
	
		\resizebox{1.0\textwidth}{!}{%
		\includegraphics{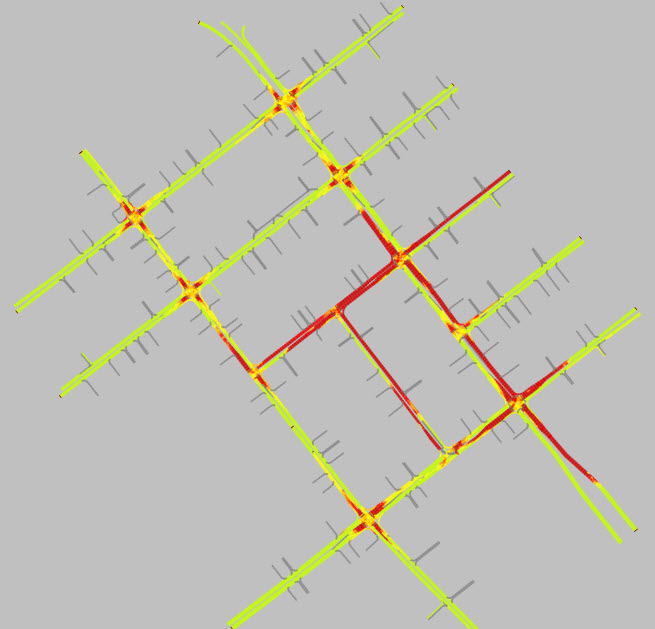} \hspace{0.1in}
		\includegraphics{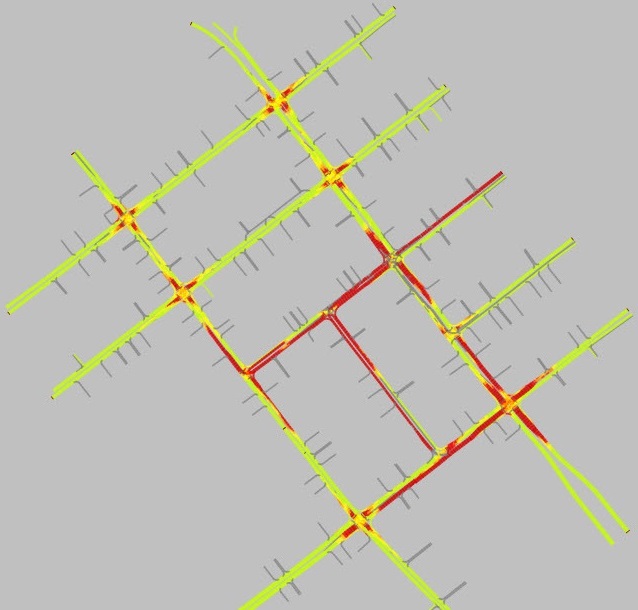}}
	
	{ \small(c) BP@70min \hspace{1.5in} (d) BP@90min}	
	
	\caption{Network speed spatial distribution under BP control, (a) 10 minutes after the incident takes place, (b) 40 minutes after the incident takes place, (c) 70 minutes after the incident takes place, and (d) 90 minutes after the incident takes place.} 
	\label{F:BP}
\end{figure}
\begin{figure}[h!]
	\centering	

	\resizebox{1.0\textwidth}{!}{%
		\includegraphics{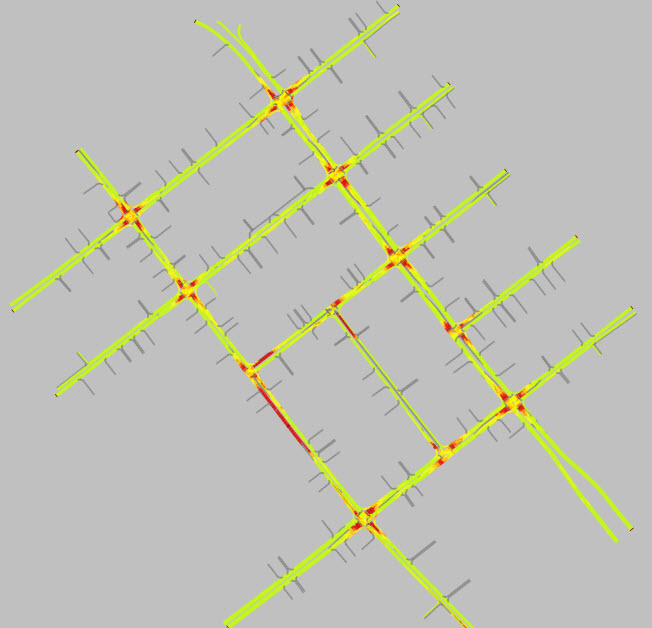} \hspace{0.1in}
		\includegraphics{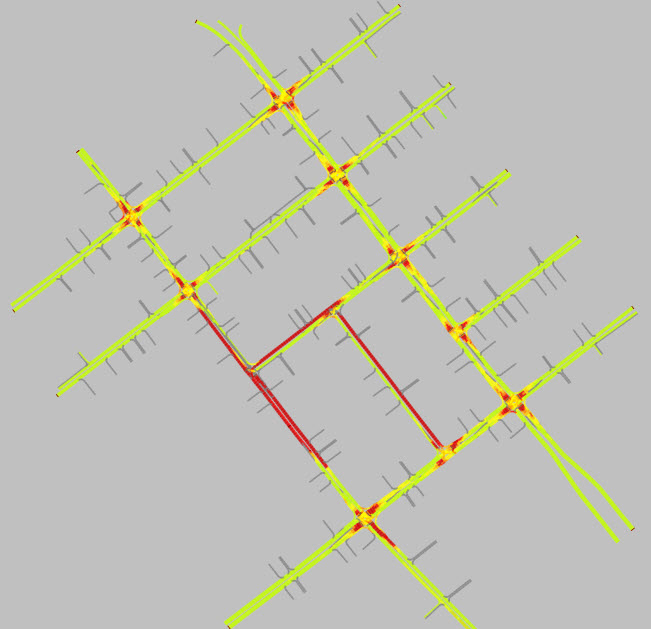}}
	
	{ \small(a) CABP@10min \hspace{1.5in} (b) CABP@40min}
	
		\resizebox{1.0\textwidth}{!}{%
		\includegraphics{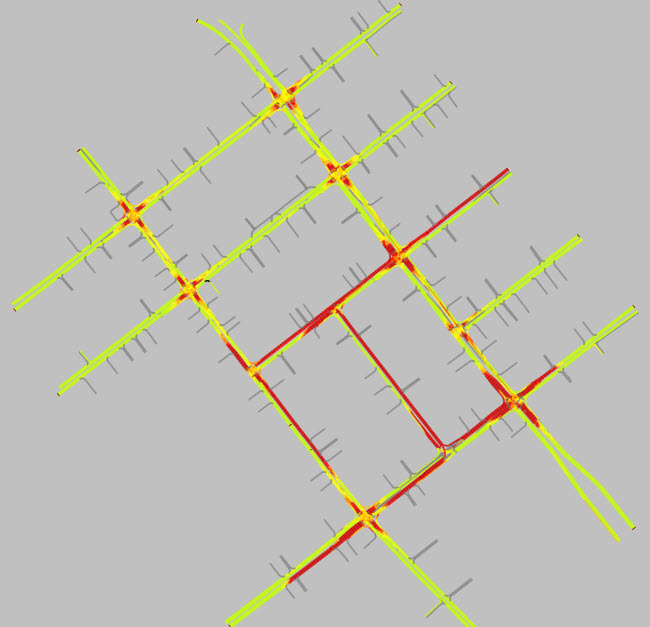} \hspace{0.1in}
		\includegraphics{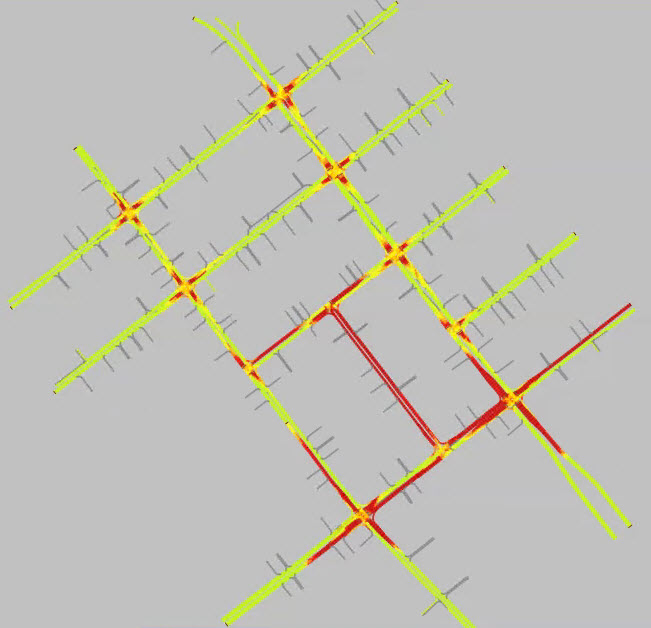}}
	
	{ \small(c) CABP@70min \hspace{1.5in} (d) CABP@90min}	
	
	\caption{Network speed spatial distribution under CABP control, (a) 10 minutes after the incident takes place, (b) 40 minutes after the incident takes place, (c) 70 minutes after the incident takes place, and (d) 90 minutes after the incident takes place.} 
	\label{F:CABP}
\end{figure}
\begin{figure}[h!]
	\centering	

	\resizebox{1.0\textwidth}{!}{%
		\includegraphics{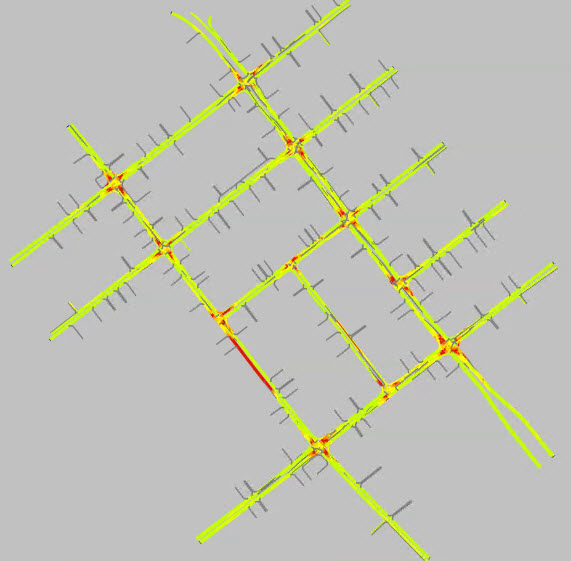} \hspace{0.1in}
		\includegraphics{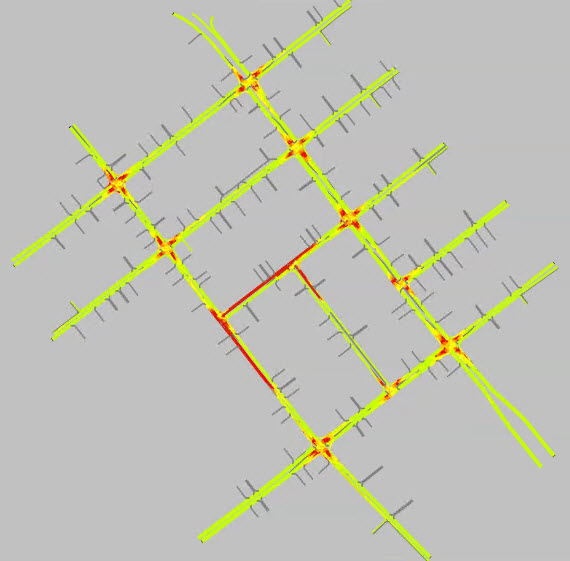}}
	
	{ \small(a) PWBP@10min \hspace{1.5in} (b) PWBP@40min}
	
		\resizebox{1.0\textwidth}{!}{%
		\includegraphics{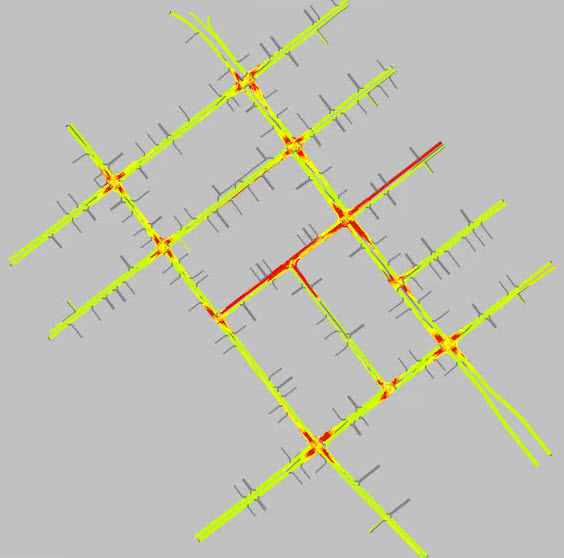} \hspace{0.1in}
		\includegraphics{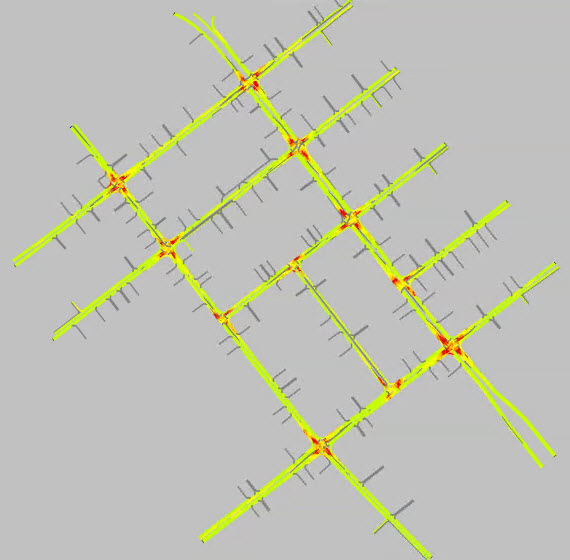}}
	
	{ \small(c) PWBP@70min \hspace{1.5in} (d) PWBP@90min}	
	
	\caption{Network speed spatial distribution under PWBP control, (a) 10 minutes after the incident takes place, (b) 40 minutes after the incident takes place, (c) 70 minutes after the incident takes place, and (d) 90 minutes after the incident takes place.} 
	\label{F:PWBP}
\end{figure}

The reason of the performance difference between BP, CABP and PWBP originates from how the model deals with scenarios in \autoref{F:NWC}b and \autoref{F:NWC}c. With an incident located half-way between intersections A to B, the incident results in congested conditions (queueing) between the incident location and intersection A and low volume traffic between incident location and intersection B. When the queue spills back to intersection A (similar to \autoref{F:NWC}b), PWBP will forbid the movements from A to B, while BP and CABP fail to capture the spillback dynamics.  In addition, PWBP does not allocate green time at intersection B to the movement from A when there are actually no vehicle near the stop line (similar to \autoref{F:NWC}c), while BP and CABP may still allocate green time to this movement.

\section{Conclusion and outlook}
\label{S:Conc}

Backpressure (BP) based intersection control is a control policy that was originally developed for communications networks. Many of the assumptions made in the original theory were adopted in the BP applications to traffic networks despite them not being applicable to vehicular traffic.  Specifically, infinite arc capacities, point queues, independence of commodities (turning movements), and there being no analogue for start-up lost times in communications networks.  These are critical features in intersection control.  To accommodate these features, we develop a backpressure control technique that is based on macroscopic traffic flow, which we refer to as position-weighted backpressure (PWBP).  PWBP considers the spatial distribution of vehicles when calculating the backpressure weights.

The proposed PWBP control policy is tested using a microscopic traffic simulation model of an eleven-intersection network in Abu Dhabi.  Comparisons against coordinated and optimized fixed signal timing, standard BP, and a capacity-aware variant of BP (CABP) were carried out.  The results indicate that PWBP can accommodate higher demand levels than the other three control policies and outperforms them in terms of total network delay, congestion propagation speed,  recoverability from heavy congestion, and response to an incident.

This paper has focused on prioritization of movements at network intersections.  As a possible future research direction, this can be extended to include real-time route guidance. Another possible avenue for future research is a combined perimeter/interior control policy.  Perimeter control \citep{yang2017multi,chiabaut2018perimeter,ortigosa2014study,ambuhl2018case,he2018providing,yang2018perimeter} is emerging as a useful tool for network control at a macroscopic level.  A study of the trade-offs between the capacity region of an intersection control policy and perimeter control could serve as a powerful network-wide control tool. 

\section*{Acknowledgments}
\label{Ack}
This work was funded in part by the C2SMART Center, a Tier 1 USDOT University Transportation Center, and in part by the New York University Abu Dhabi Research Enhancement Fund.  The authors also wish to acknowledge the Abu Dhabi Department of Transportation for their support.	
	
	
	
	
	\bibliographystyle{plain}
	\bibliography{refs_R1}

\begin{thebibliography}{10}

\bibitem{ambuhl2018case}
L.~Amb{\"u}hl, A.~Loder, M.~Menendez, and K.~Axhausen.
\newblock A case study of {Z}urich’s two-layered perimeter control.
\newblock In {\em 7th Transport Research Arena (TRA 2018)}, 2018.

\bibitem{cervero1986unlocking}
R.~Cervero.
\newblock Unlocking suburban gridlock.
\newblock {\em Journal of the American Planning Association}, 52(4):389--406,
  1986.

\bibitem{chiabaut2018perimeter}
N.~Chiabaut, M.~M.~K{\"u}ng, M.~Menendez, and L.~Leclercq.
\newblock Perimeter control as an alternative to dedicated bus lanes: {A} case
  study.
\newblock {\em Transportation Research Record: Journal of the Transportation
  Research Board}, 2672(20):110--120, 2018.

\bibitem{de2011traffic}
J.~De~Gier, T.~Garoni, and O.~Rojas.
\newblock Traffic flow on realistic road networks with adaptive traffic lights.
\newblock {\em Journal of Statistical Mechanics: Theory and Experiment},
  2011(04):P04008, 2011.

\bibitem{feng2018spatiotemporal}
Y.~Feng, C.~Yu, and H.~Liu.
\newblock Spatiotemporal intersection control in a connected and automated
  vehicle environment.
\newblock {\em Transportation Research Part C}, 89:364--383, 2018.

\bibitem{gartner1983opac}
N.~Gartner.
\newblock {OPAC}: {A} demand-responsive strategy for traffic signal control.
\newblock {\em Transportation Research Record}, 906:75--81, 1983.

\bibitem{gayah2014impacts}
V.~Gayah, X.~Gao, and A.~Nagle.
\newblock On the impacts of locally adaptive signal control on urban network
  stability and the macroscopic fundamental diagram.
\newblock {\em Transportation Research Part B}, 70:255--268, 2014.

\bibitem{georgiadis2006resource}
L.~Georgiadis, M.~Neely, and L.~Tassiulas.
\newblock {\em Resource allocation and cross-layer control in wireless
  networks}.
\newblock Now Publishers, Hanover, MA, 2006.

\bibitem{gettman2007data}
D.~Gettman, S.~Shelby, L.~Head, D.~Bullock, and N.~Soyke.
\newblock Data-driven algorithms for real-time adaptive tuning of offsets in
  coordinated traffic signal systems.
\newblock {\em Transportation Research Record}, 2035:1--9, 2007.

\bibitem{gregoire2015capacity}
J.~Gregoire, X.~Qian, E.~Frazzoli, A.~De~La~Fortelle, and T.~Wongpiromsarn.
\newblock Capacity-aware backpressure traffic signal control.
\newblock {\em IEEE Transactions on Control of Network Systems}, 2(2):164--173,
  2015.

\bibitem{guo2019urban}
G~Guo, L.~Li, and X.~Ban.
\newblock Urban traffic signal control with connected and automated vehicles:
  {A} survey.
\newblock {\em Transportation Research Part C}, (To appear), 2019.

\bibitem{he2018providing}
H.~He, K.~Yang, H.~Liang, M.~Menendez, and S.~Guler.
\newblock Providing public transport priority at urban network perimeters: {A}
  bi-modal perimeter control approach.
\newblock In {\em 97th Annual Meeting of the Transportation Research Board},
  volume No. 18-04772, 2018.

\bibitem{heung2005coordinated}
T.~Heung, T.~Ho, and Y.~Fung.
\newblock Coordinated road-junction traffic control by dynamic programming.
\newblock {\em IEEE Transactions on Intelligent Transportation Systems},
  6(3):341--350, 2005.

\bibitem{jabari2012stochasticDiss}
S.E. Jabari.
\newblock {\em A stochastic model of macroscopic traffic flow: {T}heoretical
  foundations}.
\newblock Doctoral dissertation. University of Minnesota Twin Cities,
  Minneapolis MN, 2012.

\bibitem{jabari2016node}
S.E. Jabari.
\newblock Node modeling for congested urban road networks.
\newblock {\em Transportation Research Part B}, 91:229--249, 2016.

\bibitem{jabari2012stochastic}
S.E. Jabari and H.~Liu.
\newblock A stochastic model of traffic flow: {T}heoretical foundations.
\newblock {\em Transportation Research Part B}, 46(1):156--174, 2012.

\bibitem{jabari2013stochastic}
S.E. Jabari and H.~Liu.
\newblock A stochastic model of traffic flow: {G}aussian approximation and
  estimation.
\newblock {\em Transportation Research Part B}, 47:15--41, 2013.

\bibitem{jabari2018stochastic}
S.E. Jabari, F.~Zheng, H.~Liu, and M.~Filipovska.
\newblock Stochastic {L}agrangian modeling of traffic dynamics.
\newblock In {\em 97th Annual Meeting of the Transportation Research Board},
  volume No. 18-04170, 2018.

\bibitem{jabari2014probabilistic}
S.E. Jabari, J.~Zheng, and H.~Liu.
\newblock A probabilistic stationary speed–-density relation based on
  {N}ewell’s simplified car-following model.
\newblock {\em Transportation Research Part B}, 68:205--223, 2014.

\bibitem{lammer2008self}
S.~L{\"a}mmer and D.~Helbing.
\newblock Self-control of traffic lights and vehicle flows in urban road
  networks.
\newblock {\em Journal of Statistical Mechanics: Theory and Experiment},
  2008(04):P04019, 2008.

\bibitem{lammer2010self}
S.~L{\"a}mmer and D.~Helbing.
\newblock {\em Self-stabilizing decentralized signal control of realistic,
  saturated network traffic (Technical Report No. 10-09-019)}.
\newblock Santa Fe Institute, 2010.

\bibitem{le2015decentralized}
T.~Le, P.~Kov{\'a}cs, N.~Walton, H.~Vu, L.~Andrew, and S.~Hoogendoorn.
\newblock Decentralized signal control for urban road networks.
\newblock {\em Transportation Research Part C}, 58:431--450, 2015.

\bibitem{li2017recasting}
P.~Li and X.~Zhou.
\newblock Recasting and optimizing intersection automation as a
  connected-and-automated-vehicle ({CAV}) scheduling problem: {A} sequential
  branch-and-bound search approach in phase-time-traffic hypernetwork.
\newblock {\em Transportation Research Part B}, 105:479--506, 2017.

\bibitem{ma2013coordinated}
W.~Ma, H.~Xie, Y.~Liu, L.~Head, and Z.~Luo.
\newblock Coordinated optimization of signal timings for intersection approach
  with presignals.
\newblock {\em Transportation Research Record: Journal of the Transportation
  Research Board}, 2355:93--104, 2013.

\bibitem{mirchandani2001real}
P.~Mirchandani and L.~Head.
\newblock A real-time traffic signal control system: {A}rchitecture,
  algorithms, and analysis.
\newblock {\em Transportation Research Part C}, 9(6):415--432, 2001.

\bibitem{neely2010stochastic}
M.~Neely.
\newblock Stochastic network optimization with application to communication and
  queueing systems.
\newblock {\em Synthesis Lectures on Communication Networks}, 3(1):1--211,
  2010.

\bibitem{neely2005dynamic}
M.~Neely, E.~Modiano, and C.~Rohrs.
\newblock Dynamic power allocation and routing for time-varying wireless
  networks.
\newblock {\em IEEE Journal on Selected Areas in Communications},
  23(1):89--103, 2005.

\bibitem{newell2002simplified}
G.~Newell.
\newblock A simplified car-following theory: {A} lower order model.
\newblock {\em Transportation Research Part B}, 36(3):195--205, 2002.

\bibitem{ortigosa2014study}
J.~Ortigosa, M.~M.~Menendez, and H.~Tapia.
\newblock Study on the number and location of measurement points for an {MFD}
  perimeter control scheme: {A} case study of {Z}urich.
\newblock {\em EURO Journal on Transportation and Logistics}, 3(3-4):245--266,
  2014.

\bibitem{papageorgiou2003review}
M.~Papageorgiou, C.~Diakaki, V.~Dinopoulou, A.~Kotsialos, and Y.~Wang.
\newblock Review of road traffic control strategies.
\newblock {\em Proceedings of the IEEE}, 91(12):2043--2067, 2003.

\bibitem{rodriguez2019location}
M.~Rodriguez-Vega, C.~Canudas-De-Wit, and H.~Fourati.
\newblock Location of turning ratio and flow sensors for flow reconstruction in
  large traffic networks.
\newblock {\em Transportation Research Part B}, 121:21--40, 2019.

\bibitem{seo2017traffic}
T.~Seo, A.~Bayen, T.~Kusakabe, and Y.~Asakura.
\newblock Traffic state estimation on highway: {A} comprehensive survey.
\newblock {\em Annual Reviews in Control}, 43:128--151, 2017.

\bibitem{smith1980local}
M.~Smith.
\newblock A local traffic control policy which automatically maximises the
  overall travel capacity of an urban road network.
\newblock {\em Traffic Engineering \& Control}, 21(HS-030 129), 1980.

\bibitem{smith2011dynamics}
M.~Smith.
\newblock Dynamics of route choice and signal control in capacitated networks.
\newblock {\em Journal of Choice Modelling}, 4(3):30--51, 2011.

\bibitem{tassiulas1992stability}
L.~Tassiulas and A.~Ephremides.
\newblock Stability properties of constrained queueing systems and scheduling
  policies for maximum throughput in multihop radio networks.
\newblock {\em IEEE Transactions on Automatic Control}, 37(12):1936--1948,
  1992.

\bibitem{tettamanti2010distributed}
T.~Tettamanti and I.~Varga.
\newblock Distributed traffic control system based on model predictive control.
\newblock {\em Periodica Polytechnica Civil Engineering}, 54(1):3--9, 2010.

\bibitem{treiber2013traffic}
M.~Treiber and A.~Kesting.
\newblock {\em Traffic Flow Dynamics: {D}ata, Models, and Simulation}.
\newblock Springer--Verlag, Berlin, 2013.

\bibitem{van2018macroscopic}
P.~van Erp, V.~Knoop, and S.~Hoogendoorn.
\newblock Macroscopic traffic state estimation using relative flows from
  stationary and moving observers.
\newblock {\em Transportation Research Part B}, 114:281--299, 2018.

\bibitem{varaiya2013max}
P.~Varaiya.
\newblock Max pressure control of a network of signalized intersections.
\newblock {\em Transportation Research Part C}, 36:177--195, 2013.

\bibitem{wongpiromsarn2012distributed}
T.~Wongpiromsarn, T.~Uthaicharoenpong, Y.~Wang, E.~Frazzoli, and D.~Wang.
\newblock Distributed traffic signal control for maximum network throughput.
\newblock In {\em 15th International IEEE Conference on Intelligent
  Transportation Systems (ITSC)}, pages 588--595, 2012.

\bibitem{xiao2014pressure}
N.~Xiao, E.~Frazzoli, Y.~Li, Y.~Wang, and D.~Wang.
\newblock Pressure releasing policy in traffic signal control with finite queue
  capacities.
\newblock In {\em IEEE 53rd Annual Conference on Decision and Control (CDC)},
  pages 6492--6497, 2014.

\bibitem{yang2017multi}
K.~Yang, N.~Zheng, and M.~Menendez.
\newblock Multi-scale perimeter control approach in a connected-vehicle
  environment.
\newblock {\em Transportation Research Part C}, 94:32--49, 2017.

\bibitem{yang2018perimeter}
K.~Yang, N.~Zheng, and M.~Menendez.
\newblock A perimeter control approach integrating dedicated express toll
  lanes.
\newblock In {\em 97th Annual Meeting of the Transportation Research Board},
  volume No. 18-02838, 2018.

\bibitem{dujardin2011multiobjective}
D.~Yann, F.~Boillot, D.~Vanderpooten, and P.~Vinant.
\newblock Multiobjective and multimodal adaptive traffic light control on
  single junctions.
\newblock In {\em 14th International IEEE Conference on Intelligent
  Transportation Systems (ITSC)}, pages 1361--1368, 2011.

\bibitem{you2013coordinated}
X.~You, L.~Li, and W.~Ma.
\newblock Coordinated optimization model for signal timings of full continuous
  flow intersections.
\newblock {\em Transportation Research Record: Journal of the Transportation
  Research Board}, 2356:23--33, 2013.

\bibitem{yu2018integrated}
C.~Yu, Y.~Feng, H.~Liu, W.~Ma, and X.~Yang.
\newblock Integrated optimization of traffic signals and vehicle trajectories
  at isolated urban intersections.
\newblock {\em Transportation Research Part B}, 112:89--112, 2018.

\bibitem{zheng2018traffic}
F.~Zheng, S.E. Jabari, H.~Liu, and D.~Lin.
\newblock Traffic state estimation using stochastic {L}agrangian dynamics.
\newblock {\em Transportation Research Part B}, 115:143--165, 2018.

\bibitem{zheng2017estimating}
J.~Zheng and H.~Liu.
\newblock Estimating traffic volumes for signalized intersections using
  connected vehicle data.
\newblock {\em Transportation Research Part C}, 79:347--362, 2017.

\end{thebibliography}
	
	
	
	
	
	

\end{document}